\documentclass[journal]{IEEEtran}

\usepackage{amsmath}
\usepackage{amssymb}
\usepackage{amsthm}
\usepackage{multirow}
\usepackage{booktabs}
\usepackage{enumitem}
\usepackage{rotating}
\usepackage{xcolor}

\usepackage{tikz}
\usetikzlibrary{arrows, decorations.markings, decorations.pathreplacing}

\tikzstyle{vecArrow} = [thick, decoration={markings,mark=at position
   1 with {\arrow[semithick]{open triangle 60}}},
   double distance=2pt, shorten >= 5.5pt,
   preaction = {decorate},
   postaction = {draw,line width=1.4pt, white,shorten >= 4.5pt}]
\tikzstyle{innerWhite} = [semithick, white,line width=1.4pt, shorten >= 4.5pt]

\usepackage{agda}
\usepackage[T1]{fontenc}
\usepackage{newunicodechar}

\newunicodechar{𝕋}{\ensuremath{\mathbb{T}}}
\newunicodechar{ℕ}{\ensuremath{\mathbb{N}}}
\newunicodechar{ℓ}{\ensuremath{\ell}}
\newunicodechar{∞}{\ensuremath{\infty}}

\newunicodechar{⊔}{\ensuremath{\sqcup}}
\newunicodechar{∸}{\ensuremath{\dotdiv}}
\newunicodechar{∥}{\ensuremath{\parallel}}

\newunicodechar{≤}{\ensuremath{\leq}}
\newunicodechar{≱}{\ensuremath{\ngeq}}
\newunicodechar{▷}{\ensuremath{\vartriangleright}}
\newunicodechar{⇿}{\ensuremath{\leftrightarrow}} % Poor match
\newunicodechar{∷}{\ensuremath{::}}

\newunicodechar{⌊}{\ensuremath{\lfloor}}
\newunicodechar{⌋}{\ensuremath{\rfloor}}
\newunicodechar{∨}{\ensuremath{\vee}}

\newunicodechar{≡}{\ensuremath{\equiv}}
\newunicodechar{≢}{\ensuremath{\nequiv}}
\newunicodechar{≈}{\ensuremath{\approx}}
\newunicodechar{≉}{\ensuremath{\napprox}}
\newunicodechar{≟}{\ensuremath{\stackrel{?}{=}}}

\newunicodechar{∀}{\ensuremath{\forall}}
\newunicodechar{∃}{\ensuremath{\exists}}
\newunicodechar{⇒}{\ensuremath{\Rightarrow}}
\newunicodechar{∧}{\ensuremath{\wedge}}
\newunicodechar{∈}{\ensuremath{\in}}
\newunicodechar{∉}{\ensuremath{\notin}}
\newunicodechar{∴}{\ensuremath{\therefore}}
\newunicodechar{∎}{\ensuremath{\qed}}
\newunicodechar{∘}{\ensuremath{\circ}}
\newunicodechar{⊕}{\ensuremath{\oplus}}

\newunicodechar{δ}{\ensuremath{\delta}}
\newunicodechar{ψ}{\ensuremath{\Psi}}

\newunicodechar{ₐ}{\ensuremath{_a}}
\newunicodechar{ₑ}{\ensuremath{_e}}
\newunicodechar{ᵢ}{\ensuremath{_i}}
\newunicodechar{ₖ}{\ensuremath{_k}}
\newunicodechar{ₗ}{\ensuremath{_l}}
\newunicodechar{ₘ}{\ensuremath{_m}}
\newunicodechar{ₙ}{\ensuremath{_n}}
\newunicodechar{ₓ}{\ensuremath{_x}}

\newunicodechar{₀}{\ensuremath{_0}}
\newunicodechar{₁}{\ensuremath{_1}}
\newunicodechar{₂}{\ensuremath{_2}}
\newunicodechar{₃}{\ensuremath{_3}}
\newunicodechar{₊}{\ensuremath{_+}}

\newunicodechar{ᵉ}{\ensuremath{^e}}
\newunicodechar{ᵐ}{\ensuremath{^m}}
\newunicodechar{ʳ}{\ensuremath{^r}}

\newunicodechar{∅}{\ensuremath{\emptyset}}
\newunicodechar{⊥}{\ensuremath{\bot}}

\newcommand{\agdaIndent}{\AgdaSpace{}\AgdaSpace{}\AgdaSpace{}\AgdaSpace{}}

\usepackage{ifxetex}

\newtheorem{definition}{Definition}
\newtheorem{theorem}{Theorem}
\newtheorem{lemma}{Lemma}

\DeclareMathOperator{\best}{best}

\newcommand{\NN}{\mathbb{N}}
\newcommand{\NNinf}{\NN_\infty}
\newcommand{\RR}{\mathbb{R}}

\newcommand{\tble}[1]{\ensuremath{\overline{#1}}}
\newcommand{\xt}{\tble{x}}
\newcommand{\yt}{\tble{y}}

\newcommand{\mtrx}[1]{\ensuremath{\mathbf{#1}}}

\newcommand{\A}{\mtrx{A}}
\newcommand{\I}{\mtrx{I}}
\newcommand{\X}{\mtrx{X}}
\newcommand{\Y}{\mtrx{Y}}

\newcommand{\size}[1]{\left\vert #1 \right\vert}

\newcommand{\allNodes}{\mathbb{V}}
\newcommand{\nodes}{V}
\newcommand{\allEdges}{\mathbb{E}}
\newcommand{\assignments}{A}
\newcommand{\epochs}{\mathbb{E}}
\newcommand{\timess}{\mathbb{T}}
\newcommand{\network}{N}

\newcommand{\aFun}{F}
\newcommand{\sStateFun}{\sigma}
\newcommand{\aStateFun}{\delta}
\newcommand{\localView}{\mtrx{L}}

\newcommand{\configs}{\mathcal{C}}
\newcommand{\freeConfigs}{\configs^{free}}

\newcommand{\carrier}{S}
\newcommand{\tcarrier}{S^n}
\newcommand{\mcarrier}{S^{n \times n}}
\newcommand{\choice}{\oplus}
\newcommand{\extension}{E}
\newcommand{\allExtensions}{\extension^{*}}
\newcommand{\invalid}{\ensuremath{\overline{\infty}}}
\newcommand{\trivial}{\ensuremath{\overline{0}}}
\newcommand{\invalidedge}{f^{\invalid}}
\newcommand{\accSet}{\mcarrier_p}

\newcommand{\paths}{\mathbb{P}}
\newcommand{\trivialpath}{[\:]}
\newcommand{\invalidpath}{\bot}
\newcommand{\weightf}[1]{weight_{#1}}
\newcommand{\pathf}{path}
\newcommand{\pathsaligned}{\leftrightsquigarrow}
\newcommand{\cons}{::}

\newcommand{\genericRoutingAlgebra}{(\carrier,\ \oplus,\ \extension,\ \trivial,\ \invalid, \ \invalidedge)}
\newcommand{\genericPathAlgebra}{(\carrier,\ \choice,\ \extension,\ \trivial,\ \invalid, \ \invalidedge,\ \pathf)}

\newcommand{\sigX}{\aFun^{ep}(\X)}
\newcommand{\sigsqX}{(\aFun^{ep})^2(\X)}
\newcommand{\sigY}{\aFun^{ep}(\Y)}

\newcommand{\consistent}[1]{C_{#1}}
\newcommand{\tconsistent}[1]{\consistent{\A}^{n}}
\newcommand{\mconsistent}[1]{\consistent{\A}^{n \times n}}
\newcommand{\lconsistent}{C^{ep}}

\makeatletter
\providecommand{\leftsquigarrow}{%
  \mathrel{\mathpalette\reflect@squig\relax}%
}
\newcommand{\reflect@squig}[2]{%
  \reflectbox{$\m@th#1\rightsquigarrow$}%
}
\makeatother

\newcommand{\lessWeight}{<_\choice}
\newcommand{\leqWeight}{\leq_\choice}
\newcommand{\lessAss}{\prec_\choice}
\newcommand{\leqAss}{\curlyeqprec_\choice}
\newcommand{\extendedBy}[1]{\rightsquigarrow_{#1}}
\newcommand{\extendedByU}{\rightsquigarrow^{ep}}
\newcommand{\threatenedBy}[1]{\trianglelefteq_{#1}}

\newcommand{\proofOrder}{<^{ep}}
\newcommand{\rheight}{h^{ep}}
\newcommand{\rheightMax}{H}
\newcommand{\rmetric}[1]{r^{ep}_{#1}}
\newcommand{\tmetric}[1]{d^{ep}_{#1}}
\newcommand{\smetric}{D^{ep}}

\newcommand{\rheightI}{h^{ep}_I}
\newcommand{\rheightC}[1]{h^{ep}_{C#1}}
\newcommand{\rheightIMax}{H_I}
\newcommand{\rheightCMax}{H_C}
\newcommand{\rmetricI}{r^{ep}_I}
\newcommand{\rmetricC}[1]{r^{ep}_{C#1}}

\theoremstyle{definition}
\newtheorem{example}{Example}

\newcommand{\spacing}{\:\:\:\:}

\newcommand{\postcasesep}{\vspace{0.5em}}
\newenvironment{case}[2]{\noindent \underline{Case #1}: #2 \postcasesep \\} {}

\newcommand{\matthewrevision}[1]{#1}

\newcommand{\timrevision}[1]{#1}

\begin{document}

\title{Formally Verified Convergence of Policy-Rich DBF Routing Protocols}

%\affiliation{University of Cambridge}
%\email{mld46@cam.ac.uk}
%\affiliation{University of Cambridge}
%\email{tgg22@cam.ac.uk}

\author{
\IEEEauthorblockN{Matthew L. Daggitt}, 
\IEEEauthorblockN{Timothy G. Griffin}}

\maketitle

\begin{abstract}
In this paper we present new general convergence results about the behaviour of the Distributed Bellman-Ford (DBF)
family of routing protocols, which includes distance-vector protocols (e.g. RIP) and path-vector protocols (e.g. BGP).
\matthewrevision{Our results apply to ``policy-rich" protocols, by which we mean protocols that can have complex policies (e.g. conditional route transformations) that violate traditional assumptions made in the standard presentation of Bellman-Ford protocols. }

First, we propose a new algebraic model for abstract routing problems which has fewer primitives than previous models and can represent more expressive policy languages. The new model is also the first to allow concurrent reasoning about distance-vector and path-vector protocols.

Second, we explicitly demonstrate how DBF routing protocols are instances of a larger class of asynchronous iterative algorithms, for which there already exist powerful results about convergence. These results allow us to build upon conditions previously shown by Sobrinho to be sufficient and necessary for the convergence of path-vector protocols and generalise and strengthen them in various ways: we show that, with a minor modification, they also apply to distance-vector protocols; we prove they guarantee that the final routing solution reached is unique, thereby eliminating the possibility of anomalies such as BGP wedgies; we relax the model of asynchronous communication, showing that the results still hold if routing messages can be lost, reordered, and duplicated.

Thirdly, our model and our accompanying theoretical results have been fully formalised in the Agda theorem prover. The resulting library is a powerful tool for quickly prototyping and formally verifying new policy languages. As an example, we formally verify the correctness of a policy language with many of the features of BGP including communities, conditional policy, path-inflation and route filtering.
\end{abstract}


\section{Introduction}
\label{sec:introduction}

\subsection{What is policy-rich routing?}
\label{sec:what-is-policy-rich-routing}

This paper proves new, general results about the convergence of asynchronous routing protocols in the Distributed Bellman-Ford (DBF) family, which includes distance-vector (RIP-like) and path-vector (BGP-like) protocols.
In particular, the results apply to what we call \textit{policy-rich} protocols, and so
we begin by informally explaining this terminology.

Suppose that a router participating in
a DBF computation has a best route $r$ to a destination $d$. 
If it then receives a route $r'$ for $d$ from an immediate neighbour
it will first apply some policy $f$ associated
with that neighbour and produce 
a candidate route $f(r')$.
It will then compare $r$ with $f(r')$ to determine which is best.
Let us denote the outcome of this selection process as $\best(f(r'),\ r)$. 
A very simple example is the shortest-paths problem where
$r$ is just an integer representing distance, 
$f_w(r) = w + r$ for some weight $w$, and $\best(r,\ r') = \min(r,\ r')$. 

If we dig deeply into the classical theory underlying
best-path algorithms 
--- both synchronous~\cite{gondran2008graphs,baras2010path},
    and asynchronous~\cite{bertsekas92data} 
--- we find that it always assumes the following equation, or something equivalent, must hold: 
\begin{equation}
\label{eq:distributivity} 
   f(\best(r_1,\ r_2)) = \best(f(r_1),\ f(r_2))
\end{equation}
This property is referred to as \textit{distributivity}.
The left-hand side of the equation
can be interpreted as a decision made by a router sending
routes $r_1$ and $r_2$ while the right-hand side is
a decision made by the neighbour receiving those routes.
Assuming the equality holds, the classical theory proves that routing protocols
arrive at \emph{globally optimal routes} --- where every pair of routers use the best path between them. 

By a \textit{policy-rich language} we mean one in which distributivity does not hold. 
A clear example of how
distributivity violations might arise in routing 
can be seen in the use of \textit{route maps} which are functions (scripts)
that take routes as input and return routes as output. 
For example, if $g$ and $h$ are route maps, then we might
define another route map $f$ as:
\begin{equation}
\label{eq:conditional:policy} 
   f(r) = \text{if\ } P(r) \text{\ then\ } g(r) \text{\ else\ } h(r), 
\end{equation}
where $P$ is a predicate on routes (such as ``does this route contain the
BGP community 17?''). 
To see how easily distributivity can be violated,
suppose that:
\[
\begin{array}{lcl} 
P(a) & = & \text{true},\\
P(b) & = & \text{false},\\ 
\best(a,\ b)    & = & a,\\
\best(g(a),\ h(b)) & = & h(b).
\end{array}
\] 
Then the left-hand side of Equation~\ref{eq:distributivity} is:
\[
   f(\best(a,\ b)) = f(a) = g(a), 
\]
while the right-hand side becomes: 
\[
   \best(f(a),\ f(b)) = \best(g(a),\ h(b)) = h(b)
\]
For Eq~\ref{eq:distributivity} to hold we need 
$g(a) = h(b)$, which may not be the case
(indeed, if $g(a) = h(b)$ were always true, then there
would be no point in defining $f$!). 
Perhaps the most common example of such conditional policies is \emph{route filtering}, where $h(r)$ is equal to the invalid route.

A specialist schooled in the classical theory of path
finding might be tempted to forbid the use of
such ``broken" policies when using the Bellman-Ford or Dijkstra algorithms.
Yet, once again, practice has outstripped theory. 
The Border Gateway Protocol (BGP)~\cite{rfc4271} --- a key
component of the Internet's infrastructure ---
is a policy-rich routing protocol.

Hence a pertinent question is how to \emph{tame} the policy language of BGP-like protocols to ensure good behaviour?
Of course we could mandate that all protocols must conform to Equation~\ref{eq:distributivity}. However, we would then not be able to implement
typical inter-domain policies that are based
on commercial relationships~\cite{huston:99a,huston:99b}. In addition, something as simple as shortest-path routing with route filtering would be banned. 

Gao \& Rexford~\cite{gao2001stable} showed that the simple commercial relationships described in \cite{huston:99a, huston:99b}, if universally followed, are enough to guarantee convergence. However their model is BGP-specific and gives us no guidance on how policy-rich protocols should be constrained in general. Furthermore their conditions impose constraints on the topology of the network and therefore they require continuous verification as the topology changes.

Instead an alternative middle ground for generic policy-rich protocols has been achieved for both Dijkstra's algorithm~\cite{sobrinho2010routing}
and the DBF family~\cite{sobrinho2005algebraic}.
Rather than insisting on distributivity, we require that for all routes $r$ and policies $f$ we have:
\begin{equation}
\label{eq:nondecreasing} 
   r = \best(r,\ f(r)). 
\end{equation}
In other words, applying policy to a route cannot produce a route that is more preferred.
Although Equation~\ref{eq:nondecreasing} is sufficient
for Dijkstra's Algorithm~\cite{sobrinho2010routing}, it
must be strengthened
for DBF algorithms~\cite{sobrinho2005algebraic} to:
\begin{equation}
\label{eq:increasing} 
   r = \best(r,\ f(r)) \:\:\: \wedge \:\: r \not= f(r). 
\end{equation}
That is, applying policy to a route cannot produce a route that is more preferred
\emph{and}
it cannot leave a route unchanged.
We call policy languages that obey Equation~\ref{eq:nondecreasing} \emph{increasing} and those that obey Equation~\ref{eq:increasing} \emph{strictly increasing}. 
Note that if policies $g$ and $h$ are strictly increasing, then
the conditional policy $f$ defined in
Equation~\ref{eq:conditional:policy} is also strictly increasing.
In other words, a strictly increasing policy language
remains strictly increasing when route maps are introduced.

However, without distributivity we can no longer achieve \emph{globally optimal} routes and so we must be content with
\emph{locally optimal} routes~\cite{sobrinho2010routing} --
each router obtains the best routes possible given the best routes revealed by its neighbours.

It is natural to ask if the strictly increasing condition is too strong and prohibits desirable policies? Sobrinho~\cite{sobrinho2005algebraic} shows that it is the weakest possible property that guarantees a path-vector protocol always converges. This implies that if the policy language is \emph{not} strictly increasing then there exists a network in which the protocol diverges. 

\timrevision{
As we discuss in Section~\ref{sec:bgp-lite-protocol},
today's BGP is in fact ``broken'' in the sense that 
it is possible to write policies that can result in anomalous behaviour such as
non-convergence \cite{rfc3345,varadhan00oscillations} and multiple stable states~\cite{rfc4264}.
The latter are problematic as the extra stable states are nearly always unintended and often violate the intent of policy writers.
Leaving an unintended stable state requires manual intervention and, in the worst case,
a high degree of coordination between competing networks. This phenomenon is colloquially referred to as a \emph{BGP wedgie}.
Such anomalies arise from the fact that 
the policy language of BGP is not strictly increasing, and there is no known way of making it so without removing the ability of AS's to hide commercially sensitive information from each other. Therefore, the next natural question to ask is: given a protocol with a non-strictly increasing policy language, which networks does the protocol converge over? Sobrinho shows that a network topology being \emph{free} with respect to the policy language is a sufficient and necessary condition for a DBF protocol to converge over that network topology.
}

Subsequent to these foundational results, the algebraic approach has been extended to model other more complex methods of routing such as multi-path routing~\cite{chau2006policy, sobrinho2020routing}, hierarchical routing~\cite{gurney2007lexicographic}, and quality of service routing~\cite{geng2017algebra}.

\subsection{Our contributions}
\label{sec:contributions}

\subsubsection{Answers to big-picture questions}

Given these prior results, the reader might assume that the convergence of standard DBF routing protocols is a solved problem. However, this is far from the case. In this paper, we address the following open questions:

\begin{enumerate}[label=Q\arabic*)]
\item Sobrinho's proofs only guarantee that the protocols converge. Do they always converge to the \emph{same} stable state?

\emph{Answer:} We prove that the answer is yes. From a theoretical perspective, this implies that the set of DBF protocols that converge and the set of DBF protocols that converge deterministically are the same. From a practical perspective, this implies that as long as a DBF protocol is guaranteed to converge then you don't have to worry about problems like BGP wedgies.

\item The proofs of Sobrinho and Gao-Rexford both assume TCP-like in-order, reliable communication between routers. Is this a fundamental requirement for convergence? 

\emph{Answer:} We prove that the answer is no, by showing that the results still hold for a weaker, UDP-like model of communication in which route advertisements may be delayed, lost, reordered and duplicated. From a theoretical perspective,  this provides strong evidence that the difficulty of solving a routing problem in a distributed environment is entirely independent of the communication model used. This is not the case in other distributed protocols, e.g. consensus protocols~\cite{fischer1985impossibility}.

\item Sobrinho reasons about path-vector protocols. Do there exist corresponding results for distance-vector protocols?

\emph{Answer:} we show that they do. In particular, that distance-vector protocols converge over strictly increasing algebras/free networks if the number of possible path-weights assigned to routes is \emph{finite}. From a practical perspective this implies the feasibility of policy-rich distance-vector protocols, a relatively under-explored design space.
\end{enumerate}

\subsubsection{Improvements to the model}
In addition to the new high-level results above, we also make several important contributions to the underlying mathematical model for DBF protocols. A more technical discussion of these contributions can be found in Section~\ref{sec:model-contributions}.

\begin{enumerate}[label=M\arabic*)]
\item We propose a new algebraic structure for modelling policy languages which is both simpler and more expressive than that of Sobrinho. In particular, it contains only~6 algebraic primitives instead of~8, and is capable of naturally modelling path-dependent policies such as route filtering and path-inflation. 

\item We identify that DBF protocols are a single instance of a much broader class of \emph{asynchronous iterative algorithms}, for which there exist prior models and convergence results. From a theoretical point of view, this allows us to cleanly abstract the underlying routing problem being solved from the distributed environment in which it is being solved. Practically, it greatly simplifies the proofs and allows us to prove asynchronous convergence purely by reasoning about the behaviour of the protocol in an entirely synchronous environment. In contrast, previous work~\cite{gao2001stable,sobrinho2005algebraic} has typically created their own custom model of message passing in an asynchronous environment and therefore have to directly reasoning about message timings and have avoided reasoning about unreliable communication.
 
\item Thanks to this existing theory, we construct a mathematical model of an abstract DBF routing protocol that is the first to precisely describe the \emph{entire} evolution of the protocol, including exact message timings, routers and links failing, routers and links being added, and routers being reset. As far as we are aware, this is therefore the first fully provably correct executable model of an abstract DBF protocol. In contrast, previous proofs have focused on proving results about a snapshot of a single network configuration in which all routers, links and policies remain static.
\end{enumerate} 

\subsubsection{Formal verification}

In the past decade there has been a strong trend towards the verification of infrastructure-related software in areas as diverse as operating systems~\cite{Klein:2009}, compilers~\cite{cake:ml:2014,comp:cert:2009} and networking~\cite{foster:pldi:2013,formal:nat:2017}.

In line with this trend we have fully formalised all the results in this paper in the Agda theorem prover~\cite{agda:2009}. As Agda is capable of expressing both proofs and programs, and as our model is executable, this constitutes the first ever fully formalised, executable model of a general DBF protocol as far as we are aware. The advantages of this formalisation effort are as follows:

\begin{enumerate}[label=V\arabic*)]
\item It provides a much stronger guarantee of the correctness of our proofs. Reasoning about algorithms in distributed environments is notoriously tricky, and as Sobrinho themselves say in~\cite{sobrinho2005algebraic}, their proof of convergence is only ``semi-formal''. Having a computer check our proofs is a valuable and powerful assurance as to their correctness.

\item The generality of our verified proofs means that they also function as a library that can be used by others to rapidly prototype, explore the behaviour of and formally verify the correctness of novel policy languages. This is a significant step towards reusable and modular formal verification of routing protocols. The library is freely available~\cite{agda-routing}, and we hope that it will prove useful to the community. 

\item To demonstrate the utility of our library, we construct a model of a path-vector routing protocol that contains many of the features of BGP including local preferences, conditional policy, path inflation and route filtering. We formally verify in Agda that the policy language is strictly increasing, and therefore no matter what policies are used the protocol will always converge deterministically to a single stable state given a sufficient period of stability. Such a protocol is therefore provably safe-by-design.
\end{enumerate}

\subsection{Road map}

In Section~\ref{sec:model}, we use an algebraic approach inspired by that of Sobrinho~\cite{sobrinho2005algebraic} and existing literature on asynchronous iterative algorithms~\cite{daggitt2022dynamic}, to construct an executable model of an abstract, fully asynchronous, distributed DBF routing protocol.
This approach provides a clear, implementation-independent specification of the family of protocols which we can reason about.
In Section~\ref{sec:results} we discuss how our novel use of an existing theory of asynchronous iterative algorithms does much of the heavy lifting, and finish by applying existing results from this field to prove our new results about distance-vector and path-vector protocols.
In Section~\ref{sec:formalisation} we describe our Agda formalisation of these results, and in Section~\ref{sec:bgp-lite-protocol} we construct a safe-by-design algebra using the library. 

\subsection{Previous versions}

Previous versions of this work have appeared before in conference proceedings~\cite{daggitt2018asynchronous} and in the thesis of the first author~\cite{daggitt2019thesis}. The many improvements in this paper over these prior versions are discussed in Section~\ref{sec:previous-versions}.

\section{Model}
\label{sec:model}

In this section we construct a model of an arbitrary DBF routing protocol. For the reader's convenience, a glossary of all the notation used in this and later sections can be found in Appendix~\ref{app:glossary}.

\subsection{Generality vs implementation details}
\label{sec:generality_vs_implementation}

In any mathematical model of a real-world system, there is a fundamental tension between capturing the most general form of the system and capturing every detail of the implementation. More implementation details can always be added to better model a particular system, at the cost of losing the ability to apply the work more generally.

In this paper we err on the side of generality. While we acknowledge that BGP is the dominant policy-rich path-vector protocol, the aim of our work is not only to inform how BGP policies may be constrained to provide better guarantees, but also to guide the design of future policy-rich protocols.

\timrevision{
For example, we do not explicitly model the details of either a soft-state or a hard-state implementation approach
for DBF routing protocols.
RIP~\cite{rfc2453} and EIGRP~\cite{rfc7868} are examples of the soft-state approach: routing information expires after a fixed time and so must be periodically refreshed.
BGP~\cite{rfc4271} is the primary example of the hard-state approach: routing information is preserved at a router until it is explicitly modified by a neighbour's update message or a local link failure. 
We use a very general model of asynchronous computation which we believe can accommodate such implementation choices. 
%(this intuition is not fully formalised in this work).
%Had we explicitly modelled a hard-state protocol, it would prevent us from applying our work to soft-state protocols as well.
}

\subsection{Paths}
\label{sec:paths-model}

Prior to anything else, it is important to discuss exactly what is meant by a path in the context of a distributed, asynchronous routing protocol. Traditionally paths are defined with respect to some graph $G=(V,E)$. However because a network's topology changes over time, the path along which an entry in a router's routing table was generated may not correspond to a valid path in the current topology. Furthermore, and perhaps counter-intuitively, such a path may be formed from an amalgamation of paths from previous topologies, and therefore may not even have existed in any single past topology. Consequently, out of necessity, our notion of a path cannot be tied to the network's topology.

Therefore, we begin by defining $\allNodes$ as the set of all valid router identifiers. An \emph{edge} is an arbitrary pair of router identifiers, and therefore the set of edges is $\allEdges = \allNodes \times \allNodes$.  A \emph{path} is a defined to be a sequence of contiguous edges. A \emph{cycle} is a path which finishes at the same router it started at. A \emph{simple path} is a path that doesn't contain any cycles. \matthewrevision{The trivial path of zero length from a node to itself is denoted as~$\trivialpath$}. For our purposes, we also consider an additional simple path~$\invalidpath$, which represents the absence of a path. We refer to the set of simple paths as~$\paths$.

\matthewrevision{An edge $e$ is aligned with a path $p$, written $e \pathsaligned p$, if either~$p$ is the trivial path or the destination of~$e$ is equal to the origin of~$p$. If~$e \pathsaligned p$ then~$e$ may be concatenated to~$p$ using the~$\cons$ operator to form a new path~$e \cons p$.}

\subsection{Routing algebras}
\label{sec:routing-algebras}

In order to model all possible DBF protocols rather than one particular protocol, we first abstract over the routing problem being solved by representing it as a \emph{raw routing algebra}.

\begin{definition}[Raw routing algebra]
\label{def:raw-routing-algebra}
A \emph{raw routing algebra} is a tuple $\genericRoutingAlgebra$ where:
\begin{itemize}
\item $\carrier$ is the set of path-weights. A path-weight is the data that the routing protocol associates with each candidate path.
\item $\oplus : \carrier \times \carrier \rightarrow \carrier$ is the choice operator, which given two path-weights returns the preferred path-weight. In the introduction, this was informally referred to as ``\emph{best}''. 
\item $\extension: \allEdges \rightarrow 2^{\carrier \rightarrow \carrier}$ is a map from an edge to the set of policy functions that may be used on that edge. Each policy in $\extension_{ij}$ is a function ${f : \carrier \rightarrow \carrier}$, which when given a path-weight returns a new path-weight formed by extending the original path by the edge $(i,j)$. For notational convenience we define the complete set of policies across all edges as ${\allExtensions = \bigcup_{(i,j) \in \allEdges} E_{ij}}$.
\item $\trivial \in \carrier$ is the path-weight of the trivial path from any router to itself.
\item $\invalid \in \carrier$ is the invalid path-weight, representing the absence of a path.
\item \matthewrevision{$\invalidedge : ((i,j) \in \allEdges) \rightarrow \extension_{ij}$ is the set of constantly invalid policies. For example $\invalidedge_{ij} : \carrier \rightarrow \carrier$ is the policy for edge $(i,j)$ that maps every input path-weight to the invalid path-weight. }
\end{itemize}
\end{definition}
\matthewrevision{\noindent These are the core primitives that we will define the operation of the protocol with respect to. The ``raw'' in the name refers to the fact that, at this stage, we do not assume the primitives of the algebra to satisfy any axioms. We discuss various possible axioms in~Sections~\ref{sec:distance-vector-protocols}~\&~\ref{sec:path-vector-protocols}}

\matthewrevision{ 
\begin{example}[Shortest-paths algebra] 
In the algebra for the shortest-paths routing problem, the path-weights~$\carrier$ are the lengths of the paths i.e.~$\NNinf$. When choosing between two path-weights, the $\min$ operator is used to choose the one with the shortest length. The set of extensions functions are those that add the length $w$ of the edge to the path-weights current length. The trivial path-weight has length~0 and the invalid path-weight has length~$\infty$. The constantly invalid policies are those that add~$\infty$ to the length. 
\end{example}
 
  Table~\ref{table:semirings} contains additional examples of algebras representing a variety of simple routing problems.  Note that in all of these examples, every edge has the same set of policy functions i.e. $E_{ij}$ and $\invalidedge_{ij}$ do not depend on~$i$ and~$j$. This indicates that the policy functions do not use the identity of the routers at either end of the edge. The utility of assigning every edge its own set of policy functions will become clearer when we model path-vector protocols in Section~\ref{sec:path-vector-protocols} where the extension functions will add the current router to the path that the path-weight was generated along. A significantly more complex example algebra that uses these features is discussed in-depth in Section~\ref{sec:bgp-lite-protocol}. } We also note that in practice, a policy function $f \in \extension_{ij}$ may be the composition of an export policy of router~$j$ and an import policy of router~$i$.

\begin{table}
\begin{center} 
\begin{tabular}{lcccccc}
\toprule
Algebra 
& $\carrier$ 	
& $\oplus$ 
& $\extension$ 	
& $\trivial$ 
& $\invalid$ 
& \matthewrevision{$\invalidedge$} \\
\midrule 
Shortest paths 
& $\NNinf$ 	
& $\min$   
& $\extension^+$ 	
& $0$            
& $\infty$ 
& \matthewrevision{$f^+_{\infty}$} \\[3pt]
Longest paths
& $\NNinf$ 	
& $\max$   
& $\extension^+$      
& $\infty$       
& $0$            
& \matthewrevision{$f^+_{0}$}	\\[3pt]
Widest paths	
& $\NNinf$ 	
& $\max$   
& $\extension^{\min}$  
& $\infty$         
& $0$            
& \matthewrevision{$f^{\min}_{0}$} \\[3pt]   
Most reliable paths
& $[0,\ 1] \subset \RR$ 	
& $\max$   
& $\extension^{\times}$  
& $1$       
& $0$            
& \matthewrevision{$f^\times_{0}$}   \\
\bottomrule
\end{tabular} 
\end{center}
\caption{Examples algebras for some simple routing problems, where $\extension^\otimes_{ij} \triangleq \{f^\otimes_w(x) = w \otimes x \mid w \in \carrier\}$ for arbitrary operator $\otimes$ \\ e.g. $\extension^+_{ij} = \{f^+_w(x) = w + x \mid w \in \NNinf\}$. Note that for these simple algebras, every edge has the same set of possible policies.}
\label{table:semirings}
\end{table}

\subsection{Synchronous DBF protocol}
\label{sec:synchronous-model}

Now that we can model \emph{what} problem our abstract DBF protocol is solving, we move on to modelling \emph{how} it solves it. We start by defining a synchronous version of the protocol where update messages between routers are exchanged instantaneously and in parallel.
Let $\nodes \subseteq \allNodes$ be the finite set of routers that will ever participate in the protocol and let $n = \size{\nodes}$. We represent an instantaneous snapshot of the network topology by an $n \times n$ adjacency matrix~$\A$ where $\A_{ij} \in \extension_{ij}$ is the current policy function on the link from router~$i$ to router~$j$. The absence of a link from $i$ to $j$ is represented by the constantly invalid policy~$\invalidedge_{ij}$.

The global state of the synchronous protocol can be entirely represented by the product of the state of each router's routing table. We therefore represent an individual state by a matrix $\X$ where row $\X_i$ is router~$i$'s routing table and so the element~$\X_{ij}$ is the path-weight of~$i$'s best current route to router~$j$. We will use $\tcarrier$ to denote the set of routing table states (i.e. the set of vectors over $\carrier$ of length $n$), and $\mcarrier$ to denote the set of global states (i.e. the set of $n \times n$ matrices over $\carrier$).

We define the \emph{sum} of two matrices $\X$ and $\Y$ as:
\begin{align*}
(\X \oplus \Y)_{ij} \:\: & \triangleq \:\: \X_{ij} \oplus \Y_{ij} \\
\intertext{and the \emph{application} of $\A$ to $\X$ as:}
\A (\X)_{ij} \:\: & \triangleq \:\: \bigoplus_{k} \:\: \A_{ik}(\X_{kj})
\end{align*}
In one \emph{synchronous} round of a distributed Bellman-Ford computation every router in the network propagates its routing table to its neighbours who then update their own tables accordingly. We model this operation as $\aFun_\A$: 
\begin{equation*}
\aFun_\A(\X) \triangleq \A(\X) \oplus \mtrx{I}
\end{equation*}
where $\mtrx{I}$ is the \emph{identity matrix}:
\begin{equation*}
\mtrx{I}_{ij} \triangleq \begin{cases}
\trivial 	& \text{if} \:\: i = j \\
\invalid 	& \text{otherwise}
\end{cases}
\end{equation*}
The nature of the underlying computation becomes clearer if we expand out the definition of a single element of~$\aFun_\A(\X)$. Under the mild assumptions that the trivial path-weight $\trivial$ is chosen over any other path-weight and that any other path-weight is chosen over the invalid path-weight $\invalid$, we have:
\begin{align}
\label{eq:sigma_expanded}
\aFun_\A(\X)_{ij}
&= \A(\X)_{ij} \oplus \I_{ij} \nonumber\\
&= \bigoplus_k \A_{ik}(\X_{kj}) \oplus \begin{cases}
\trivial & \text{if $i = j$} \nonumber\\
\invalid & \text{if $i \neq j$}
\end{cases} \nonumber\\
&= \begin{cases}
\bigoplus_k \A_{ik}(\X_{kj}) \oplus \trivial & \text{if $i = j$} \\
\bigoplus_k \A_{ik}(\X_{kj}) \oplus \invalid & \text{if $i \neq j$}
\end{cases} \nonumber\\
&= \begin{cases}
\trivial & \text{if $i = j$} \\
\bigoplus_k \A_{ik}(\X_{kj}) & \text{if $i \neq j$}
\end{cases}
\end{align}
i.e. router $i$'s always uses the trivial path-weight to route to itself, and its new route to router~$j$ is the best choice out of the extensions of the routes to~$j$ offered to it by each of its neighbours~$k$. 

\matthewrevision{Note that, although $\aFun_\A$ is defined in terms of~$\A$ and~$\X$ which represent the \emph{global} network topology and protocol state respectively, we are still modelling a distributed protocol as each router only has a purely \emph{local} view of the network and protocol state. This can be seen by again considering the expansion above. $\aFun_\A(\X)_{ij}$, the new entry for destination $j$ in the routing table of $i$ after a single update, only depends on~$\A_{ik}$ for all $k$ (i.e. the policies on the incoming edges to~$i$) and~$\X_{kj}$ (i.e. the entry for~$j$ in the routing table of~$k$, as broadcast by~$k$). As previously discussed, if there is no working connection from~$k$ to~$i$ then~$\A_{ik} = \invalidedge_{ik}$ the constantly invalid function, and so no information from~$k$ arrives. }

We can now model the evolution of the synchronous version of protocol, $\sStateFun$, as repeated applications of $\aFun_{\A}$.
\begin{definition}[Synchronous state function]
\label{def:sStateFun}
Given a network topology $\A$ and an initial state $\X$, the state of the synchronous DBF protocol at time $t$ is defined as:
\begin{align*}
\sStateFun^0(\X) &\triangleq \X  \\ 
\sStateFun^{t+1}(\X) &\triangleq \aFun_{\A}(\sStateFun^{t}(\X))
\end{align*}
\end{definition}

\subsection{Asynchronous DBF protocol}
\label{sec:asynchronous-model}

Although we have defined a \emph{synchronous} model, in reality DBF protocols are distributed and asynchronous. Update messages are dispatched at non-deterministic intervals and may be delayed, reordered, duplicated or even lost along the way. New routers and links may be added to the network, existing routers and links may fail, and routers may be rebooted and reset to their original state.
Reasoning about the outcomes of such unpredictable events is known to be extremely challenging. 

One advantage of first defining the synchronous model is that it highlights that the underlying algorithm is inherently iterative. Consequently, an important contribution of this paper to recognise that a DBF protocol is simply a special instance of the family of dynamic asynchronous iterative algorithms. We can therefore use existing work~\cite{uresin1990parallel, daggitt2022dynamic} that construct general mathematical models of such iterations. We now describe a specialisation of the model presented in~\cite{daggitt2022dynamic} to DBF protocols.

We assume a discrete and linear notion of time,~$\timess$, denoting the times of events of interest in the network (a routing table updating, a router failing etc.). Layered on top of time is the notion of epochs,~$\epochs$, which represent contiguous periods of time during which the network topology remains unchanged, i.e. no links or routers fail or are added and no policies are changed. Given this, we can now model a dynamic network,~$\network$, as a function from epochs to adjacency matrices where $\network^e$ is the topology of the network during epoch $e$. 

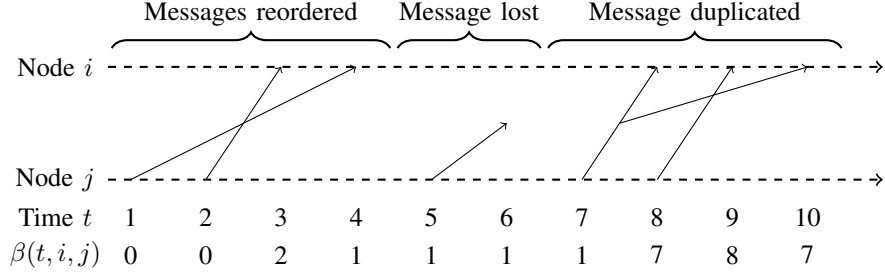
\begin{figure*}[ht]
\centering
\begin{tikzpicture}
\def\a{0}
\def\b{0.7}
\def\c{2}
\def\d{4.5}
\def\e{6}
\def\f{6.5}
\def\g{10}
\def\h{10.5}

\def\braceLabelY{2.2}
\def\braceY{1.75}
\def\brcD{0.05}

\def\th{-0.5}
\def\bh{-1}
\def\ih{1.5}
\def\jh{0}

\def\tickHeight{0.3}
\def\labelY{-0.03}

\tikzstyle{bracesU}=[thick, decorate, decoration={brace,amplitude=6pt,raise=0pt}]
\tikzstyle{bracesD}=[thick, decorate, decoration={brace,amplitude=6pt,raise=0pt,mirror}]

\draw [thick,dashed,->] (0.7,\ih) -- (11,\ih);
\draw [thick,dashed,->] (0.7,0) -- (11,0);
\node at (0,\ih) {Node $i$};
\node at (0,\jh) {Node $j$};

\node at (0,\th) {Time $t$};
\node at (1,\th) {1};
\node at (2,\th) {2};
\node at (3,\th) {3};
\node at (4,\th) {4};
\node at (5,\th) {5};
\node at (6,\th) {6};
\node at (7,\th) {7};
\node at (8,\th) {8};
\node at (9,\th) {9};
\node at (10,\th){10};

\node at (0,\bh) {$\beta(t,i,j)$};
\node at (1,\bh) {0};
\node at (2,\bh) {0};
\node at (3,\bh) {2};
\node at (4,\bh) {1};
\node at (5,\bh) {1};
\node at (6,\bh) {1};
\node at (7,\bh) {1};
\node at (8,\bh) {7};
\node at (9,\bh) {8};
\node at (10,\bh){7};

% Events
\draw [->](1,\jh) -- (4,\ih);
\draw [->](2,\jh) -- (3,\ih);

\draw [->](5,\jh) -- (6,0.75);

\draw [->](7,\jh) -- (8,\ih);
\draw [->](7.5,0.75) -- (10,\ih);
\draw [->](8,\jh) -- (9,\ih);

% Ticks
%\draw [thick] (1,0) -- (1,-\tickHeight);

% Lower braces
\draw[bracesU] (\b+\brcD,\braceY) -- (\d-\brcD,\braceY);
\draw[bracesU] (\d+\brcD,\braceY) -- (\f-\brcD,\braceY);
\draw[bracesU] (\f+\brcD,\braceY) -- (\h-\brcD,\braceY);

% Lower brace labels
\node[align=center] at ({(\b+\d)/2},\braceLabelY){Messages reordered};
\node[align=center] at ({(\d+\f)/2},\braceLabelY){Message lost};
\node[align=center] at ({(\f+\h)/2},\braceLabelY){Message duplicated};
\end{tikzpicture}
\caption{ Behaviour of the data flow function $\beta$. $\beta(t,i,j)$ is the time at which the latest message to arrive at $i$ up to time $t$ was sent from $j$. Messages from router $j$ to router $i$ may be reordered, lost or even duplicated. The only constraint is that every message must arrive after it was sent. Reproduced from~\cite{daggitt2022dynamic}.}
\label{fig:schedule}
\end{figure*}

The non-deterministic evolution of the protocol is described by a \emph{schedule}.

\begin{definition}[Schedule] (Definition~9 in \cite{daggitt2022dynamic})
\label{def:schedule}
A \emph{schedule} is a tuple of functions $(\alpha, \beta, \eta, \pi)$ where:
\begin{itemize}
\item $\alpha : \timess \rightarrow 2^\nodes$ is the \emph{activation function}, where $\alpha(t)$ is the set of routers which update their routing table at time~$t$.

\item $\beta : \timess \times \nodes \times \nodes \rightarrow \timess$ is the \emph{data flow function}, where $\beta(t,i,j)$ is the time at which the route advertisements used by router $i$ at time $t$ was sent by router $j$.

\item $\eta : \timess \rightarrow \epochs$ is the \emph{epoch function}, where $\eta(t)$ is the epoch at time $t$.

\item $\pi : \epochs \rightarrow 2^\nodes$ is the \emph{participants function}, where~$\pi(e)$ is the subset of routers currently participating in the protocol during epoch e.
\end{itemize}
such that:
\begin{enumerate}[label=\textbf{S\arabic*}]
\item information only travels forward in time 
\begin{equation*}
\forall i,j,t : \beta(t+1,i,j) \leq t
\end{equation*}
\item epochs increase monotonically
\begin{equation*}
\forall t_1 , t_2 : t_1 \leq t_2 \Rightarrow \eta(t_1) \leq \eta(t_2)
\end{equation*}
\end{enumerate}
\end{definition}
\noindent This is a very weak model of asynchronous communication. As shown in Figure~\ref{fig:schedule}, nothing forbids the data flow function~$\beta$ from delaying, losing, reordering or duplicating messages. It also distinguishes between a router being a member of the network but not currently participating in the protocol. This allows us to model messages that continue to arrive from the previous epoch, e.g. router $i$ can receive outdated route advertisements from router $j$ even after router $j$ has ceased to participate.

However, it does mean that we need to distinguish between the topology of network and the topology as seen by the active participants. In particular we need to enforce that information only flows between two participating routers.
\begin{definition}[Participating topology]
\label{def:participating-topology}
Given a network $N$, an epoch $e \in \epochs$ and a set of participants $p \in 2^\nodes$, we define $\A^{ep}$ to be the participating topology:
\begin{equation*}
\A^{ep}_{ij} = \begin{cases}
\network^e_{ij} & \text{if $\{i, j\} \subseteq p$}\\
\invalidedge_{ij} & \text{otherwise}
\end{cases}
\end{equation*}
\end{definition}
\noindent For notational convenience, given a schedule and a network, we also define $\rho : \timess \rightarrow 2^\nodes$, the set of participants at time $t$, and $\A^t$, the current participating topology at time $t$, as:
\begin{align*}
\rho(t) & \triangleq \pi(\eta(t)) \\
\A^t 	& \triangleq \A^{\eta(t)\rho(t)}
\end{align*}
Using this machinery it is now possible to define the full asynchronous state function, $\delta$, in terms of the original iteration function $\aFun_{\A^t}$.

\begin{definition}[Asynchronous state function](Definition 10 from~\cite{daggitt2022dynamic} with $F^{t} = \aFun_{\A^{t}}$, $\bot = \I$)
\label{def:aStateFun}
Given a dynamic network~$\network$, an initial state~$\X$, and a schedule $(\alpha, \beta, \eta, \pi)$, the state of the asynchronous DBF protocol at time $t$ is defined as:
\begin{equation*}
\aStateFun^t_{ij}(\X) = \begin{cases}
\I_{ij} & \text{if $i \notin \rho(t)$} \\
\X_{ij} & \text{else if $t = 0$ or $i \notin \rho(t-1)$}\\
\aStateFun^{t-1}_{ij}(\X) & \text{else if $i \notin \alpha(t)$} \\
\aFun_{\A^t}(\localView)_{ij} & \text{otherwise}
\end{cases}
\end{equation*}
where $\localView_{kj} = \aStateFun^{\beta(t,i,k)}(\X)_{kj}$.
\end{definition}

This models the evolution of the protocol as follows. If a router is not currently participating then it adopts its non-participating state $\I_{ij}$. If a router has just begun to participate, either because $t = 0$ or because it was not participating at the previous time step, it adopts the initial state $\X_{ij}$. If the router is a continuing participant and is currently inactive then its state remains unchanged from its state at the previous time step,~$\delta^{t-1}_{ij}(\X)$.
Otherwise, it is currently active, and it updates its state in accordance with router $i$'s $\emph{local view}$,~$\localView$, of the global state as currently advertised by its neighbours. More concretely,~$\localView_{kj}$ is the last route advertisement for router~$j$ that router~$i$ received from its neighbour~$k$.

Note that this is a straightforward generalisation of the synchronous model as we can immediately recover $\sStateFun$ from $\aStateFun$ by setting $\alpha(t) = V$, $\beta(t,i,j) = t-1$, $\eta(t) = 0$ and $\pi(t) = V$ i.e. there is a single epoch in which at each time step, every router is both participating and active and all messages only take a single time step to propagate.

\subsection{Axioms for distance-vector protocols}
\label{sec:distance-vector-protocols}

\matthewrevision{Having defined the evolution of an abstract DBF protocol, $\aStateFun$, parametrised by an algebra~$\genericRoutingAlgebra$,  it is clear that $\aStateFun$'s behaviour will be dependent on the properties of the algebra. So far, we have informally claimed that the algebraic primitives represent certain operations, for example $\oplus$ chooses between two candidate path-weights. We now make this formal by defining the axioms the primitives should obey in order for them to match our basic intuition about how they should operate in a standard distance vector protocol.}

\begin{definition}[Routing algebra]
\label{def:routing-algebra2}
A raw routing algebra $\genericRoutingAlgebra$ is a \emph{routing algebra} if it obeys the following axioms:
\begin{enumerate}[label=\textbf{R\arabic*}]
\item \label{ass:sel}
$\oplus$ is selective - \emph{choosing between two path-weights always returns one of the two path-weights}.
\begin{equation*}
\forall x,y \in \carrier: x \oplus y \in \{x,\ y\}
\end{equation*}

\item \label{ass:assoc}
$\oplus$ is associative - \emph{the order doesn't matter when making a sequence of choices between path-weights}.
\begin{equation*}
\forall x, y, z \in \carrier: x \oplus (y \oplus z) = (x \oplus y) \oplus z
\end{equation*}

\item \label{ass:comm} 
$\oplus$ is commutative - \emph{the order doesn't matter when choosing between two path-weights}.
\begin{equation*}
\forall x,y \in \carrier: x \oplus y = y \oplus x
\end{equation*}

\item \label{ass:trivial-annihilator}
$\trivial$ is an annihilator for $\oplus$ - \emph{the trivial path-weight is more desirable than any other path-weight}.
\begin{equation*}
\forall x \in \carrier: x \oplus \trivial = \trivial = \trivial \oplus x
\end{equation*}

\item \label{ass:invalid-identity}
$\invalid$ is an identity for $\oplus$ - \emph{any path-weight is preferable to the invalid path-weight}.
\begin{equation*}
\forall x \in \carrier: x \oplus \invalid = x = \invalid \oplus x
\end{equation*}

\item \label{ass:invalid-fixed-point}
$\invalid$ is a fixed point for all policy functions - \emph{the extension of the invalid path-weight is also the invalid path-weight}.
\begin{equation*}
\forall f \in \allExtensions: f(\invalid) = \invalid
\end{equation*}
\matthewrevision{
\item \label{ass:invalid-function-point}
$\invalidedge$ is the set of constantly invalid policies - \emph{no information flows through non-existent links}.
\begin{equation*}
\forall (i,j) \in \allEdges,\ x \in S: \invalidedge_{ij}(x) = \invalid
\end{equation*}}
\end{enumerate}
\end{definition}
\matthewrevision{
\noindent It is important to note that these axioms are only those necessary and sufficient to guarantee that the algebraic primitives obey our intuition about them. As discussed in Section~\ref{sec:results}, they are \emph{not} sufficient to guarantee that $\aStateFun$ itself is well-behaved. 

We also note that some of these axioms may be relaxed when modelling more complex variations on standard routing protocols. For example, in multi-path routing the choice operator may not be selective as it can return a mix of the candidate path-weights from both of its arguments~\cite{chau2006policy,daggitt2018asynchronous}. Although an important line of work, this paper will not consider these variants further.}
 
\subsection{Axioms for path-vector protocols}
\label{sec:path-vector-protocols}

Path-vector protocols (e.g. BGP) are a subclass of DBF protocols which track the path along which a path-weight was generated. Routers then use this information to eliminate any path-weights that were generated via a cycle. It is well known that path-vector protocols have better convergence behaviour than distance-vector protocols. \matthewrevision{We would therefore like to extend our model to handle path-vector protocols as well.

The standard approach is to take an algebra $\genericRoutingAlgebra$ and model the path-vector protocol's algorithm as operating over the set $\carrier \times \paths$ instead of the set $\carrier$. The synchronous iteration function $\aFun_\A$ in Equation~\ref{eq:sigma_expanded} is then updated to track paths and eliminate path-weights generated along looping paths.

However, this approach has significant drawbacks. Firstly, the model of a path-vector protocol is no longer strictly a special case of the model for DBF protocols as they operate over different sets and have different algorithms. In turn, this means that they require different proofs of correctness. While this can (perhaps) be finessed in a pen-and-paper proof, it would result in a significant duplication of work during the computer formalisation described in Section~\ref{sec:formalisation}. Secondly, it means that the choice operator $\oplus$ can't make decisions based on properties of the path (e.g. its length), as $\oplus$ still operates over $\carrier$ rather than $\carrier \times \paths$. Lastly, we lose abstraction over how the paths are stored. For example, both $\carrier \times \paths$ and $\paths \times \carrier$ should be valid sets for a path-vector protocol to operate over, but the standard approach artificially elevates $\carrier \times \paths$ over $\paths \times \carrier$.

We therefore now propose a new approach to modelling path-vector protocols which overcomes all of these problems. We argue that the paths should be modelled as being stored \emph{inside} the path-weight $\carrier$, rather than outside of it. Likewise $\oplus$ and $\extension$ should handle the tracking of paths and removal of looping paths, rather than lifting these additional operations to the algorithm. However, how can we enforce that paths are stored abstractly within elements of $\carrier$? We cannot directly inspect the internal structure of $\carrier$, $\oplus$ and $\extension$ and, in order to for the model to be as general as possible, nor is it desirable to. 

The answer is to assume that there exists a projection function $\pathf : \carrier \rightarrow \paths$ that extracts the path stored within the path-weight. We then axiomitise $\pathf$ to interact with $\oplus$ and $\extension$ in the right way. Using this approach, we now define axioms for an algebra that can be used to model path-vector-specific behaviour.}

\begin{definition}[Path algebra]
\label{def:path-algebra}
A \emph{path algebra} is a raw routing algebra $\genericRoutingAlgebra$ equipped with an additional function $\pathf : \carrier \rightarrow \paths$ that obeys the following properties:
\begin{enumerate}[label=\textbf{P\arabic*}]
\item \label{ass:path-invalid}
A path-weight \matthewrevision{stores} the invalid path iff it is the invalid path-weight.
\begin{equation*}
\forall x \in \carrier: x = \invalid \Leftrightarrow \pathf(x) = \invalidpath
\end{equation*}
\item \label{ass:path-trivial}
The trivial path-weight \matthewrevision{stores} the empty path.
\begin{equation*}
\forall x \in \carrier: x = \trivial \Rightarrow \pathf(x) = \trivialpath
\end{equation*}
\item \label{ass:path-extension}
Applying an edge's policy function to a path-weight simply extends the path-weight's path by that edge, unless it does not result in a non-simple path in which case it becomes the invalid path.
\begin{equation*}
\begin{array}{l}
\forall x \in \carrier, (i,j) \in \allEdges, f \in \extension_{ij} : \\
\pathf(f(x)) = 
\begin{cases}
	\invalidpath       
		& \text{if $i \in \pathf(x)$} \\
	\invalidpath       
		& \text{if not $j \pathsaligned \pathf(x)$} \\
	
	(i,j) \cons \pathf(x)
		& \text{otherwise}
\end{cases}
\end{array}
\end{equation*}
\end{enumerate}
\end{definition}
\noindent The $\pathf$ function allows us to extract the path stored inside the path-weight, along which the path-weight was generated, without revealing how the path is stored internally inside the path-weight. Assumptions P1--P3 ensure that the algebra's implementation is tracking the paths correctly. In P3, the first case $i \notin \pathf(x)$ guarantees that path-weights generated along paths with cycles in them are eliminated. The second case when $j \pathsaligned \pathf(x)$ does not hold (i.e. $\pathf(x)$ is not actually a path to $j$), ensures that invalid routing announcements are eliminated. Note that being able to state P3 is one of the reasons that we require that each edge $(i, j)$ comes with its own set of policy functions $\extension_{ij}$ in Definition~\ref{def:raw-routing-algebra}. If instead we had only required the existence of a single set of policy functions (i.e. $\allExtensions$) then we would be unable to define it.

\matthewrevision{
As an example, we now construct a simple path-algebra for a shortest-paths path-vector protocol. We construct a more interesting path-algebra that models a BGP-like protocol in Section~\ref{sec:bgp-lite-protocol}.  

\begin{example}[Shortest-paths path-vector algebra]
The set of path-weights $\carrier$ is $(\paths \times \mathbb{N}) \uplus \{\infty\}$ where $\infty$ represents the invalid path-weight. The trivial path-weight is the pair $(0, \trivialpath)$. The choice operator, $\oplus$, first compares the lengths of the path-weights, and then breaks ties by comparing the paths lexicographically:
\begin{align}
x &\oplus \infty &\triangleq&\ x \\
\infty &\oplus y &\triangleq&\ y \notag \\
(m, p) &\oplus (n , q) &\triangleq&
\begin{cases}
x & \text{if $m < n$} \\
y & \text{else if $n < m$} \\
x & \text{else if $p \leq_{lex} q$} \\
y & \text{otherwise}
\end{cases} \notag
\end{align}
The policy functions $\extension_{ij}$ for an arbitrary edge $(i,j)$ is the set of $f^l_{ij}$ for all lengths $l \in \mathbb{N}$ where:
\begin{align}
f^l_{ij}(\infty) & \triangleq \infty \\
f^l_{ij}(m, p)   & \triangleq \begin{cases}
\infty & \text{if $i \in p$ or not $j \pathsaligned p$} \\
(m + l, (i,j) \cons p) & \text{otherwise} \end{cases} \notag
\end{align} 
plus the additional constantly invalid function $f^{\infty}_{ij}(x) \triangleq \infty$. Finally the required $\pathf$ function is defined as:
\begin{align}
\pathf(\infty) &\triangleq \invalidpath \\
\pathf(m, p) &\triangleq p \notag
\end{align}
\end{example}

\noindent We note that as a result of the flexibility of Definition~\ref{def:path-algebra} we can easily flip the ordering of the pair in the algebra's path-weights (i.e. from $\mathbb{N} \times \paths$ to $\paths \times \mathbb{N}$) and it would \emph{still} be a path algebra and the path function would simply change which component of the pair it returned.

In order to construct a full model of the path-vector, it is sufficient to note that a path algebra is simply a routing algebra with extra assumptions. Consequently, unlike in the standard approach, the algorithms described in Section~\ref{sec:synchronous-model}~\&~\ref{sec:asynchronous-model} can be immediately reused to construct an asynchronous model of a path-vector protocol.
}

\subsection{Discussion and novel contributions}
\label{sec:model-contributions}

\matthewrevision{
To recap, we have have now successfully constructed $\aStateFun$, an abstract model of the execution of a generic asynchronous DBF protocol. The model is parameterised by:
\begin{enumerate}
\item $\genericRoutingAlgebra$ - a routing algebra which represents the routing problem being solved.

\item $\network$ - a function from epochs to adjacency matrices which represents the topology of the network over time.

\item $(\alpha, \beta, \eta, \pi)$ - a schedule which represents the timings of events in the network, e.g. node activations, message arrivals.
\end{enumerate}
and we have defined sensible sets of axioms for the algebra that represent distance-vector and path-vector protocols.}

We will now briefly revisit the main advantages of this new model over models proposed in previous work.

\paragraph{Simpler, more general algebraic structure}
The algebraic structure we propose in Definition~\ref{def:raw-routing-algebra} is both more general and simpler than that proposed by Sobrinho~\cite{sobrinho2005algebraic}, which in turn was inspired by the classical semiring approach~\cite{gondran2008graphs, carre1971algebra}. The main difference is in how we model path-weight extension. In semiring algebras it is modelled by a second binary operator $\otimes : \carrier \times \carrier \rightarrow \carrier$. However this is incapable of modelling modern routing policies such as those described in Sections~\ref{sec:introduction}~\&~\ref{sec:bgp-lite-protocol}. Sobrinho generalised this binary operator by introducing labels $L$, signatures $\Sigma$, an operator $\triangleright: L \times \Sigma \rightarrow \Sigma$ and a function $g : \Sigma \rightarrow \carrier$ mapping signatures to path-weights. However, by modelling policies as a set of functions,~$\extension$, and as demonstrated in Section~\ref{sec:bgp-lite-protocol}, we retain the ability to model modern routing polices while reducing the size of the algebra from 8 primitives to 6 primitives.

\paragraph{Paths as part of the algebra}
The second way in which the algebra is more general is that, by assigning each edge its own set of policy functions, we model the tracking of paths in path-vector protocols inside of the algebra itself. In contrast, Sobrinho~\cite{sobrinho2005algebraic} models the paths as part of the algorithm, i.e. in the equivalent definition of $\aStateFun$. This means that the paths along which the path-weights were generated are not provided as inputs to his~$\choice$ and~$\extension$ operators. Therefore the algebra cannot alter or make decisions based on those paths. Consequently, unlike ours, Sobrinho's approach cannot model policies such as route filtering or path inflation in BGP. 

\paragraph{Abstraction of paths}
A third advantage of our new Definition~\ref{def:path-algebra} of a path algebra is that it only requires that the $\pathf$ function can extract a simple path from a path-weight, but does not force the protocol to reveal how paths are actually stored \emph{internally}. This abstraction results in proofs that are far more robust to variations in how paths are stored internally by the protocol. For example, in Section~\ref{sec:bgp-lite-protocol} we show that our algebras can model BGP-like path inflation by simply adjusting the $\pathf$ function to strip out the inflated ASs, without having to adjust the proofs of convergence. In contrast, the model of path-vector protocols in Sobrinho~\cite{sobrinho2005algebraic} uses a concrete definition of paths, and therefore is incapable of modelling path inflation without altering the proof of convergence itself.

\paragraph{Use of existing asynchronous theory}
The additional complexity introduced by the more expressive dynamic asynchronous model $\aStateFun$, such as routers failing and rebooting, at first glance appears daunting. However, as we will discuss further in Section~\ref{sec:results}, our use of the existing theory of dynamic iterative algorithms means that we have access to a range of powerful theorems that, impressively, are capable of totally abstracting away the asynchronicity of the system.

\paragraph{Executable model of the protocol}
Sobrinho~\cite{sobrinho2005algebraic} only provides an abstract temporal model of individual events that occur during the operation of the protocol. In contrast, as far as we are aware, ours is the first fully formal and executable model of the full dynamic behaviour of an abstract DBF protocols in which routers and links can fail, join and reboot. This means that our model can be run and its behaviour observed, and that our proofs can be written in a style amenable to the formal verification discussed in Section~\ref{sec:formalisation}.

\section{Convergence results}
\label{sec:results}

\subsection{Some useful concepts}

Initially, we define some concepts that, although not required for constructing the model in Section~\ref{sec:model}, are used in stating and reasoning about its convergence behaviour.

\vspace{1em}
\subsubsection{Ordering over path weights}

Given any routing algebra $\genericRoutingAlgebra$, the choice operator $\oplus$ induces a preference order over path-weights as follows:
\begin{definition}[Preference relation over path-weights ]
\label{def:leq-path-weights}
A path-weight~$x$ is preferred to path-weight~$y$, written $x \leqWeight y$, if $x$ is chosen over $y$:
\begin{align*}
x \leqWeight y \spacing & \triangleq \spacing x \oplus y = x \\
x \lessWeight y \spacing & \triangleq \spacing x \leqWeight y \: \wedge \: x \not= y. 
\end{align*}
\end{definition}
\noindent As $\oplus$ is associative, commutative and selective, we have that $\leqWeight$ is a total order and that for all path-weights $x$ we have $\trivial \leqWeight x \leqWeight \invalid$.

\vspace{1em}
\subsubsection{Assignments}

In our model which router is using which path-weight is implicitly defined by the path-weight's location in the routing state (i.e. given a routing state $\X$ then the path-weight $\X_{ij}$ is being used by router $i$).  Nonetheless a more explicit notion of which router is using which path-weight is crucial for stating sufficient and necessary conditions for the convergence of the model over a given network. We shall refer to such a pair of a starting node and a path-weight as an \emph{assignment}.
\begin{definition}[Assignment]
\label{def:assignments}
An assignment $(i, x)$ is a pair of a router $i \in \nodes$ and a path-weight $x \in \carrier$, and represents that router $i$ is currently using path-weight $x$. The set of assignments $\assignments \triangleq \nodes \times \carrier$.
\end{definition}
\noindent Given an assignment $(i, x)$, during the correct operation of a path-vector protocol one would expect the first node in $\pathf(x)$ to be equal to $i$. Note that an assignment is a generalisation of the notion of a \emph{couplet} defined by Sobrinho ~\cite{sobrinho2005algebraic}. Whereas a couplet stores the entire path along which the path-weight was generated, an assignment only stores the first router in the path. This generalisation is necessary as we seek to reason about distance-vector as well as path-vector protocols, and in the former case we do not have access to the full path.

Several binary relations may be defined over the set of assignments. The first is the lifting of the preference order over path-weights to assignments.
\begin{definition}[Preference relation over assignments]
\label{def:assignment-preference}
An assignment~$a$ is preferred to assignment~$b$, written $a \leqAss b$, if they have the same router and the path-weight of $a$ is preferred to the path-weight of $b$:
\begin{align*}
(i, x) \leqAss (j , y) \spacing \triangleq \spacing i &= j \wedge x \leqWeight y \\
(i, x) \lessAss (j , y) \spacing \triangleq \spacing i &= j \wedge x \lessWeight y
\end{align*}
\end{definition}
\noindent Note that unlike $\leqWeight$ which is a total order over path-weights, $\leqAss$ is only a partial order over assignments, as it can only relate assignments to the same router. Intuitively, this makes sense as routers in a DBF protocol only compare paths that begin at themselves.

\begin{definition}[Extension relation]
\label{def:assignment-extension}
An assignment $a$ \emph{extends} assignment $b$ in the current topology $\A$, written ${a \extendedBy{\A} b}$, if the extension of the non-invalid path-weight of $a$ across the relevant link is equal to the path-weight of $b$:
\begin{equation*}
(j, y) \extendedBy{\A} (i , x) \triangleq \A_{ij}(y) = x \wedge y \neq \invalid  
\end{equation*}
\end{definition}
\noindent In the context of a path-vector protocol, the axioms of a path algebra are sufficient to guarantee that if $(j , y) \extendedBy{\A} (i , x)$ then $\pathf(x) = (i , j) \cons \pathf(y)$.

\begin{definition}[Threatens relation]
\label{def:assignment-threatens}
An assignment $a$ \emph{threatens} assignment $b$ in the current topology $\A$, written ${a \threatenedBy{\A} b}$, if the extension of the non-invalid path-weight of $a$ across the relevant link is preferred to the path-weight of $b$:
\begin{equation*}
(j , y) \threatenedBy{\A} (i , x) \triangleq \A_{ij}(y) \leqWeight x \wedge y \neq \invalid
\end{equation*}
\end{definition}
\noindent Intuitively $(j , y) \threatenedBy{\A} (i , x)$ means that after the next update router $i$ may switch away from $x$ and instead use a path-weight that is equal to or preferred to $\A_{ij}(y)$. It is easy to see that if $a \extendedBy{\A} b$ then $a \threatenedBy{\A} b$.

\vspace{0.5em}
\subsubsection{Non-participating nodes}

As seen in Definition~\ref{def:participating-topology}, the constraint that non-participating routers do not advertise routes to or receive routes from other routers is enforced at the level of the adjacency matrix $\A^{ep}$ and therefore the iteration function $\aFun_{\A^{ep}}$. Consequently, after a single iteration the state of any non-participating router's routing table will always be the corresponding entry of the identity matrix, i.e. the trivial route to itself and invalid routes to all other routers in the network. In~\cite{daggitt2022dynamic} such states are said to be \emph{accordant} with the current set of participants.

\begin{definition}[Accordant states] (Definition 15 from~\cite{daggitt2022dynamic} with $\bot = \I$)
\label{def:accordant-states}
The set of routing states that are accordant with set of participants $p$, $\accSet$, is defined as:
\begin{equation*}
\accSet \triangleq \{ \X \in \mcarrier \mid \forall i,j: i \notin p \Rightarrow \X_{ij} = \I_{ij} \} 
\end{equation*}
\end{definition}

\subsection{Definition of convergence}
\label{sec:definition-of-convergence}

Before proving convergence results about our model, we must first establish what does convergence even mean for a dynamic, asynchronous, iterative algorithm such as a DBF protocol. In this section, will build up a suitable definition by discussing in turn the complications introduced by each of these adjectives.

\paragraph{Iterative} The definition of convergence for a synchronous iterative algorithm is relatively simple. After a finite number of iterations, the algorithm must reach a \emph{stable state} such that further iterations result in no further change. In terms of the synchronous model defined in Section~\ref{sec:synchronous-model}, this would translate to:
\begin{equation*}
\exists k^* \in \NN: \forall k \in \NN : k \geq k^* \Rightarrow \sStateFun^{k}(\X) = \sStateFun^{k^*}(\X)
\end{equation*}
The resulting stable state $\X^* = \sStateFun^{k^*}(\X)$ is a \emph{fixed point} for~$\aFun_\A$, as:
\begin{equation*}
\aFun_\A(\X^*) 
= \aFun_\A(\sStateFun^{k^*}(\X)) 
= \sStateFun^{k^* + 1}(\X)
= \sStateFun^{k^*}(\X)
= \X^*
\end{equation*}
From the perspective of a routing protocol, $\X^*$ being a fixed point is equivalent to the statement that in state $\X^*$ no router can improve its selected path-weights by unilaterally choosing to switch. Such a state can therefore be viewed as a \emph{locally} optimal solution to the optimisation problem of finding the best consistent set of path choices in the current network topology. As will be discussed further in Section~\ref{sec:routing-algebra-classes}, whether such a state is necessarily \emph{globally} optimal depends on the properties of the routing algebra being used.

\paragraph{Asynchronous} At first glance, it appears relatively easy to adapt the above definition to a basic asynchronous model with only a single epoch: simply require that \emph{for every schedule} the asynchronous iteration reaches a stable state. Unfortunately due to the asynchronous and distributed nature of the computation, and as discussed in Section~\ref{sec:asynchronous-model}, our definition of a schedule does not assume reliable communication between routers. Consequently, in an arbitrary schedule a router may never activate or a pair of routers may never succeed in communicating. In such a schedule, it is clearly impossible for the protocol to converge to the desired stable state. Therefore, we need to restrict ourselves to only requiring that the protocol converges over \emph{reasonable} schedules. 

As long established in the asynchronous iterations literature~\cite{uresin1990parallel}, the right definition of ``reasonable" is that the schedule contains \emph{pseudocycles}.  Informally, a pseudocycle can be thought of the asynchronous counterpart of a single synchronous iteration.  \matthewrevision{We now present the formal definition of this concept from~\cite{daggitt2022dynamic}. 
\begin{definition}[Expiry period, Definition~12 in~\cite{daggitt2022dynamic}]
A period of time $[t_1,t_2]$ is a \emph{expiry period} for node~$i$ if $\eta(t_1) = \eta(t_2)$ and for all nodes $j$ and times $t \geq t_2$ then $t_1 \leq \beta(t,i,j)$.
\end{definition}
\noindent In other words, an expiry period is a period of time $[t_1, t_2]$ within a single epoch, such that after time $t_2$ node $i$ never receives a message which was sent before $t_1$, i.e. all messages to $i$ that were in flight before $t_1$ have either been dropped or have arrived by time $t_2$.
\begin{definition}[Dynamic activation period, Definition~11 in~\cite{daggitt2022dynamic}]
A period of time $[t_1,t_2]$ is a \emph{dynamic activation period} for node~$i$ if $\eta(t_1) = \eta(t_2)$ and there exists a time $t \in [t_1,t_2]$ such that $i \in \alpha(t)$.
\end{definition}
\noindent Intuitively, an activation period for node $i$ is a period of time within a single epoch in which router $i$ updates its routing table at least once.
\begin{definition}[Dynamic pseudocycle, Definition~13 in~\cite{daggitt2022dynamic}]
A period of time $[t_1,t_2]$ is a dynamic pseudocycle if $\eta(t_1) = \eta(t_2)$ and for all nodes $i$ $\in \rho(t_1)$ there exists a time $t \in [t_1,t_2]$ such that $[t_1,t]$ is an expiry period for node $i$ and $[t,t_2]$ is an activation period for node $i$.
\end{definition}
\noindent A pseudocycle is a period of time within a single epoch during which every participating node undergoes an expiry period followed by an activation period. Lemma~8 in~\cite{daggitt2022dynamic} proves that during a pseudocycle, the asynchronous model $\aStateFun$ makes at least as much progress towards a stable state as the synchronous model $\sStateFun$ would with a single synchronous iteration. This is therefore the right notion of a schedule behaving reasonably.}

A second aspect that we must consider when moving from the synchronous to the asynchronous model is that of non-determinism. As discussed in Section~\ref{sec:what-is-policy-rich-routing}, BGP suffers from non-deterministic convergence where the protocol may end up in multiple stable states, and which stable state is reached depends on the schedule and the initial state, rather than just the topology. This problem is known colloquially as a BGP wedgie~\cite{rfc4264}. Our definition of convergence should therefore require it to be deterministic, i.e. in each epoch only one stable state may be reached and that it is dependent only on the topology and the set of participants.

\paragraph{Dynamic} Finally, we now consider the full dynamic model with multiple epochs. In this scenario the network topology may change so frequently (i.e. the epochs may be so short), that the protocol \emph{never} gets sufficient time to converge to a solution. Therefore we would like our definition of convergence to only require that the protocol will converge \emph{if given a sufficient period of stability}. It is suggested in~\cite{daggitt2022dynamic} that dynamic iterative algorithms, such as DBF routing protocols, should therefore be referred to as \emph{convergent} rather than saying that \emph{they converge}.

The final issue is that the protocol may only be convergent over some, rather than all, epochs and sets of participants. As shown by Sobrinho~\cite{sobrinho2005algebraic} and discussed further in Section~\ref{sec:routing-algebra-classes}, DBF protocols only converge over network topologies that are \emph{free}. Therefore the definition of convergent must be granular enough to express that the protocol is only convergent over a subset of epochs and sets of participants. A pair containing an epoch and a set of participants is referred to as a \emph{configuration}. The set of all configurations $\configs$ is therefore defined as:
\begin{equation*}
\configs \triangleq \epochs \times 2^\nodes
\end{equation*}

Taking all the above points into account, we therefore arrive at the following definition of what it means for our protocol to be \emph{convergent}:

\begin{definition}[Convergent] (Definition 14 from \cite{daggitt2022dynamic})
\label{def:convergent}
The DBF protocol is convergent over a set of configurations $C \subseteq \configs$ iff:
\begin{enumerate}
\item for every configuration $(e, p) ∈ C$ there exists a fixed point $\X^*_{ep}$ for $F_{\A^{ep}}$ and a number of iterations $k^*_{ep}$.

\item for every initial state $\X$, schedule $(\alpha, \beta, \eta, \pi)$ and time~$t_1$ then if $(\eta(t_1), \rho(t_1)) \in C$
and the time period~$[t_1, t_2]$ contains $k^*_{\eta(t_1)\rho(t_1)}$ pseudocycles then for every time $t_3$ such that $t_3 \geq t_2$ and $\eta(t_2) = \eta(t_3)$ then ${\aStateFun^{t_3} (\X) = \X^*_{\eta(t_1)\rho(t_1)}}$
\end{enumerate}
\end{definition}
\noindent The essence of this definition is that for every valid configuration there exists a fixed point and a number of pseudocycles such that if, at an arbitrary point in time the configuration of the current epoch is in the set of valid configurations and the schedule contains that number of pseudocycles, then after those pseudocycles the protocol will have reach the fixed point and will remain in that fixed point until the end of the current epoch.

\subsection{An existing convergence theorem}
\label{sec:convergence_theorem}

One of the most important advantages to using an established model of asynchronous iterations in Section~\ref{sec:asynchronous-model} is that there already exist theorems describing sufficient conditions of such an iteration being convergent. The most powerful (and perhaps initially surprising) feature of these results is that the sufficient conditions only involve properties of the functions being iterated, $\aFun_{\A^{ep}}$, rather than the full asynchronous iteration,~$\delta$. Consequently one can use them to prove that $\delta$ is convergent without ever having to directly reason about the unreliable, asynchronous and dynamic nature of the underlying network, e.g. messages being lost or reordered, routers and links failing or being added.

 In this paper, we will use the AMCO conditions from~\cite{daggitt2022dynamic}. For a broader survey of the range of alternative sufficient conditions available see~\cite{uresin1990parallel}.

\begin{definition}[AMCO, Definition 23 in \cite{daggitt2022dynamic}]
\label{def:amco}
The functions $\aFun_{\A^{ep}}$ are a set of  \emph{asychronously metrically contracting operators} (AMCO) over a set of configurations~$C$ if for every epoch~$e$ and set of participants~$p$ such that ${(e,p) \in C}$ then for all ${i \in p}$ there exists a dissimilarity function ${d^{ep}_i : \tcarrier \times \tcarrier \rightarrow \mathbb{N}}$ such that:
\begin{enumerate}[label=\textbf{D\arabic*}]
\item \emph{$\tmetric{i}$ is indiscernable}
\begin{equation*}
\forall \xt, \yt : \tmetric{i}(\xt,\yt) = 0 \Leftrightarrow \xt = \yt
\end{equation*}
\item \emph{$\tmetric{i}$ is bounded}
\begin{equation*}
\exists n : \forall \xt, \yt: \tmetric{i}(\xt,\yt) \leq n
\end{equation*}
\end{enumerate}
and if $\smetric(\X, \Y) = \max_{i \in p} \tmetric{i}(\X_i, \Y_i)$ then:
\begin{enumerate}[resume, label=\textbf{D\arabic*}]
\item $\aFun_{\A^{ep}}$ is strictly contracting on orbits over $\smetric$
\begin{align*}
& \forall \X \in \accSet: \aFun_{\A^{ep}}(\X) \neq \X \\ 
& \quad \Rightarrow \smetric(\X, \aFun_{\A^{ep}}(\X)) > \smetric(\aFun_{\A^{ep}}(\X), (\aFun_{\A^{ep}})^2(\X))
\end{align*}
\item $\aFun_{\A^{ep}}$ is strictly contracting on fixed points over $\smetric$:
\begin{align*}
& \forall \X \in \accSet, \X^* \in \mcarrier: \aFun_{\A^{ep}}(\X^*) = \X^* \wedge \X \neq \X^* \\ 
& \quad \Rightarrow \smetric(\X^*, \X) > \smetric(\X^*, \aFun_{\A^{ep}}(\X))
\end{align*}
\item $\aFun_{\A^{ep}}$ enforces accordancy:
\begin{equation*}
\forall \X \in \mcarrier: \aFun_{\A^{ep}}(\X) \in \accSet
\end{equation*}
\end{enumerate}
where $(\aFun_{\A^{ep}})^2$ is application of $\aFun_{\A^{ep}}$ twice.
\end{definition}

The properties required of an AMCO may seem complicated at first glance but each have clear interpretations. Firstly for every configuration in $C$ there must exist a function $\tmetric{i}$ which measures the dissimilarity between two routing table states (although it need not satisfy the formal definition of a metric). D1 requires that the dissimilarity between two states is zero iff the states are equal, and D2 that there exists a maximum dissimilarity. The notion of dissimilarity over routing tables, $\tmetric{i}$, is then naturally lifted to dissimilarity over routing states $\smetric$ by calculating the maximum dissimilarity between the participating routing tables. D3 requires that the dissimilarity between the states after consecutive iterations strictly decreases, and D4 that if there exists a stable state then the dissimilarity to it strictly decreases after each iteration. Finally, D5 requires that applying the operator to an accordant state must always result in another accordant state.

\begin{theorem}
\label{thm:amco-convergent}
If $F_{\A^{ep}}$ is a dynamic AMCO over a set of configurations $C$ then $\delta$ is convergent over $C$.
\end{theorem}
\begin{proof}
See Theorem 5 in \cite{daggitt2022dynamic} and the Agda proof~\cite{agda-routing}.
\end{proof}

Informally a synchronous version of the proof runs as follows: D3 requires that each synchronous iteration reduces the dissimilarity between consecutative states. Consequently, as the dissimilarity is a bounded natural number by D2, then the dissimilarity must eventually reach zero, which by D1 means that one has reached a stable state. D4 then guarantees that this stable state is unique. With some non-trivial work, it is possible to lift this synchronous argument to an asynchronous argument because the dissimilarity function over routing states is decomposable into individual dissimilarity functions over routing tables.

The key point to take away is that by using this theorem in our later proofs, we reduce the problem of proving our model of the DBF protocol, $\aStateFun$, is convergent to merely constructing a suitable dissimilarity function $\tmetric{i}$ which interacts with $\aFun_{\A^{ep}}$ in the right way.

\subsection{Types of routing algebras}
\label{sec:routing-algebra-classes}

In this section, we adapt the hierarchy of necessary and sufficient conditions for convergence established by Sobrinho~\cite{sobrinho2005algebraic} to our new simpler algebraic model\footnote{Our terminology differs slightly from that of Sobrinho as we use the traditional algebraic/order theoretic names for these conditions. In particular we use \emph{distributive} instead of \emph{isotonic} and \emph{increasing} instead of \emph{monotonic}. This is because the order theoretic notion of monotonicity, ${x \leq y \Rightarrow f(x) \leq f(y)}$, is much more closely related to the notion of distributivity, ${f(x \oplus y) = f(x) \oplus f(y)}$, than increasingness, ${x \leq f(x)}$.}. The key insight of Sobrinho is that each of the routing algebra axioms in Definition~\ref{def:routing-algebra2} only reference either choice or extension, and that is the relationship \emph{between} them that determines the convergence behaviour of the protocol.

As discussed in Section~\ref{sec:what-is-policy-rich-routing}, the gold-standard are \emph{distributive} algebras:
\begin{definition}[Distributive algebra]
A raw routing algebra is \emph{distributive} if:
\begin{equation*}
\forall f \in \allExtensions, x, y \in \carrier: f(x \oplus y) = f(x) \oplus f(y)
\end{equation*}
\end{definition}
\noindent As has been proven many times in the classic algebraic routing literature and, again by Sobrinho, if the algebra is distributive then any stable state will necessarily be a globally optimal state in which every pair of routers is assigned the optimal path between them.

The next most desirable property for the algebra to have is to be \emph{increasing}/\emph{strictly increasing}:
\begin{definition}[Increasing algebra]
A raw routing algebra is \emph{increasing} if:
\begin{equation*}
\forall f \in \allExtensions, x \in \carrier: x \leqWeight f(x)
\end{equation*}
\end{definition}

\begin{definition}[Strictly increasing algebra]
A raw routing algebra is \emph{strictly increasing} if:
\begin{equation*}
\forall f \in \allExtensions, x \in \carrier: x \neq \invalid \Rightarrow x \lessWeight f(x)
\end{equation*}
\end{definition}
\noindent Sobrinho proved that the algebra being strictly increasing is a necessary and sufficient condition for a path-vector protocol to converge to some stable state in a given epoch. However, if the algebra is not distributive then that stable state will only be locally optimal rather than globally optimal, and some pairs of routers will end up using suboptimal paths.

Finally there is the \emph{freeness} property over a particular network topology:
\begin{definition}[Free network topology]
A network topology $\A$ is free with respect to the raw routing algebra if there does not exist a cycle $[(v_1 , v_2), (v_2, v_3), ..., (v_m, v_1)]$ and a set of path-weights
$\{x_1 , x_2 , \ldots, x_{m} \}$ such that:
\begin{equation*}
\forall i : (v_{i-1}, x_{i - 1}) \threatenedBy{\A} (v_{i}, x_{i})
\end{equation*}
where $i - 1$ is calculated mod $m$.
\end{definition}
\noindent Intuitively a network is free if there does not exist a cycle where each router's current assignment can threaten the assignment of the previous router in the cycle. Sobrinho proved that the network topology being free with respect to the underlying algebra is a necessary and sufficient condition for a path-vector protocol to converge over that network, and that in the shortest-paths algebra this is equivalent to there being no negative weight cycles. Let $\freeConfigs$ be the set of configurations in which the participating topology is free, i.e.
\begin{equation*}
\freeConfigs = \{ (e, p) \in \configs \mid \A^{ep} \text{ is free} \} 
\end{equation*}

We now prove a couple of small lemmas showing some of the relationships between these properties.
\begin{lemma}
\label{lem:strictly-implies-free}
If a routing algebra is strictly increasing then every network topology is free.
\end{lemma}
\begin{proof}
See Theorem~2 \& subsequent discussion in \cite{sobrinho2005algebraic} or our Agda proof~\cite{agda-routing}.
\end{proof}

\begin{lemma}
\label{lem:incr-and-paths-imply-strictly-incr}
If a path algebra $\genericPathAlgebra$ is increasing then it is also strictly increasing.
\end{lemma}
\begin{proof}
Consider arbitrary $i, j \in \nodes$ and $f \in \extension_{ij}$ and $x \in \carrier$ such that ${x \neq \invalid}$. We already have that $x \leqWeight f(x)$ as the algebra is increasing so it remains to show that ${x \neq f(x)}$.

Suppose that $f(x) = x$. We then know that $f(x) \neq \invalid$ as $x \neq \invalid$. Likewise we know that $\pathf(x) \neq \invalidpath$ as otherwise by Assumption~\ref{ass:path-invalid} we would have $x = \invalid$. Therefore, by Assumption~\ref{ass:path-extension}, we have that $\pathf(x) = \pathf(f(x)) = (i,j) \cons \pathf(x)$ which is a contradiction. Therefore $x \neq f(x)$ and we have the required result.
\end{proof}

\subsection{DBF convergence results}
\label{sec:convergence-results}

We now proceed to state our main convergence results, the proofs of which may be found in the appendices. We start by considering distance-vector protocols.

\begin{theorem} 
\label{thm:dv-finite-free}
If $\genericRoutingAlgebra$ is a routing algebra and $\carrier$ is finite then $\delta$ is convergent over $\freeConfigs$.
\end{theorem}
\begin{proof}
See Appendix~\ref{app:distance_vector_proof}.
\end{proof}

\noindent Theorem~\ref{thm:dv-finite-free} says that for any distance-vector protocol with a finite set of path-weights then for any epoch and set of participating routers such that the network topology is free and given a sufficient period of stability in which the participating routers continue to activate and communicate, the protocol will always converge to the same stable state no matter the exact timing of messages and even in the presence of unreliable communication.

A more useful result for strictly increasing, finite algebras is then immediately obtainable which shows that such a distance-vector protocol is \emph{always} convergent.
\begin{theorem} 
\label{thm:dv-finite-strictly-increasing}
If $\genericRoutingAlgebra$ is a strictly increasing routing algebra and $\carrier$ is finite then $\delta$ is always convergent.
\end{theorem}
\begin{proof}
As the algebra is strictly increasing then by Lemma~\ref{lem:strictly-implies-free} every configuration is free, and therefore we have the required result immediately by  Theorem~\ref{thm:dv-finite-free}.
\end{proof} 

In practice the finiteness condition is restrictive and excludes many routing algebras of interest. For example even the shortest-path algebra uses the infinite set $\mathbb{N}$ as path weights. However recall that Theorems~\ref{thm:dv-finite-free}~\&~\ref{thm:dv-finite-strictly-increasing} guarantee that, given a sufficient period of stability, the protocol will reconverge even in the presence of messages from previous epochs. On the other hand shortest-path distance-vector protocols experience count-to-infinity problems when the state at the start of the epoch contains path-weights generated along paths that do not exist in the current topology. Finiteness is a sufficient condition to reduce count-to-infinity problems to merely count-to-convergence in distance-vector protocols. In the real world this is reflected in the design of RIP which artificially imposes a maximum hop count to ensure that $\carrier$ is finite.

A common approach to avoid the finiteness requirement (and count-to-convergence issues) is that of path-vector protocols which track the paths along which the path-weights are generated. Path-weights are then removed if their path contains a cycle.

\begin{theorem}
\label{thm:pv-free}
Given an algebra $\genericPathAlgebra$ which is a routing algebra and a path algebra then $\delta$ is convergent over $\freeConfigs$.
\end{theorem}
\begin{proof}
See Appendix~\ref{app:path-vector-proof}.
\end{proof}
\noindent Again a more useful version of this proof for strictly increasing path algebras is immediately obtainable.
\begin{theorem}
\label{thm:pv-strictly-increasing}
Given an increasing algebra $\genericPathAlgebra$ which is a routing algebra and a path algebra then then $\delta$ is always convergent.
\begin{proof}
As the path algebra is increasing then by Lemma~\ref{lem:incr-and-paths-imply-strictly-incr} it is also strictly increasing, and therefore by Lemma~\ref{lem:strictly-implies-free} every configuration is free. Hence we have the required result immediately by Theorem~\ref{thm:pv-free}.
\end{proof}
\end{theorem}

\subsection{Novel contributions}

We now recap the main contributions of these proofs over that of previous work.

\paragraph{Deterministic convergence}
Sobrinho's proofs only guarantee that the protocol will converge, and nothing in his proof forbids the stable state reached depending on the state at the start of the epoch or the ordering of messages between routers during the epoch. In contrast our proofs guarantee that the protocol always converges to the same stable state, no matter the initial state or what the ordering of messages is. This therefore eliminates the possibility of phenomena such as BGP wedgies. 

\paragraph{Unreliable communication}
The definition of the data-flow function, $\beta$, in Definition~\ref{def:schedule}, Section~\ref{sec:asynchronous-model} means that we only assume \emph{unreliable} communication between routers. This is in contrast to previous proofs of convergence, for both general DBF protocols~\cite{sobrinho2005algebraic} and BGP in particular~\cite{gao2001stable}, that assume reliable communication. This means that our model and resulting proofs in later sections in theory apply to protocols built on-top of unreliable transport protocols such as UDP~\cite{postel80udp}.

\paragraph{Extension to distance-vector protocols}
Theorems~\ref{thm:dv-finite-free}~\&~\ref{thm:dv-finite-strictly-increasing} apply to distance-vector protocols, whereas Sobrinho's proofs only applied to path-vector protocols. In particular, this shows that convergence could still be guaranteed if complex conditional policies were added to distance-vector protocols like RIP.

\section{Formalisation in Agda}
\label{sec:formalisation}

We have formalised every mathematical result in this paper, down to the most trivial lemma, in the Agda theorem prover. This includes the results about asynchronous iterative algorithms from~\cite{daggitt2022dynamic}. As our proofs have been checked by a computer, we are far more confident in their correctness than usual. Even where we omit details or use standard informal mathematical reasoning to improve readability (e.g. Lemma~\ref{lem:paths_str_contr_on_fixed_point} in Appendix~\ref{app:path-vector-proof}), the proofs are backed by mathematical arguments which are guaranteed to be fully rigorous.

Furthermore the act of formalisation itself was invaluable in creating and shaping these proofs. For example the unstructured way in which we initially laid out the pen-and-paper proof of Lemma~\ref{lem:paths_str_contr_on_orbits} in Appendix~\ref{app:path-vector-proof} led us to overlook Case~2.3. Only when formalising the result, did we notice that this case remained unproven. This in turn led us to presenting the proof in the much cleaner structure displayed in this paper.

The library of proofs is freely available~\cite{agda-routing} and is highly modular. To prototype a new policy language users need only define it as an algebra and prove that the algebra obeys the strictly increasing conditions in order to guarantee the resulting protocol is always convergent. Furthermore the protocol is then executable so that its dynamic behaviour can be observed for a given network and schedule. In Section~\ref{sec:bgp-lite-protocol} we give an outline of how to define such an algebra in Agda. We hope that the library's extensible nature means that it will be of use to the community when designing new policy languages.

\section{A safe-by-design algebra}
\label{sec:bgp-lite-protocol}

We now present an example of how one could use our Agda library to develop a safe-by-design routing protocol. In particular, we explore the construction of a path-vector algebra that contains many of the features of BGP such as local preferences, community values~\cite{rfc1997} and conditional policies. These policies can perform operations such as path-filtering, path-inflation, modifying local preferences and communities. The conditions themselves are implemented using a simple language of predicates that includes the ability to inspect communities. The algebra is a generalisation of the Stratified Shortest Paths algebra~\cite{Griffin2012}.

It differs from the algebra underlying today's BGP in two crucial ways. Firstly, BGP's implementation of the MED attribute violates the assumption that $\oplus$ is associative~\cite{griffin2002analysis}. Secondly, BGP allows ASs to hide their local preferences and set them to arbitrary values upon importing routes from other ASs. This violates the assumption that the algebra is increasing. 

Our algebra avoids these two issues by a) ignoring MED and b) having policies that only allow local preference to increase. For convenience, and to mirror the increasing terminology, in this algebra lower local preferences values are preferred over higher local preference values.

Note that we present our algebra only to give a practical example of how our theory can be used, and to show that ``most" of the features of BGP are inherently safe. We are not presenting it as a practical solution to these two problems in real-world BGP. We discuss the open question of whether hidden information is compatible with increasing algebras in Section~\ref{sec:hidden}.

Path-weights in our protocol are defined as follows:
\begin{code} \small \>[0]\AgdaKeyword{data}\AgdaSpace{}%
\AgdaDatatype{PathWeight}\AgdaSpace{}%
\AgdaSymbol{:}\AgdaSpace{}%
\AgdaPrimitiveType{Set}\AgdaSpace{}%
\AgdaKeyword{where}\<%
\\
\>[0][@{}l@{\AgdaIndent{0}}]%
\>[2] \AgdaInductiveConstructor{invalid}\AgdaSpace{}%
\>[10]\AgdaSymbol{:}\AgdaSpace{}%
\AgdaDatatype{PathWeight}\<%
\\
\>[2] \AgdaInductiveConstructor{valid}%
\AgdaSymbol{:}\AgdaSpace{}%
\AgdaFunction{LocalPref}\AgdaSpace{}%
\AgdaSymbol{→}\AgdaSpace{}%
\AgdaFunction{Communities}\AgdaSpace{}%
\AgdaSymbol{→}\AgdaSpace{}%
\AgdaDatatype{Path}\AgdaSpace{}%
\AgdaSymbol{→}\AgdaSpace{}%
\AgdaDatatype{PathWeight}\<%
\end{code}
\noindent i.e. there exists an invalid path-weight, and all other path-weights have a local preference, a set of communities and a path. The trivial path-weight, $\trivial$, is defined as:
\begin{code} \small \>[0]\AgdaFunction{0\#}\AgdaSpace{}%
\AgdaSpace{}%
\AgdaSymbol{:}\AgdaSpace{}\AgdaDatatype{PathWeight}\<%
\\
\>[0]\AgdaFunction{0\#}\AgdaSpace{}%
\AgdaSymbol{=}\AgdaSpace{}%
\AgdaInductiveConstructor{valid}\AgdaSpace{}%
\AgdaSymbol{(}\AgdaSymbol{$2 ^ {32} - 1$}\AgdaSymbol{)}\AgdaSpace{}%
\AgdaSymbol{∅}\AgdaSpace{}%
\AgdaSymbol{[]}\<%
\end{code}
where $2^{32} - 1$ is the highest possible local preference, $\emptyset$ is the empty set of communities and $[]$ is the empty path. The invalid path-weight, $\invalid$, is defined as:
\begin{code} \small \>[0]\AgdaFunction{∞\#}\AgdaSpace{}%
\AgdaSpace{}%
\AgdaSymbol{:}\AgdaSpace{}\AgdaDatatype{PathWeight}\<%
\\
\>[0]\AgdaFunction{∞\#}\AgdaSpace{}%
\AgdaSymbol{=}\AgdaSpace{}%
\AgdaInductiveConstructor{invalid}\<%
\end{code}
The choice operator, $\oplus$, is defined as: 
\begin{code} \small \>[0]\AgdaOperator{\AgdaFunction{\_⊕\_}}\AgdaSpace{}%
\AgdaSymbol{:}\AgdaSpace{}%
\AgdaFunction{Op₂}\AgdaSpace{}%
\AgdaDatatype{PathWeight}\<%
\\
\>[0]\AgdaBound{x}%
\>[30]\AgdaOperator{\AgdaFunction{⊕}}\AgdaSpace{}%
\AgdaBound{y}\AgdaSymbol{@(}\AgdaInductiveConstructor{invalid}\AgdaSymbol{)}%
\>[50]\AgdaSymbol{=}\AgdaSpace{}%
\AgdaBound{x}\<%
\\
\>[0]\AgdaBound{x}\AgdaSymbol{@(}\AgdaInductiveConstructor{invalid}\AgdaSymbol{)}%
\>[30]\AgdaOperator{\AgdaFunction{⊕}}\AgdaSpace{}%
\AgdaBound{y}%
\>[50]\AgdaSymbol{=}\AgdaSpace{}%
\AgdaBound{y}\<%
\\
\>[0]\AgdaBound{x}\AgdaSymbol{@(}\AgdaInductiveConstructor{valid}\AgdaSpace{}%
\AgdaBound{l}\AgdaSpace{}%
\AgdaBound{cs}\AgdaSpace{}%
\AgdaBound{p}\AgdaSymbol{)}\AgdaSpace{}%
\>[30]\AgdaOperator{\AgdaFunction{⊕}}\AgdaSpace{}%
\AgdaBound{y}\AgdaSymbol{@(}\AgdaInductiveConstructor{valid}\AgdaSpace{}%
\AgdaBound{m}\AgdaSpace{}%
\AgdaBound{ds}\AgdaSpace{}%
\AgdaBound{q}\AgdaSymbol{)}\AgdaSpace{}%
\>[50]\AgdaKeyword{with}\AgdaSpace{}%
\AgdaFunction{compare}\AgdaSpace{}%
\AgdaBound{l}\AgdaSpace{}%
\AgdaBound{m}\<%
\\
\>[0]\AgdaSymbol{...}\AgdaSpace{}%
\AgdaSymbol{|}\AgdaSpace{}%
\AgdaInductiveConstructor{tri>}\AgdaSpace{}%
\AgdaBound{l>m}%
\>[20]\AgdaSymbol{=}\AgdaSpace{}%
\AgdaBound{x}\<%
\\
\>[0]\AgdaSymbol{...}\AgdaSpace{}%
\AgdaSymbol{|}\AgdaSpace{}%
\AgdaInductiveConstructor{tri<}\AgdaSpace{}%
\AgdaBound{l<m}\AgdaSpace{}%
\>[20]\AgdaSymbol{=}\AgdaSpace{}%
\AgdaBound{y}\<%
\\
\>[0]\AgdaSymbol{...}\AgdaSpace{}%
\AgdaSymbol{|}\AgdaSpace{}%
\AgdaInductiveConstructor{tri≈}\AgdaSpace{}%
\AgdaBound{l=m}\AgdaSpace{}%
\>[20]\AgdaKeyword{with}\AgdaSpace{}%
\AgdaFunction{compare}\AgdaSpace{}%
\AgdaSymbol{(}\AgdaFunction{length}\AgdaSpace{}%
\AgdaBound{p}\AgdaSymbol{)}\AgdaSpace{}%
\AgdaSymbol{(}\AgdaFunction{length}\AgdaSpace{}%
\AgdaBound{q}\AgdaSymbol{)}\<%
\\
\>[0]\AgdaSymbol{...}%
\>[15]\agdaIndent{}\AgdaSymbol{|}\AgdaSpace{}%
\AgdaInductiveConstructor{tri<}\AgdaSpace{}%
\AgdaBound{|p|<|q|}\AgdaSpace{}%
\>[26]\AgdaSymbol{=}\AgdaSpace{}%
\AgdaBound{x}\<%
\\
\>[0]\AgdaSymbol{...}%
\>[15]\agdaIndent{}\AgdaSymbol{|}\AgdaSpace{}%
\AgdaInductiveConstructor{tri>}\AgdaSpace{}%
\AgdaBound{|p|>|q|}%
\>[26]\AgdaSymbol{=}\AgdaSpace{}%
\AgdaBound{y}\<%
\\
\>[0]\AgdaSymbol{...}%
\>[15]\agdaIndent{}\AgdaSymbol{|}\AgdaSpace{}%
\AgdaInductiveConstructor{tri≈}\AgdaSpace{}%
\AgdaBound{|p|=|q|}\AgdaSpace{}%
\>[26]\AgdaKeyword{with}\AgdaSpace{}%
\AgdaBound{p}\AgdaSpace{}%
\AgdaOperator{\AgdaFunction{≤ₗₑₓ?}}\AgdaSpace{}%
\AgdaBound{q}\<%
\\
\>[0]\AgdaSymbol{...}%
\>[8]\agdaIndent{}\agdaIndent{}\AgdaSymbol{|}\AgdaSpace{}%
\AgdaInductiveConstructor{yes}\AgdaSpace{}%
\AgdaBound{p≤q}\AgdaSpace{}%
\AgdaSymbol{=}\AgdaSpace{}%
\AgdaBound{x}\<%
\\
\>[0]\AgdaSymbol{...}%
\>[8]\agdaIndent{}\agdaIndent{}\AgdaSymbol{|}\AgdaSpace{}%
\AgdaInductiveConstructor{no}%
\>[14]\AgdaBound{q≤p}\AgdaSpace{}%
\AgdaSymbol{=}\AgdaSpace{}%
\AgdaBound{y}\<%
\end{code}
which operates as follows:
\begin{enumerate}
\item If either $x$ or $y$ is invalid return the other.
\item otherwise if the local pref of either $x$ or $y$ is strictly greater than the other return that path-weight.
\item otherwise if the length of the path of either $x$ or $y$ is strictly less than the other return that path-weight.
\item finally break ties by a lexicographic comparison of paths.
\end{enumerate}

Next we construct the set of policy functions $\extension_{ij}$. We start by defining a simple yet expressive language for conditions that can be used by our policy language to make decisions.
\begin{code} \small \>[0]\AgdaKeyword{data}\AgdaSpace{}%
\AgdaDatatype{Condition}\AgdaSpace{}%
\AgdaSymbol{:}\AgdaSpace{}%
\AgdaPrimitiveType{Set}\AgdaSpace{}%
\AgdaKeyword{where}\<%
\\
\>[0][@{}l@{\AgdaIndent{0}}]%
\>[2]\AgdaInductiveConstructor{\_and\_}%
\>[11]\AgdaSymbol{:}\AgdaSpace{}%
\AgdaDatatype{Condition}%
\>[24]\AgdaSymbol{→}\AgdaSpace{}%
\AgdaDatatype{Condition}%
\>[37]\AgdaSymbol{→}\AgdaSpace{}%
\AgdaDatatype{Condition}\<%
\\
\>[2]\AgdaInductiveConstructor{\_or\_}%
\>[11]\AgdaSymbol{:}\AgdaSpace{}%
\AgdaDatatype{Condition}%
\>[24]\AgdaSymbol{→}\AgdaSpace{}%
\AgdaDatatype{Condition}%
\>[37]\AgdaSymbol{→}\AgdaSpace{}%
\AgdaDatatype{Condition}\<%
\\
\>[2]\AgdaInductiveConstructor{not}%
\>[11]\AgdaSymbol{:}\AgdaSpace{}%
\AgdaDatatype{Condition}%
\>[24]\AgdaSymbol{→}\AgdaSpace{}%
\AgdaDatatype{Condition}\<%
\\
\>[2]\AgdaInductiveConstructor{inPath}%
\>[11]\AgdaSymbol{:}\AgdaSpace{}%
\AgdaDatatype{Node}%
\>[24]\AgdaSymbol{→}\AgdaSpace{}%
\AgdaDatatype{Condition}\<%
\\
\>[2]\AgdaInductiveConstructor{inComm}%
\>[11]\AgdaSymbol{:}\AgdaSpace{}%
\AgdaFunction{Community}%
\>[24]\AgdaSymbol{→}\AgdaSpace{}%
\AgdaDatatype{Condition}\<%
\\
\>[2]\AgdaInductiveConstructor{hasPref}%
\>[11]\AgdaSymbol{:}\AgdaSpace{}%
\AgdaFunction{LocalPref}%
\>[24]\AgdaSymbol{→}\AgdaSpace{}%
\AgdaDatatype{Condition}\<%
\end{code}
The semantics for these conditionals is defined as follows:
\begin{code} \small \>[0]\AgdaFunction{eval}\AgdaSpace{}%
\AgdaSymbol{:}\AgdaSpace{}%
\AgdaDatatype{Condition}\AgdaSpace{}%
\AgdaSymbol{→}\AgdaSpace{}%
\AgdaDatatype{PathWeight}\AgdaSpace{}%
\AgdaSymbol{→}\AgdaSpace{}%
\AgdaDatatype{Bool}\<%
\\
\>[0]\AgdaFunction{eval}\AgdaSpace{}%
\AgdaSymbol{(}\AgdaBound{p}\AgdaSpace{}\AgdaInductiveConstructor{and}\AgdaSpace{}%
\AgdaBound{q}\AgdaSymbol{)}%
\>[21]\AgdaBound{x}%
\>[36]\AgdaSymbol{=}\AgdaSpace{}%
\AgdaFunction{eval}\AgdaSpace{}%
\AgdaBound{p}\AgdaSpace{}%
\AgdaBound{x}\AgdaSpace{}%
\AgdaOperator{\AgdaFunction{∧}}\AgdaSpace{}%
\AgdaFunction{eval}\AgdaSpace{}%
\AgdaBound{q}\AgdaSpace{}%
\AgdaBound{x}\<%
\\
\>[0]\AgdaFunction{eval}\AgdaSpace{}%
\AgdaSymbol{(}\AgdaBound{p}\AgdaSpace{}\AgdaInductiveConstructor{or}\AgdaSpace{}%
\AgdaBound{q}\AgdaSymbol{)}%
\>[21]\AgdaBound{x}%
\>[36]\AgdaSymbol{=}\AgdaSpace{}%
\AgdaFunction{eval}\AgdaSpace{}%
\AgdaBound{p}\AgdaSpace{}%
\AgdaBound{x}\AgdaSpace{}%
\AgdaOperator{\AgdaFunction{∨}}\AgdaSpace{}%
\AgdaFunction{eval}\AgdaSpace{}%
\AgdaBound{q}\AgdaSpace{}%
\AgdaBound{x}\<%
\\
\>[0]\AgdaFunction{eval}\AgdaSpace{}%
\AgdaSymbol{(}\AgdaInductiveConstructor{not}\AgdaSpace{}%
\AgdaBound{p}\AgdaSymbol{)}%
\>[21]\AgdaBound{x}%
\>[36]\AgdaSymbol{=}\AgdaSpace{}%
\AgdaFunction{not}\AgdaSpace{}%
\AgdaSymbol{(}\AgdaFunction{eval}\AgdaSpace{}%
\AgdaBound{p}\AgdaSpace{}%
\AgdaBound{x}\AgdaSymbol{)}\<%
\\
\>[0]\AgdaFunction{eval}\AgdaSpace{}%
\AgdaSymbol{(}\AgdaInductiveConstructor{inComm}%
\>[18]\AgdaBound{c}\AgdaSymbol{)}\AgdaSpace{}%
\>[21]\AgdaSymbol{(}\AgdaInductiveConstructor{valid}\AgdaSpace{}%
\AgdaBound{l}\AgdaSpace{}%
\AgdaBound{cs}\AgdaSpace{}%
\AgdaBound{p}\AgdaSymbol{)}%
\>[36]\AgdaSymbol{=}\AgdaSpace{}%
\AgdaOperator{\AgdaFunction{⌊}}\AgdaSpace{}%
\AgdaBound{c}\AgdaSpace{}%
\AgdaOperator{\AgdaFunction{∈?}}\AgdaSpace{}%
\AgdaBound{cs}\AgdaSpace{}%
\AgdaOperator{\AgdaFunction{⌋}}\<%
\\
\>[0]\AgdaFunction{eval}\AgdaSpace{}%
\AgdaSymbol{(}\AgdaInductiveConstructor{inPath}%
\>[18]\AgdaBound{i}\AgdaSymbol{)}\AgdaSpace{}%
\>[21]\AgdaSymbol{(}\AgdaInductiveConstructor{valid}\AgdaSpace{}%
\AgdaBound{l}\AgdaSpace{}%
\AgdaBound{cs}\AgdaSpace{}%
\AgdaBound{p}\AgdaSymbol{)}%
\>[36]\AgdaSymbol{=}\AgdaSpace{}%
\AgdaFunction{⌊}\AgdaSpace{}%
\AgdaBound{i}\AgdaSpace{}%
\AgdaOperator{\AgdaFunction{∈?}}\AgdaSpace{}%
\AgdaBound{p}\AgdaSpace{}%
\AgdaOperator{\AgdaFunction{⌋}}\<%
\\
\>[0]\AgdaFunction{eval}\AgdaSpace{}%
\AgdaSymbol{(}\AgdaInductiveConstructor{hasPref}%
\>[18]\AgdaBound{v}\AgdaSymbol{)}\AgdaSpace{}%
\>[21]\AgdaSymbol{(}\AgdaInductiveConstructor{valid}\AgdaSpace{}%
\AgdaBound{l}\AgdaSpace{}%
\AgdaBound{cs}\AgdaSpace{}%
\AgdaBound{p}\AgdaSymbol{)}%
\>[36]\AgdaSymbol{=}\AgdaSpace{}%
\AgdaFunction{⌊}\AgdaSpace{}%
\AgdaBound{v}\AgdaSpace{}%
\AgdaOperator{\AgdaFunction{≟}}\AgdaSpace{}%
\AgdaBound{l}\AgdaSpace{}%
\AgdaOperator{\AgdaFunction{⌋}}\<%
\\
\>[0]\AgdaFunction{eval}\AgdaSpace{}%
\AgdaSymbol{\_}%
\>[21]\AgdaInductiveConstructor{invalid}%
\>[36]\AgdaSymbol{=}\AgdaSpace{}%
\AgdaInductiveConstructor{false}\<%
\end{code}
We next define the policy language as follows:
\begin{code} \small \>[0]\AgdaKeyword{data}\AgdaSpace{}%
\AgdaDatatype{Policy}\AgdaSpace{}%
\AgdaSymbol{:}\AgdaSpace{}%
\AgdaPrimitiveType{Set₁}\AgdaSpace{}%
\AgdaKeyword{where}\<%
\\
\>[0][@{}l@{\AgdaIndent{0}}]%
\>[2]\AgdaInductiveConstructor{reject}%
\>[13]\AgdaSymbol{:}\AgdaSpace{}%
\AgdaDatatype{Policy}\<%
\\
\>[2]\AgdaInductiveConstructor{decrPrefBy}%
\>[13]\AgdaSymbol{:}\AgdaSpace{}%
\AgdaDatatype{ℕ}\AgdaSpace{}%
\AgdaSymbol{→}\AgdaSpace{}%
\AgdaDatatype{Policy}\<%
\\
\>[2]\AgdaInductiveConstructor{addComm}%
\>[13]\AgdaSymbol{:}\AgdaSpace{}%
\AgdaFunction{Community}\AgdaSpace{}%
\AgdaSymbol{→}\AgdaSpace{}%
\AgdaDatatype{Policy}\<%
\\
\>[2]\AgdaInductiveConstructor{delComm}%
\>[13]\AgdaSymbol{:}\AgdaSpace{}%
\AgdaFunction{Community}\AgdaSpace{}%
\AgdaSymbol{→}\AgdaSpace{}%
\AgdaDatatype{Policy}\<%
\\
\>[2]\AgdaInductiveConstructor{inflate}%
\>[13]\AgdaSymbol{:}\AgdaSpace{}%
\AgdaDatatype{ℕ}\AgdaSpace{}%
\AgdaSymbol{→}\AgdaSpace{}%
\AgdaDatatype{Policy}\<%
\\
\>[2]\AgdaInductiveConstructor{\_;\_}%
\>[13]\AgdaSymbol{:}\AgdaSpace{}%
\AgdaDatatype{Policy}\AgdaSpace{}%
\AgdaSymbol{→}\AgdaSpace{}%
\AgdaDatatype{Policy}\AgdaSpace{}%
\AgdaSymbol{→}\AgdaSpace{}%
\AgdaDatatype{Policy}\<%
\\
\>[2]\AgdaInductiveConstructor{if\_then\_}%
\>[13]\AgdaSymbol{:}\AgdaSpace{}%
\AgdaDatatype{Condition}\AgdaSpace{}%
\AgdaSymbol{→}\AgdaSpace{}%
\AgdaDatatype{Policy}\AgdaSpace{}%
\AgdaSymbol{→}\AgdaSpace{}%
\AgdaDatatype{Policy}\<%
\end{code}
The semantics of each type of policy are defined by the function that applies policies to path-weights:
\begin{code} \small \>[0]\AgdaFunction{apply}\AgdaSpace{}%
\AgdaSymbol{:}\AgdaSpace{}%
\AgdaDatatype{Policy}\AgdaSpace{}%
\AgdaSymbol{→}\AgdaSpace{}%
\AgdaDatatype{PathWeight}\AgdaSpace{}%
\AgdaSymbol{→}\AgdaSpace{}%
\AgdaDatatype{PathWeight}\<%
\\
\>[0]\AgdaFunction{apply}\AgdaSpace{}%
\AgdaSymbol{p}%
\>[21]\AgdaInductiveConstructor{invalid}%
\>[37]\AgdaSymbol{=}\AgdaSpace{}%
\AgdaInductiveConstructor{invalid}\<%
\\
\>[0]\AgdaFunction{apply}\AgdaSpace{}%
\AgdaInductiveConstructor{reject}%
\>[21]\AgdaSymbol{x}%
\>[37]\AgdaSymbol{=}\AgdaSpace{}%
\AgdaInductiveConstructor{invalid}\<%
\\
\>[0]\AgdaFunction{apply}\AgdaSpace{}%
\AgdaSymbol{(}\AgdaInductiveConstructor{decrPrefBy}\AgdaSpace{}%
\AgdaBound{v}\AgdaSymbol{)}%
\>[21]\AgdaSymbol{(}\AgdaInductiveConstructor{valid}\AgdaSpace{}%
\AgdaBound{l}\AgdaSpace{}%
\AgdaBound{cs}\AgdaSpace{}%
\AgdaBound{p}\AgdaSymbol{)}%
\>[37]\AgdaSymbol{=}\AgdaSpace{}%
\AgdaInductiveConstructor{valid}\AgdaSpace{}%
\AgdaSymbol{(}\AgdaBound{l}\AgdaSpace{}%
\AgdaPrimitive{-}\AgdaSpace{}%
\AgdaBound{v}\AgdaSymbol{)}\AgdaSpace{}%
\AgdaBound{cs}\AgdaSpace{}%
\AgdaBound{p}\<%
\\
\>[0]\AgdaFunction{apply}\AgdaSpace{}%
\AgdaSymbol{(}\AgdaInductiveConstructor{addComm}\AgdaSpace{}%
\AgdaBound{c}\AgdaSymbol{)}%
\>[21]\AgdaSymbol{(}\AgdaInductiveConstructor{valid}\AgdaSpace{}%
\AgdaBound{l}\AgdaSpace{}%
\AgdaBound{cs}\AgdaSpace{}%
\AgdaBound{p}\AgdaSymbol{)}%
\>[37]\AgdaSymbol{=}\AgdaSpace{}%
\AgdaInductiveConstructor{valid}\AgdaSpace{}%
\AgdaBound{l}\AgdaSpace{}%
\AgdaSymbol{(}\AgdaFunction{add}%
\AgdaSpace{}\AgdaBound{c}\AgdaSpace{}%
\AgdaBound{cs}\AgdaSymbol{)}\AgdaSpace{}%
\AgdaBound{p}\<%
\\
\>[0]\AgdaFunction{apply}\AgdaSpace{}%
\AgdaSymbol{(}\AgdaInductiveConstructor{delComm}\AgdaSpace{}%
\AgdaBound{c}\AgdaSymbol{)}%
\>[21]\AgdaSymbol{(}\AgdaInductiveConstructor{valid}\AgdaSpace{}%
\AgdaBound{l}\AgdaSpace{}%
\AgdaBound{cs}\AgdaSpace{}%
\AgdaBound{p}\AgdaSymbol{)}%
\>[37]\AgdaSymbol{=}\AgdaSpace{}%
\AgdaInductiveConstructor{valid}\AgdaSpace{}%
\AgdaBound{l}\AgdaSpace{}%
\AgdaSymbol{(}\AgdaFunction{remove}%
\>[56]\AgdaBound{c}\AgdaSpace{}%
\AgdaBound{cs}\AgdaSymbol{)}\AgdaSpace{}%
\AgdaBound{p}\<%
\\
\>[0]\AgdaFunction{apply}\AgdaSpace{}%
\AgdaSymbol{(}\AgdaInductiveConstructor{inflate}\AgdaSpace{}%
\AgdaBound{n}\AgdaSymbol{)}%
\>[21]\AgdaSymbol{(}\AgdaInductiveConstructor{valid}\AgdaSpace{}%
\AgdaBound{l}\AgdaSpace{}%
\AgdaBound{cs}\AgdaSpace{}%
\AgdaBound{p}\AgdaSymbol{)}%
\>[37]\AgdaSymbol{=}\AgdaSpace{}%
\AgdaInductiveConstructor{valid}\AgdaSpace{}%
\AgdaBound{l}\AgdaSpace{}%
\AgdaBound{cs}\AgdaSpace{}%
\AgdaSymbol{(}\AgdaFunction{inflate}\AgdaSpace{}%
\AgdaBound{p}\AgdaSpace{}%
\AgdaBound{n}\AgdaSymbol{)}\<%
\\
\>[0]\AgdaFunction{apply}\AgdaSpace{}%
\AgdaSymbol{(}\AgdaBound{p}\AgdaSpace{}%
\AgdaInductiveConstructor{;}\AgdaSpace{}%
\AgdaBound{q}\AgdaSymbol{)}%
\>[21]\AgdaBound{x}%
\>[37]\AgdaSymbol{=}\AgdaSpace{}%
\AgdaFunction{apply}\AgdaSpace{}%
\AgdaBound{q}\AgdaSpace{}%
\AgdaSymbol{(}\AgdaFunction{apply}\AgdaSpace{}%
\AgdaBound{p}\AgdaSpace{}%
\AgdaBound{x}\AgdaSymbol{)}\<%
\\
\>[0]\AgdaFunction{apply}\AgdaSpace{}%
\AgdaSymbol{(}\AgdaInductiveConstructor{if}\AgdaSpace{}%
\AgdaBound{c}\AgdaSpace{}%
\AgdaInductiveConstructor{then}\AgdaSpace{}%
\AgdaBound{p}\AgdaSymbol{)}%
\>[21]\AgdaBound{x}%
\>[37]\AgdaSymbol{=}\AgdaSpace{}%
\AgdaFunction{if}%
\AgdaSpace{}\AgdaSymbol{(}%
\AgdaFunction{eval}\AgdaSpace{}%
\AgdaBound{c}\AgdaSpace{}%
\AgdaBound{x}\AgdaSymbol{)}\<%
\\
\>[10]\AgdaSpace{}\AgdaSpace{}\AgdaSpace{}\AgdaSpace{}
\AgdaSpace{}%
\AgdaFunction{then}\AgdaSpace{}%
\AgdaSymbol{(}\AgdaFunction{apply}\AgdaSpace{}%
\AgdaBound{p}\AgdaSpace{}%
\AgdaBound{x}\AgdaSymbol{)}\AgdaSpace{}%
\AgdaFunction{else}\AgdaSpace{}%
\AgdaBound{x}\<%
\end{code}
where the \AgdaFunction{inflate} function prepends $n$ copies of the path's source on the front of the path. Note that since the policy language provides no way of increasing a path-weight's local preference or decrease the length of its path, it is not possible to define a non-increasing policy.

Agda is based on type-theory rather than set-theory, and therefore unfortunately we cannot define a set of functions directly. Instead we define the type of policy functions $\extension_{ij}$ for each edge $(i,j)$. Concretely there exists such a policy function for every policy:
\begin{code} \small \>[0]\AgdaKeyword{data}\AgdaSpace{}%
\AgdaDatatype{PolicyFunction}\AgdaSpace{}%
\AgdaSymbol{\{}\AgdaBound{n}\AgdaSymbol{\}}\AgdaSpace{}%
\AgdaSymbol{(}\AgdaBound{i}\AgdaSpace{}%
\AgdaBound{j}\AgdaSpace{}%
\AgdaSymbol{:}\AgdaSpace{}%
\AgdaDatatype{Node}\AgdaSpace{}%
\AgdaBound{n}\AgdaSymbol{)}\AgdaSpace{}%
\AgdaSymbol{:}\AgdaSpace{}%
\AgdaPrimitiveType{Set₁}\AgdaSpace{}%
\AgdaKeyword{where}\<%
\\
\>[0][@{}l@{\AgdaIndent{0}}]%
\>[2]\AgdaInductiveConstructor{ext}\AgdaSpace{}%
\AgdaSymbol{:}\AgdaSpace{}%
\AgdaDatatype{Policy}\AgdaSpace{}%
\AgdaSymbol{→}\AgdaSpace{}%
\AgdaDatatype{Extension}\AgdaSpace{}%
\AgdaBound{i}\AgdaSpace{}%
\AgdaBound{j}\<%
\end{code}
and the application of a policy function is defined as:
\begin{code} \small \>[0]\AgdaOperator{\AgdaFunction{\_▷\_}}\AgdaSpace{}%
\AgdaSymbol{:}\AgdaSpace{}%
\AgdaDatatype{PolicyFunction}\AgdaSpace{}%
\AgdaBound{i}\AgdaSpace{}%
\AgdaBound{j}\AgdaSpace{}%
\AgdaSymbol{→}\AgdaSpace{}%
\AgdaDatatype{PathWeight}\AgdaSpace{}%
\AgdaSymbol{→}\AgdaSpace{}%
\AgdaDatatype{PathWeight}\<%
\\
\>[0]\AgdaSymbol{(}\AgdaInductiveConstructor{ext}\AgdaSpace{}%
\AgdaBound{pol}\AgdaSymbol{)}\AgdaSpace{}%
\AgdaOperator{\AgdaFunction{▷}}\AgdaSpace{}%
\>[28]\AgdaInductiveConstructor{invalid}%
\>[43]\AgdaSymbol{=}\AgdaSpace{}%
\AgdaInductiveConstructor{invalid}\<%
\\
\>[0]\AgdaSymbol{(}\AgdaInductiveConstructor{ext}\AgdaSpace{}%
\AgdaBound{pol}\AgdaSymbol{)}\AgdaSpace{}%
\AgdaOperator{\AgdaFunction{▷}}\AgdaSpace{}%
\>[28]\AgdaSymbol{(}\AgdaInductiveConstructor{valid}\AgdaSpace{}%
\AgdaBound{l}\AgdaSpace{}%
\AgdaBound{cs}\AgdaSpace{}%
\AgdaBound{p}\AgdaSymbol{)}%
\>[43]\AgdaKeyword{with}\AgdaSpace{}%
\AgdaSymbol{(}%
\AgdaBound{i}\AgdaSpace{}%
\AgdaOperator{\AgdaInductiveConstructor{,}}\AgdaSpace{}%
\AgdaBound{j}\AgdaSymbol{)}\AgdaSpace{}%
\AgdaOperator{\AgdaFunction{⇿?}}\AgdaSpace{}%
\AgdaBound{p}\AgdaSpace{}%
\AgdaSymbol{|}\AgdaSpace{}%
\AgdaBound{i}\AgdaSpace{}%
\AgdaOperator{\AgdaFunction{∈?}}\AgdaSpace{}%
\AgdaBound{p}\<%
\\
\>[0]\AgdaSymbol{...}\AgdaSpace{}%
\AgdaSymbol{|}\AgdaSpace{}%
\AgdaInductiveConstructor{no}%
\>[10]\AgdaSymbol{\_}%
\>[16]\AgdaSymbol{|}\AgdaSpace{}%
\AgdaSymbol{\_}%
\>[27]\AgdaSymbol{=}\AgdaSpace{}%
\AgdaInductiveConstructor{invalid}\<%
\\
\>[0]\AgdaSymbol{...}\AgdaSpace{}%
\AgdaSymbol{|}\AgdaSpace{}%
\AgdaInductiveConstructor{yes}\AgdaSpace{}%
\AgdaSymbol{\_}%
\>[16]\AgdaSymbol{|}\AgdaSpace{}%
\AgdaInductiveConstructor{yes}\AgdaSpace{}%
\AgdaSymbol{\_}%
\>[27]\AgdaSymbol{=}\AgdaSpace{}%
\AgdaInductiveConstructor{invalid}\<%
\\
\>[0]\AgdaSymbol{...}\AgdaSpace{}%
\AgdaSymbol{|}\AgdaSpace{}%
\AgdaInductiveConstructor{yes}\AgdaSpace{}%
\AgdaBound{ij⇿p}%
\>[16]\AgdaSymbol{|}\AgdaSpace{}%
\AgdaInductiveConstructor{no}%
\>[22]\AgdaBound{i∈p}%
\>[27]\AgdaSymbol{=}\AgdaSpace{}%
\AgdaFunction{apply}\AgdaSpace{}%
\AgdaBound{pol}\AgdaSpace{}%
\AgdaSymbol{(}\AgdaInductiveConstructor{valid}\AgdaSpace{}%
\AgdaBound{l}\AgdaSpace{}%
\AgdaBound{cs}\AgdaSpace{}%
\AgdaSymbol{((}%
\AgdaBound{i}\AgdaSpace{}%
\AgdaOperator{\AgdaInductiveConstructor{,}}\AgdaSpace{}%
\AgdaBound{j}\AgdaSymbol{)}\AgdaSpace{}%
\AgdaOperator{\AgdaInductiveConstructor{∷}}\AgdaSpace{}%
\AgdaBound{p}\AgdaSymbol{))}\<%
\end{code}
where \AgdaSymbol{(}\AgdaBound{i}\AgdaSpace{}%
\AgdaInductiveConstructor{,}\AgdaSpace{}%
\AgdaBound{j}\AgdaSymbol{)}\AgdaSpace{}%
\AgdaFunction{⇿?}\AgdaSpace{}%
\AgdaBound{p} tests if the edge $(i,j)$ is a valid extension of path $p$ (i.e. if $j = src(p)$), and \AgdaBound{i}\AgdaSpace{}%
\AgdaFunction{∉?}\AgdaSpace{}%
\AgdaBound{p} tests whether or not~$i$ already exists in $p$ (i.e. if the resulting path would loop).

\matthewrevision{
The constantly invalid functions are those that always use the \AgdaInductiveConstructor{reject} policy:
\begin{code} \small \>[0]\AgdaFunction{f∞}\AgdaSpace{}%
\AgdaSymbol{:}\AgdaSpace{}%
\AgdaSymbol{∀}\AgdaSpace{}%
\AgdaSymbol{(}\AgdaBound{i}\AgdaSpace{}%
\AgdaBound{j}\AgdaSpace{}%
\AgdaSymbol{:}\AgdaSpace{}%
\AgdaDatatype{Fin}\AgdaSpace{}%
\AgdaGeneralizable{n}\AgdaSymbol{)}\AgdaSpace{}%
\AgdaSymbol{→}\AgdaSpace{}%
\AgdaDatatype{Extension}\AgdaSpace{}%
\AgdaBound{i}\AgdaSpace{}%
\AgdaBound{j}\<%
\\
\>[0]\AgdaFunction{f∞}\AgdaSpace{}%
\AgdaBound{i}\AgdaSpace{}%
\AgdaBound{j}\AgdaSpace{}%
\AgdaSymbol{=}\AgdaSpace{}%
\AgdaInductiveConstructor{ext}\AgdaSpace{}%
\AgdaInductiveConstructor{reject}\<%
\end{code}}
The $path$ function from path-weights to simple paths required by a path algebra can be defined as:
\begin{code} \small \>[0]\AgdaFunction{path}\AgdaSpace{}%
\AgdaSymbol{:}\AgdaSpace{}%
\AgdaDatatype{PathWeight}\AgdaSpace{}%
\AgdaSymbol{→}\AgdaSpace{}%
\AgdaDatatype{Path}\AgdaSpace{}\<%
\\
\>[0]\AgdaFunction{path}\AgdaSpace{}%
\AgdaInductiveConstructor{invalid}%
\>[19]\AgdaSymbol{=}\AgdaSpace{}%
\AgdaInductiveConstructor{⊥}\<%
\\
\>[0]\AgdaFunction{path}\AgdaSpace{}%
\AgdaSymbol{(}\AgdaInductiveConstructor{valid}\AgdaSpace{}%
\AgdaSymbol{\_}\AgdaSpace{}%
\AgdaSymbol{\_}\AgdaSpace{}%
\AgdaBound{p}\AgdaSymbol{)}\AgdaSpace{}%
\>[19]\AgdaSymbol{=}\AgdaSpace{}%
\AgdaFunction{deflate}\AgdaSpace{}\AgdaBound{p}\<%
\end{code}
where \AgdaFunction{deflate} strips from the path any consecutive duplicate routers that might have been introduced by \AgdaFunction{inflate}. The reader might have noted at this point that all our proofs in this paper work over simple (i.e. loop free) paths, yet this algebra internally stores non-simple inflated paths. However this apparent contradiction is resolved by observing that, as was discussed in Section~\ref{sec:model-contributions},  the definition of a path algebra only requires that we can extract a simple path from a path-weight, rather than requiring that the path-weight stores a simple path. This is yet another demonstration of the power of the path algebra abstraction. 

\matthewrevision{By combining the primitives together, we can now construct a \AgdaFunction{RawRoutingAlgebra} object for the protocol:
\begin{code} \small \>[0]\AgdaFunction{bgpLite}\AgdaSpace{}%
\AgdaSymbol{:}\AgdaSpace{}%
\AgdaRecord{RawRoutingAlgebra}\AgdaSpace{}%
\AgdaSymbol{\AgdaUnderscore{}}\AgdaSpace{}%
\AgdaSymbol{\AgdaUnderscore{}}\AgdaSpace{}%
\AgdaSymbol{\AgdaUnderscore{}}\<%
\\
\>[0]\AgdaFunction{bgpLite}%
\>[359I]\AgdaSymbol{=}\AgdaSpace{}%
\AgdaKeyword{record}\<%
\\
\>[.][@{}l@{}]%
\>[2]\AgdaSpace{}\AgdaSpace{}\AgdaSpace{}\AgdaSpace{}\AgdaSymbol{\{}\AgdaSpace{}%
\AgdaField{PathWeight}%
\>[23]\AgdaSymbol{=}\AgdaSpace{}%
\AgdaDatatype{PathWeight}\<%
\\
\>[2]\AgdaSpace{}\AgdaSpace{}\AgdaSpace{}\AgdaSpace{}\AgdaSymbol{;}\AgdaSpace{}%
\AgdaField{Step}%
\>[23]\AgdaSymbol{=}\AgdaSpace{}%
\AgdaDatatype{Extension}\<%
\\
%
%\>[2]\AgdaSymbol{;}\AgdaSpace{}%
%\AgdaOperator{\AgdaField{\AgdaUnderscore{}≈\AgdaUnderscore{}}}%
%\>[23]\AgdaSymbol{=}\AgdaSpace{}%
%\AgdaOperator{\AgdaDatatype{\AgdaUnderscore{}≡\AgdaUnderscore{}}}\<%
%\\
%
\>[2]\AgdaSpace{}\AgdaSpace{}\AgdaSpace{}\AgdaSpace{}\AgdaSymbol{;}\AgdaSpace{}%
\AgdaOperator{\AgdaField{\AgdaUnderscore{}⊕\AgdaUnderscore{}}}%
\>[23]\AgdaSymbol{=}\AgdaSpace{}%
\AgdaOperator{\AgdaFunction{\AgdaUnderscore{}⊕\AgdaUnderscore{}}}\<%
\\
\>[2]\AgdaSpace{}\AgdaSpace{}\AgdaSpace{}\AgdaSpace{}\AgdaSymbol{;}\AgdaSpace{}%
\AgdaOperator{\AgdaField{\AgdaUnderscore{}▷\AgdaUnderscore{}}}%
\>[23]\AgdaSymbol{=}\AgdaSpace{}%
\AgdaOperator{\AgdaFunction{\AgdaUnderscore{}▷\AgdaUnderscore{}}}\<%
\\
\>[2]\AgdaSpace{}\AgdaSpace{}\AgdaSpace{}\AgdaSpace{}\AgdaSymbol{;}\AgdaSpace{}%
\AgdaField{0\#}%
\>[23]\AgdaSymbol{=}\AgdaSpace{}%
\AgdaFunction{0\#}\<%
\\
\>[2]\AgdaSpace{}\AgdaSpace{}\AgdaSpace{}\AgdaSpace{}\AgdaSymbol{;}\AgdaSpace{}%
\AgdaField{∞\#}%
\>[23]\AgdaSymbol{=}\AgdaSpace{}%
\AgdaFunction{∞\#}\<%
\\
\>[2]\AgdaSpace{}\AgdaSpace{}\AgdaSpace{}\AgdaSpace{}\AgdaSymbol{;}\AgdaSpace{}%
\AgdaField{f∞}%
\>[23]\AgdaSymbol{=}\AgdaSpace{}%
\AgdaFunction{f∞}\<%
\\
\>[2]\AgdaSpace{}\AgdaSpace{}\AgdaSpace{}\AgdaSpace{}\AgdaSymbol{;}\AgdaSpace{}%
...
\\
\>[2]\AgdaSpace{}\AgdaSpace{}\AgdaSpace{}\AgdaSpace{}\AgdaSymbol{\}}\<%
\end{code}
Agda is a little more strict than pen-and-paper mathematics and so the ``...'' hide a couple of additional extra fields that are required in the formalisation but are left in implicit in the definitions in this paper. For example, Agda requires us to give an explicit definition for what we mean when we say that two path-weights are equal.

Using the algebra we can automatically construct the asynchronous state function from Definition~\ref{def:aStateFun} as follows:
\begin{code} \small \>[0]\AgdaFunction{δ}\AgdaSpace{}%
\AgdaSymbol{:}\AgdaSpace{}%
\AgdaFunction{Network}\AgdaSpace{}%
\AgdaBound{n}\AgdaSpace{}%
\AgdaSymbol{→}\AgdaSpace{}%
\AgdaRecord{Schedule}\AgdaSpace{}%
\AgdaBound{n}\AgdaSpace{}%
\AgdaSymbol{→}\AgdaSpace{}%
\AgdaFunction{𝕋}\AgdaSpace{}%
\AgdaSymbol{→}\AgdaSpace{}%
\AgdaFunction{RoutingMatrix}\AgdaSpace{}%
\AgdaBound{n}\<%
\\
\>[0]\AgdaFunction{δ}\AgdaSpace{}%
\AgdaSymbol{\{}\AgdaBound{n}\AgdaSymbol{\}}\AgdaSpace{}%
\AgdaBound{N}\AgdaSpace{}%
\AgdaBound{ψ}\AgdaSpace{}%
\AgdaSymbol{=}\AgdaSpace{}%
\AgdaFunction{AsyncRouting.δ}\AgdaSpace{}%
\AgdaFunction{bgpLite}\AgdaSpace{}%
\AgdaBound{N}\AgdaSpace{}%
\AgdaBound{ψ}\AgdaSpace{}%
\AgdaSymbol{(}\AgdaFunction{I}\AgdaSpace{}%
\AgdaBound{n}\AgdaSymbol{)}\<%
\end{code}
\noindent where \AgdaBound{n} is the number of routers in the protocol. This function is fully computable, so if you provide a network, a schedule and a time $t$ it will compute the \AgdaFunction{RoutingMatrix} that represents the current protocol state.

In order to prove that \AgdaFunction{δ} is always convergent, we would like to apply our formalisation of Theorem~\ref{thm:pv-strictly-increasing}:
\begin{code} \small \>[0]\AgdaFunction{incrPaths⇒convergent}\AgdaSpace{}%
\AgdaSymbol{:}%
\>[92I]\AgdaRecord{IsRoutingAlgebra}\AgdaSpace{}%
\AgdaBound{algebra}\AgdaSpace{}%
\AgdaSymbol{→}\<%
\\
\>[.][@{}l@{}]\<[92I]%
\>[36]\AgdaFunction{IsPathAlgebra}\AgdaSpace{}%
\AgdaBound{algebra}\AgdaSpace{}%
\AgdaSymbol{→}\<%
\\
\>[.][@{}l@{}]\<[92I]%
\>[36]\AgdaFunction{IsIncreasing}\AgdaSpace{}%
\AgdaBound{algebra}\AgdaSpace{}%
\AgdaSymbol{→}\<%
\\
\>[36]\AgdaFunction{Convergent}\AgdaSpace{}%
\AgdaBound{algebra}\AgdaSpace{}%
\<%
\\
\>[0]\AgdaFunction{incrPaths⇒convergent}\AgdaSpace{}%
\AgdaSymbol{=}\AgdaSpace{}%
\AgdaBound{...}\<%
\end{code}
In order to apply it, we need to prove that the algebra \AgdaFunction{bgpLite} is a routing algebra, a path algebra and increasing:
\begin{code} \small \>[0]\AgdaFunction{bgpLite-isRoutingAlgebra}\AgdaSpace{}%
\AgdaSymbol{:}\AgdaSpace{}%
\AgdaRecord{IsRoutingAlgebra}\AgdaSpace{}%
\AgdaFunction{bgpLite}\<%
\\
\>[0]\AgdaFunction{bgpLite-isRoutingAlgebra}\AgdaSpace{}%
\AgdaSymbol{=}\AgdaSpace{}%
\AgdaKeyword{record}\<%
\\
\>[0][@{}l@{\AgdaIndent{0}}]%
\>[2]\AgdaSymbol{\{}\AgdaSpace{}%
\AgdaField{⊕-sel}%
\>[20]\AgdaSymbol{=}\AgdaSpace{}%
\AgdaFunction{...}\<%
\\
\>[2]\AgdaSymbol{;}\AgdaSpace{}%
\AgdaField{⊕-comm}%
\>[20]\AgdaSymbol{=}\AgdaSpace{}%
\AgdaFunction{...}\<%
\\
\>[2]\AgdaSymbol{;}\AgdaSpace{}%
\AgdaField{⊕-assoc}%
\>[20]\AgdaSymbol{=}\AgdaSpace{}%
\AgdaFunction{...}\<%
\\
\>[2]\AgdaSymbol{;}\AgdaSpace{}%
\AgdaField{⊕-zeroʳ}%
\>[20]\AgdaSymbol{=}\AgdaSpace{}%
\AgdaFunction{...}\<%
\\
\>[2]\AgdaSymbol{;}\AgdaSpace{}%
\AgdaField{⊕-identityʳ}%
\>[20]\AgdaSymbol{=}\AgdaSpace{}%
\AgdaFunction{...}\<%
\\
\>[2]\AgdaSymbol{;}\AgdaSpace{}%
\AgdaField{▷-fixedPoint}\AgdaSpace{}%
\>[20]\AgdaSymbol{=}\AgdaSpace{}%
\AgdaFunction{...}\<%
\\
\>[2]\AgdaSymbol{;}\AgdaSpace{}%
\AgdaField{f∞-reject}\AgdaSpace{}%
\>[20]\AgdaSymbol{=}\AgdaSpace{}%
\AgdaFunction{...}\<%
\\
\>[2]\AgdaSymbol{\}}\<%
\end{code}
\begin{code} \small \>[0]\AgdaFunction{bgpLite-isPathAlgebra}\AgdaSpace{}%
\AgdaSymbol{:}\AgdaSpace{}%
\AgdaRecord{IsPathAlgebra}\AgdaSpace{}%
\AgdaFunction{bgpLite}\<%
\\
\>[0]\AgdaFunction{bgpLite-isPathAlgebra}\AgdaSpace{}%
\AgdaSymbol{=}\AgdaSpace{}%
\AgdaKeyword{record}\<%
\\
\>[0][@{}l@{\AgdaIndent{0}}]%
\>[2]\AgdaSymbol{\{}\AgdaSpace{}%
\AgdaField{path}%
\>[21]\AgdaSymbol{=}\AgdaSpace{}%
\AgdaFunction{path}\<%
\\
\>[2]\AgdaSymbol{;}\AgdaSpace{}%
\AgdaField{r≈0⇒path[r]≈[]}%
\>[21]\AgdaSymbol{=}\AgdaSpace{}%
\AgdaFunction{...}\<%
\\
\>[2]\AgdaSymbol{;}\AgdaSpace{}%
\AgdaField{r≈∞⇒path[r]≈∅}%
\>[21]\AgdaSymbol{=}\AgdaSpace{}%
\AgdaFunction{...}\<%
\\
\>[2]\AgdaSymbol{;}\AgdaSpace{}%
\AgdaField{path[r]≈∅⇒r≈∞}%
\>[21]\AgdaSymbol{=}\AgdaSpace{}%
\AgdaFunction{...}\<%
\\
\>[2]\AgdaSymbol{;}\AgdaSpace{}%
\AgdaField{path-reject}%
\>[21]\AgdaSymbol{=}\AgdaSpace{}%
\AgdaFunction{...}\<%
\\
\>[2]\AgdaSymbol{;}\AgdaSpace{}%
\AgdaField{path-accept}%
\>[21]\AgdaSymbol{=}\AgdaSpace{}%
\AgdaFunction{...}\<%
\\
\>[2]\AgdaSymbol{;}\AgdaSpace{}%
\AgdaField{...}%
\\
\>[2]\AgdaSymbol{\}}\<%
\end{code}
\begin{code} \small \>[0]\AgdaFunction{bgpLite-isIncreasing}\AgdaSpace{}%
\AgdaSymbol{:}\AgdaSpace{}%
\AgdaFunction{IsIncreasing}\AgdaSpace{}%
\AgdaFunction{bgpLite}\<%
\\
\>[0]\AgdaFunction{bgpLite-isIncreasing}\AgdaSpace{}%
\>[42]\AgdaSymbol{=}\AgdaSpace{}%
\AgdaFunction{...}\<%
\end{code}
\noindent We omit the formalisations of the 13 proofs required above, as they are too long to reproduce here and can be found in the library. Some of them are very easy. For example, \AgdaFunction{r≈0⇒path[r]≈[]} is trivially true by the definitions of \AgdaFunction{0\#} and \AgdaFunction{path} above and therefore the proof is simply an appeal to the reflexivity of equality:
\begin{code} \small \>[0]\AgdaFunction{r≈0⇒path[r]≈[]}\AgdaSpace{}%
\AgdaSymbol{:}\AgdaSpace{}%
\AgdaGeneralizable{r}\AgdaSpace{}%
\AgdaOperator{\AgdaDatatype{≡}}\AgdaSpace{}%
\AgdaFunction{0\#}\AgdaSpace{}%
\AgdaSymbol{→}\AgdaSpace{}%
\AgdaFunction{path}\AgdaSpace{}%
\AgdaGeneralizable{r}\AgdaSpace{}%
\AgdaOperator{\AgdaDatatype{≡}}\AgdaSpace{}%
\AgdaInductiveConstructor{valid}\AgdaSpace{}%
\AgdaInductiveConstructor{[]}\<%
\\
\>[0]\AgdaFunction{r≈0⇒path[r]≈[]}\AgdaSpace{}%
\AgdaInductiveConstructor{refl}\AgdaSpace{}%
\AgdaSymbol{=}\AgdaSpace{}%
\AgdaInductiveConstructor{refl}\<%
\end{code}
Others, notably associativity, commutativity and increasingness, are more difficult.
However, all of these proofs are trivial in comparison to the proof of Theorem~\ref{thm:pv-strictly-increasing}, which we can now apply for free as follows:
\begin{code} \small \>[0]\AgdaFunction{bgpLite-convergent}\AgdaSpace{}%
\AgdaSymbol{:}\AgdaSpace{}%
\AgdaFunction{Convergent}\AgdaSpace{}%
\AgdaFunction{bgpLite}\<%
\\
\>[0]\AgdaFunction{bgpLite-convergent}\AgdaSpace{}%
\AgdaSymbol{=}\AgdaSpace{}%
\AgdaFunction{incrPaths⇒convergent}\AgdaSpace{}%
\\
\>[2]\AgdaSpace{}\AgdaSpace{}\AgdaSpace{}\AgdaSpace{}\AgdaFunction{bgpLite-isRoutingAlgebra}\AgdaSpace{}%
\\
\>[2]\AgdaSpace{}\AgdaSpace{}\AgdaSpace{}\AgdaSpace{}\AgdaFunction{bgpLite-isPathAlgebra}\AgdaSpace{}%
\\
\>[2]\AgdaSpace{}\AgdaSpace{}\AgdaSpace{}\AgdaSpace{}\AgdaFunction{bgpLite-isIncreasing}\end{code}
\noindent to obtain the final computer-checked proof that the \AgdaFunction{bgpLite} algebra is always convergent}, i.e. from any starting state \AgdaFunction{δ} converges to a unique solution even in the presence of message loss, reordering and duplication. An alternative view of this result is that it is impossible to write policies that interfere with the convergence of this protocol. One final interesting point is that this proof is constructive, i.e. from it you can extract a time at which the protocol is guaranteed to have converged.

There are many other features of BGP that are safely increasing but, for space reasons, are not included in this model. For example, we hypothesise that hierarchical AS paths are possible to add with minor tweaks to the $path$ function and the policy language.

\section{Open questions}
\label{sec:conclusions}

\subsection{Rate of convergence of path-vector protocols}

One aspect of convergence that this paper does not address is how long it takes for convergence to occur. Consider a network of $n$ routers.
It has long been known that path-vector protocols with distributive, strictly increasing policies require $\Theta(n)$ time to converge in the worst-case~\cite{gondran2008graphs,baras2010path}. However, it was only recently proved that non-distributive, strictly increasing path-vector protocols require $\Theta(n^2)$ time to converge in the worst-case~\cite{daggitt2018rate}. Meanwhile, a tight bound for the time required for the worst-case convergence of path-vector protocols in free networks is still not known however. Current results have it lying between $\Omega(n^2)$ and $O(n!)$.

Another question of interest is that it appears that not all non-distributive, strictly increasing algebras require $\Omega(n^2)$ iterations (e.g. the shortest-widest-paths algebra). We suspect that a careful analysis of policy language features might be able to tease apart distinct classes with respect to worst-case convergence time. 

\subsection{Hidden information}
\label{sec:hidden}

Unlike in (external) BGP, in the algebra described in Section~\ref{sec:bgp-lite-protocol}
the local preference attribute is not deleted when exporting a route.
This raises a more general issue for routing protocols that allow information to be hidden.
It is an open question as to whether it is possible to have increasing
algebras with hidden information without requiring global coordination. 

Ensuring increasing policies in today's BGP  may require communicating lost information with some other mechanism such as community values.
Of course only the relative ranking of local preference values assigned within an AS
matter.
For example one AS might use a local preference of 100 for
its most preferred routes, while another could use 2000.
In this context can we ensure increasing policies using only bilateral agreements between neighbouring networks or does it truly require  global coordination?
If the latter, then a political rather than a technical solution is required.

\subsection{Verification of data-center policies}
\timrevision{
BGP is widely used today to implement (private) connectivity
within and between data-centers~\cite{B4:SDN,rfc7938,Abhashkumar2021RunningBI}.
In such an environment the network architects have total control of the global topology and therefore, unlike in inter-domain routing, hidden information is not an issue.
Recent work~\cite{Abhashkumar2021RunningBI}
provides good examples of how conditional policies, combined
with filtering and the manipulation of local preference on routes,
are used at Facebook to achieve design goals in this environment. 
Perhaps tools such as Propane~\cite{propane:2016} could be extended to either ensure that
all policies are strictly increasing, or at the very least provide warnings when they are not.
}

\section{Previous versions of this work}
\label{sec:previous-versions}

A previous version of this work has appeared before in conference proceedings~\cite{daggitt2018asynchronous}. 
This version makes the following new contributions, listed in order of appearance:
\begin{enumerate}
\item in this version, the set of policy functions $\extension$ in a routing algebra in Definition~\ref{def:raw-routing-algebra} are parameterised by edges, which allows us to define a path algebra in Definition~\ref{def:path-algebra} without having it depend on $\A$, the current topology of the network. Although an apparently minor change, this makes the path algebras significantly easier to work with.

\item in this version we use the new dynamic model of asynchronous iterations from~\cite{daggitt2022dynamic} rather than the static model from~\cite{uresin1990parallel} which only captures the dynamics of a single epoch. This is not only a more accurate model as it can capture messages from non-participating nodes from previous epochs, but it also results in the model being fully executable.

\item as a consequence of the previous point, in this version we therefore use the dynamic AMCO convergence results from~\cite{daggitt2022dynamic} instead of the static ultra-metric convergence results from Gurney~\cite{gurney:ultra:metric} used in the previous version.

\item in the previous version we only proved results about strictly increasing algebras (Theorems~\ref{thm:dv-finite-strictly-increasing}~\&~\ref{thm:pv-strictly-increasing}), whereas in this version we include results about free networks (Theorems~\ref{thm:dv-finite-free}~\&~\ref{thm:pv-free}). 
\end{enumerate}
With the exception of Theorems~\ref{thm:dv-finite-free}~\&~\ref{thm:pv-free}, much of this work appears in Chapters~4~\&~5 in the PhD thesis of the first author~\cite{daggitt2019thesis}.

\section{Conflicts of interest}

On behalf of all authors, the corresponding author states that there is no conflict of interest.

\section{Acknowledgements}

Matthew L. Daggitt was supported by an EPSRC Doctoral Training grant and the EPSRC AISEC grant.
\vspace{-2em}

\begin{IEEEbiography}[{\includegraphics[width=1in,height=1.25in,clip,keepaspectratio]{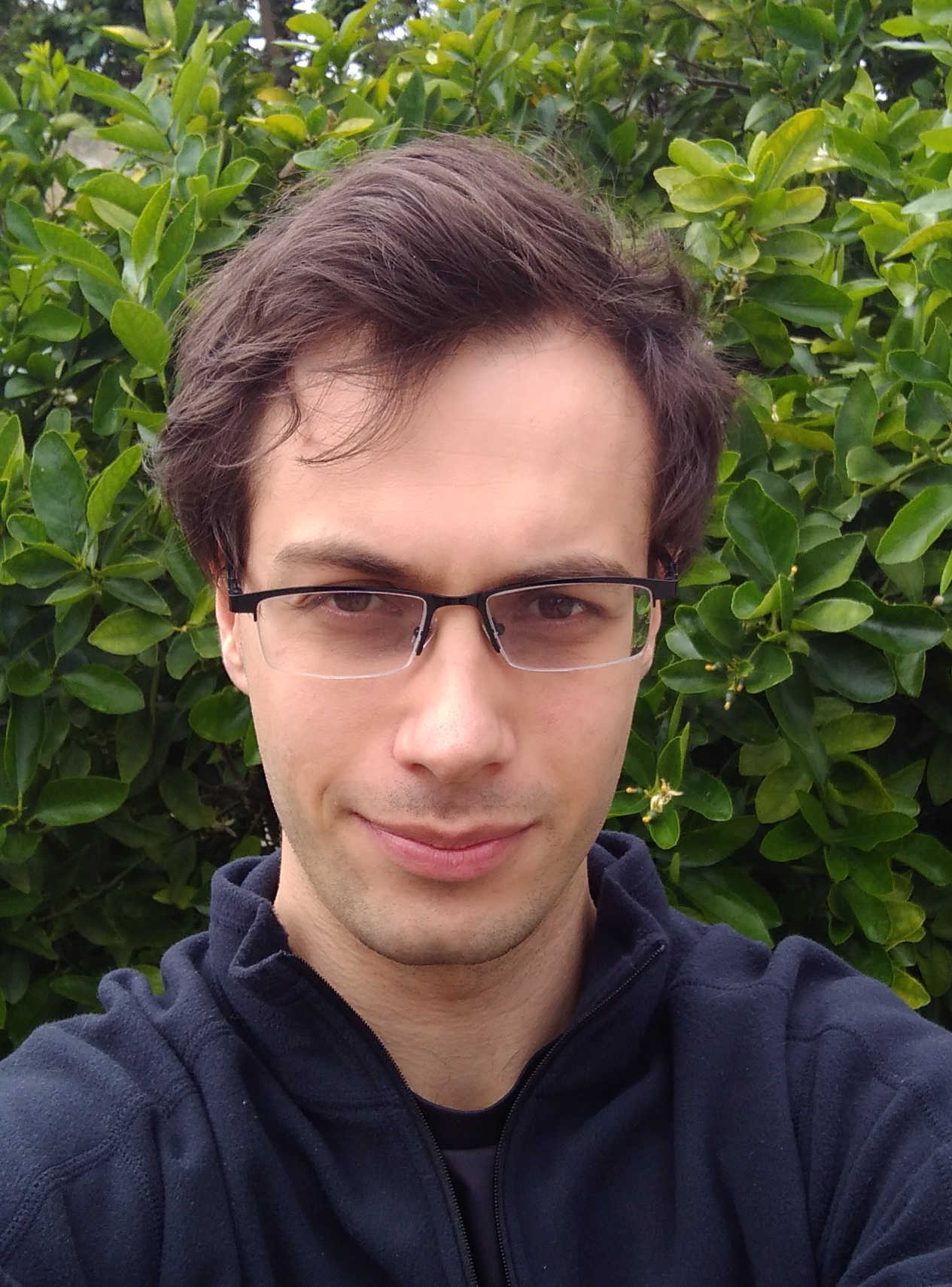}}]{Dr Matthew L. Daggitt} received his PhD degree in computer science from the University of Cambridge in 2018 and is currently working as a post-doc at Heriot-Watt University, Scotland. His main interests lie in the formal analysis and verification of complex systems, such as networking protocols and machine learning algorithms.
\end{IEEEbiography}

\begin{IEEEbiography}[{\includegraphics[width=1in,height=1.25in,clip,keepaspectratio]{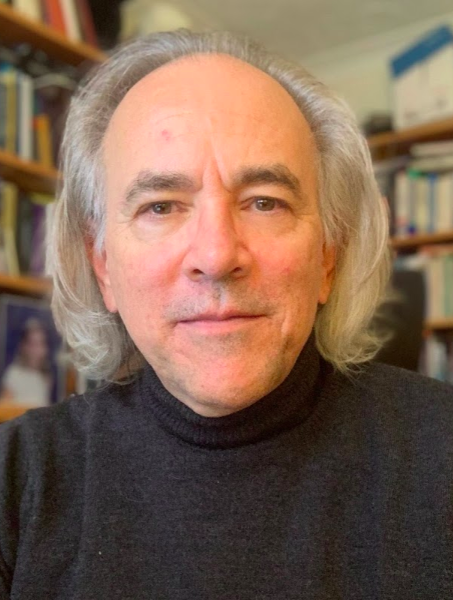}}]{Prof. Timothy G Griffin}
received his B.S. degree in mathematics from the University of Wisconsin, Madison, in 1979, and M.S. and his Ph.D. degrees in computer science from Cornell University, Ithaca, NY, in 1985 and 1988, respectively. Previous experience includes teaching at UNICAMP in Brazil and more than a dozen years in industry at Bell Laboratories, AT\&T Research, and Intel Labs. He joined the Computer Lab at the University of Cambridge in 2005.
\end{IEEEbiography}

\bibliographystyle{IEEEtran}
\bibliography{IEEEabrv,references}

\appendices

\clearpage
\section{Notational glossary} 
\label{app:glossary}

\begin{table}[h]
\newcommand{\na}{N/A}
\begin{tabular}{lll}
\toprule
Notation & Meaning & First defined \\
\midrule
$\allNodes$ & Set of all node labels & Section~\ref{sec:paths-model} \\
$\allEdges$ & Set of all edges & Section~\ref{sec:paths-model} \\
$\paths$ & Set of all simple paths & Section~\ref{sec:paths-model} \\
$\trivialpath$ & The trivial path & Section~\ref{sec:paths-model} \\
$\invalidpath$ & The invalid path & Section~\ref{sec:paths-model} \\
$\pathsaligned$  & Alignment of edges and paths & Section~\ref{sec:paths-model} \\
$\cons$  & Concentation of edges to paths & Section~\ref{sec:paths-model} \\
$i$      & Arbitrary origin node & \na \\
$j$      & Arbitrary destination node & \na \\
$k$      & Arbitrary neighbour node & \na \\
$p$      & Arbitrary path & \na \\
\midrule
$S$ 			& Set of path-weights & Definition~\ref{def:raw-routing-algebra} \\
$\oplus$ 		& Choice operator over path-weights & Definition~\ref{def:raw-routing-algebra} \\
$\extension_{ij}$ & Set of policy functions for edge $(i,j)$ & Definition~\ref{def:raw-routing-algebra} \\
$\allExtensions$ & Set of all policy functions & Definition~\ref{def:raw-routing-algebra} \\
$\trivial$		& Trivial path-weight & Definition~\ref{def:raw-routing-algebra} \\
$\invalid$ 		& Invalid path-weight & Definition~\ref{def:raw-routing-algebra} \\
$\invalidedge$ & The constantly invalid policy & Definition~\ref{def:raw-routing-algebra} \\
$\pathf$ & Function from path-weights to paths & Definition~\ref{def:path-algebra} \\
$\leqWeight$ & Total order over path-weights & Definition~\ref{def:leq-path-weights} \\
$\lessWeight$ & Strict total order over path-weights & Definition~\ref{def:leq-path-weights} \\
$x,y,z$ & Arbitrary path-weights & \na \\
$f$ & Arbitrary policy function & \na \\
\midrule
$\nodes$ & Set of labels of nodes that participate & Section~\ref{sec:synchronous-model} \\
$\tcarrier$  & Set of routing tables ($n$ vectors over $\carrier$) & Section~\ref{sec:synchronous-model} \\
$\mcarrier$ 	& Set of states ($n \times n$ matrices over $S$) & Section~\ref{sec:synchronous-model} \\
$\A$ 			& Arbitrary topology/adjacency matrix & Section~\ref{sec:synchronous-model} \\
$\I$ 			& Identity matrix & Section~\ref{sec:synchronous-model} \\
$\aFun_\A$		& A single sync. iteration with topology $\A$ & Section~\ref{sec:synchronous-model} \\
$\sigma$ 		& Synchronous state function & Definition~\ref{def:sStateFun} \\
$\xt,\yt$		& Arbitrary routing tables & \na \\
$\X, \Y$        & Arbitrary routing states & \na \\
$\X^*$ 			& Arbitrary stable routing state & \na  \\
\midrule
$\timess$		& The set of times of events of interest & Section~\ref{sec:asynchronous-model} \\
$\epochs$       & The set of epochs & Section~\ref{sec:asynchronous-model} \\
$\network$		& An arbitrary network & Section~\ref{sec:asynchronous-model} \\
$\alpha$ 		& Schedule activation function & Definition~\ref{def:schedule} \\
$\beta$ 		& Schedule data flow function & Definition~\ref{def:schedule} \\
$\eta$          & Schedule epoch function & Definition~\ref{def:schedule} \\
$\pi$           & Schedule participants function & Definition~\ref{def:schedule} \\
$\rho$          & Participants function & Section~\ref{sec:asynchronous-model} \\
$\A^{ep}$		& Participation adjacency matrix & Definition~\ref{def:participating-topology} \\
$\delta$ 		& Asynchronous state function & Definition~\ref{def:aStateFun} \\
$\accSet$ & Set of accordant states & Definition~\ref{def:accordant-states} \\
$\configs$ & Set of all configurations & Section~\ref{sec:definition-of-convergence} \\
$\freeConfigs$ & Set of free configurations & Section~\ref{sec:routing-algebra-classes} \\
$t$				& Arbitrary time & \na \\
$e$				& Arbitrary epoch & \na \\
$p$				& Arbitrary set of participants & \na \\
\midrule
$\assignments$ & Set of assignments & Definition~\ref{def:assignments} \\
$\lessAss$ & Strict partial order over assignments & Definition~\ref{def:assignment-preference} \\
$\extendedBy{\A}$ & Extended by relation over assignments & Definition~\ref{def:assignment-extension} \\
$\threatenedBy{\A}$ & Threatens relation over assignments & Definition~\ref{def:assignment-threatens} \\
$a, b, c$ & Arbitrary assignments & \na \\
\midrule
$\proofOrder$ & Dislodging order over assignments & Appendix~\ref{app:distance_vector_proof} \\
$h^{ep}$ & Height of assignments & Appendix~\ref{app:distance_vector_proof} \\
$r^{ep}_i$ & Dissimilarity over path-weights & Appendix~\ref{app:distance_vector_proof} \\
$d^{ep}_i$ & Dissimilarity over routing table states & Definition~\ref{def:amco} \\
$D^{ep}$   & Dissimilarity over routing states & Definition~\ref{def:amco} \\
\bottomrule
\end{tabular}
\end{table}

\clearpage
\section{Proof of Theorem~\ref{thm:dv-finite-free} - distance-vector protocols}
\label{app:distance_vector_proof}

This appendix contains the proof that distance-vector protocols over a finite routing algebra are convergent whenever the network topology is free.

Assume that we have a finite routing algebra $\genericRoutingAlgebra$ and some network $\network$. We will prove that asynchronous model of the protocol $\delta$ is convergent whenever the network topology $\A^{ep}$ is free, by showing that the set of synchronous iterations $F_{\A^{ep}}$ is an AMCO over $\freeConfigs$. The required result then follows by Theorem~\ref{thm:amco-convergent}.

\subsection{Constructing the dissimilarity functions}

Consider an arbitrary configuration consisting of an epoch~$e$ and set of participants~$p$ such that the participation topology~$\A^{ep}$ is free. 
To fulfil the AMCO conditions, for each router $i$ we must construct a suitable dissimilarity function,~$\tmetric{i}$, that can be used to compare the different states of its routing table. Intuitively, our notion of dissimilarity will be based upon the ability of the router $i$'s current entries of \emph{dislodging} entries in other router's routing tables.

We now describe the formal construction of such a dissimilarity function. Note that for notational convenience, whenever we refer to a concept that depends an adjacency matrix we will use $ep$ as a shorthand for $\A^{ep}$. For example we will write $\aFun^{ep}$ instead of $\aFun_{\A^{ep}}$ and $\extendedBy{}^{ep}$ instead of $\extendedBy{\A^{ep}}$. 

\vspace{1em}

\subsubsection{Dislodgement order over assignments}

To begin with we construct a strict partial order $<^{ep}$ over the set of assignments, which we will call the \emph{dislodgement order}.
The key idea for this is adapted from Sobrinho's notion of a \emph{couplet digraph} which we generalise to an \emph{assignment} digraph. 
The vertices of this digraph are the set of assignments $\assignments$.
There is an edge from assignment $a$ to assignment $b$ in the digraph if, in the current topology, either the latter is an extension of the former (i.e. $a \extendedByU b$), or the former is strictly preferred to the latter (i.e. $a \lessAss b$). 
Subsequently, we have $a \proofOrder b$ if there exists a path from assignment $a$ to assignment $b$ in this digraph. 

Mathematically, this relation is the transitive closure of the union of the $\extendedByU$ and $\lessAss$ relations. Concretely, if we define the $m$-step version of the order $\proofOrder_m$ inductively as:
\begin{align*}
a \proofOrder_1 \:\: b \spacing & \triangleq \spacing a \extendedByU b \vee a \lessAss b \\
a \proofOrder_{m+1} b \spacing & \triangleq \spacing \exists c : a \proofOrder_1 c \wedge c \proofOrder_m b \\
\intertext{then $\proofOrder$ is defined as:}
a \proofOrder \:\: b \spacing & \triangleq \spacing \exists m : a \proofOrder_m b
\end{align*}
Intuitively, $a \proofOrder b$ describes a chain of events triggered by a router choosing to use assignment $a$ that would eventually cause another router to remove assignment $b$ from its routing table, i.e. $a$ can eventually cause $b$ to be dislodged. In particular, each $\extendedByU$ link represents an extension event and each $\lessAss$ link represents a choice event.

The crucial properties of this relation are as follows:
\begin{enumerate}[label=\textbf{O\arabic*}]
\item \emph{$\proofOrder$ is irreflexive} - An assignment cannot dislodge itself.
\begin{equation*}
\forall a \in \assignments: a \nless^{ep} a
\end{equation*}
\item \emph{$\proofOrder$ is transitive} - If assignment $a$ can dislodge $b$ and $b$ can dislodge $c$ then $a$ can dislodge $c$.
\begin{equation*}
\forall a,b,c \in \assignments: a \proofOrder b \wedge b \proofOrder c \Rightarrow a \proofOrder c
\end{equation*}
\item \emph{$\proofOrder$ respects $\extendedByU$} - If assignment $a$ is extended by assignment $b$ then $a$ can dislodge $b$.
\begin{equation*}
\forall a,b \in \assignments: a \extendedByU b \Rightarrow a \proofOrder b
\end{equation*}
\item \emph{$\proofOrder$ respects $\lessWeight$} - If assignment $a$ is strictly preferred to assignment $b$ then $a$ can dislodge $b$.
\begin{equation*}
\forall a,b \in \assignments: a \lessAss b \Rightarrow a \proofOrder b
\end{equation*}
\end{enumerate}
Properties O2, O3 and O4 are immediate from the definition of~$\proofOrder$. In contrast O1, the proof of irreflexivity, is non-trivial and uses the assumption that $\A^{ep}$ is free. However, modulo the differences in the algebraic structures and the change from couplets to assignments, the proof is identical to that of Lemma~1 in the appendix of~\cite{sobrinho2005algebraic}. Readers interested in a proof using our algebraic model may consult our accompanying Agda code.

\vspace{1em}

\subsubsection{Height of assignments} 

As both $\carrier$ and $\nodes$ are finite, the set of assignments $\assignments$ is also finite. Therefore all upwards closed subsets under the relation $<^{ep}$ must also be finite. Consequently the \emph{height} of an assignment can be defined as the number of assignments it is capable of dislodging:
\begin{equation*}
\rheight(a) \triangleq \left\vert \{ b \in \assignments \mid a <^{ep} b \} \right\vert
\end{equation*}
The crucial properties of $\rheight$ are as follows:
\begin{enumerate}[label=\textbf{H\arabic*}]
\item If $a$ is extended by $b$ then the height of $b$ is strictly less than the height of $a$:
\begin{equation*}
\forall a,b \in \assignments: a \extendedByU b \Rightarrow \rheight(b) < \rheight(a)
\end{equation*}
\item If $a$ is strictly preferred to $b$ then the height of $b$ is strictly less than the height of $a$:
\begin{equation*}
\forall a,b \in \assignments: a \lessAss b \Rightarrow \rheight(b) < \rheight(a)
\end{equation*}
\item There exists a maximum height:
\begin{equation*}
\exists \rheightMax \in \NN : \forall a \in \assignments : \rheight(a) \leq \rheightMax
\end{equation*}
\end{enumerate}
H1 and H2 trivially follow from properties O1, O2, O3 and O4. H3 follows from the fact that $A$ is finite.

\vspace{1em}

\subsubsection{Dissimilarity between path-weights} 

For each router $i$ we can now construct a dissimilarity function ${\rmetric{i} : \carrier \times \carrier \rightarrow \mathbb{N}}$ between path-weights as follows:
\begin{equation*}
\rmetric{i}(x,y) \triangleq 
\begin{cases}
	0 							& \text{if $x = y$} \\
	1 + \max(\rheight(i, x), \rheight(i, y))	& \text{otherwise}
\end{cases} \\
\end{equation*}
By this definition the dissimilarity between a pair of non-equal path-weights is proportional to the maximum number of assignments that are dislodgeable when router $i$ chooses to use either $x$ or $y$. Intuitively this is a reasonable measure of dissimilarity as if you have two alternative assignments then, from a convergence perspective, the seriousness of the disagreement between them is directly proportional to the number of resulting changes in the rest of the network when adopting one over the other.

\vspace{1em}

\subsubsection{Dissimilarity between routing tables} For each router $i$, the dissimilarity between a pair of routing table states that it may adopt is defined to be the maximum of the pairwise dissimilarity between their entries:
\begin{equation*}
\tmetric{i}(\xt,\yt) \triangleq \max_{j \in \nodes} \: \rmetric{i}(\xt_j,\yt_j)
\end{equation*}
If the two routing tables are identical then they will have zero dissimilarity between them, otherwise the dissimilarity is equal to the maximum number of assignments that are dislodgeable by $i$ choosing to use a route that  $\xt$ and $\yt$ disagree on. This is the dissimilarity function $\tmetric{i}$ required by the AMCO conditions in Definition~\ref{def:amco}.

\vspace{1em}

\subsubsection{Dissimilarity between routing states} As required by the AMCO conditions, the dissimilarity function over routing states is defined as follows:
\begin{equation*}
\smetric(\X,\Y) \triangleq \max_{i \in p} \: \tmetric{i}(\X_{i},\Y_{i})
\end{equation*}
Again $\smetric$ measures the dissimilarity between states~$\X$~and~$\Y$. If all the elements of $\X$ and $\Y$ are equal then there is zero dissimilarity between them, otherwise the dissimilarity is equal to the maximum number of assignments that are dislodgeable by any participating router $i$ choosing a route that $\X$ and $\Y$ disagree on.

\subsection{Properties of the dissimilarity functions}

We must now show that $\tmetric{i}$ and $\smetric$ satisfy the AMCO conditions. It is clear that $\tmetric{i}$ is always less than or equal to $\rheightMax + 1$ by H3 and that $\tmetric{i}(\xt,\yt) = 0$ iff $\xt = \yt$. In fact it can be shown that $\rmetric{i}$/$\tmetric{i}$/$\smetric$ form an ultrametric space over $\carrier$/$\tcarrier$/$\mcarrier$ but showing this is not necessary for the AMCO conditions.

Therefore it remains to show that $\aFun^{ep}$ is both strictly contracting over orbits and strictly contracting over any fixed point. In fact, in this particular case  $\aFun^{ep}$ can be shown to be strictly contracting:
\begin{equation*}
\X \neq_p \Y \Rightarrow \smetric(\aFun^{ep}(\X),\aFun^{ep}(\Y)) < \smetric(\X,\Y)
\end{equation*}
which implies it is strictly contracting over both orbits and any fixed points. 

However, instead of proving this result directly, we will first separate out a couple of smaller lemmas which we will later reuse in Appendix~\ref{app:path-vector-proof} in the proof of convergence for path-vector protocols.

\begin{lemma} \label{lem:constant_diagonal}
After an iteration, any router $i$'s path-weight to itself is always the trivial path-weight:
\begin{equation*}
\forall i, \X: \aFun^{ep}(\X)_{ii} = \trivial
\end{equation*}
\end{lemma}
\begin{proof}
$\aFun^{ep}(\X)_{ii} = \I_{ii}$ by Equation~\ref{eq:sigma_expanded} in Section~\ref{sec:synchronous-model} and $\I_{ii} = \trivial$ by definition of $\I$.
\end{proof}

\begin{lemma}
\label{lem:strictly-contracting-sub}
For each destination, the maximum dissimilarity between all entries of two routing states must strictly decrease after a single synchronous iteration:
\begin{align*}
&\forall \X,\Y,i,j,v > 0: \\
&\quad (\forall k : \rmetric{k}(\X_{kj},\Y_{kj}) \leq v) \Rightarrow
\rmetric{i}(\aFun^{ep}(\X)_{ij},\aFun^{ep}(\Y)_{ij}) < v
\end{align*}
\end{lemma}
\begin{proof}
Assume
\begin{equation}
\label{eq:rXYk<=v}
\forall k : \rmetric{k}(\X_{kj},\Y_{kj}) \leq v
\end{equation}

\begin{case}{1}{$\sigX_{ij} = \sigY_{ij}$}
Then the required inequality is immediate as:
\begin{align*}
\rmetric{i}(\sigX_{ij},\sigY_{ij}) 
&= 0
& \text{(by Case 1 and def. of $\rmetric{i}$)}
\\
&< v
& \text{(by $v > 0$)}
\end{align*}
\end{case}

\begin{case}{2}{$\sigX_{ij} \neq \sigY_{ij}$}
Without loss of generality assume that $\sigX_{ij}$ is a more desirable path-weight than $\sigY_{ij}$ and therefore:
\begin{equation} 
\label{eq:sXij<sYij}
\sigX_{ij} \lessWeight \sigY_{ij} 
\end{equation}

\begin{case}{2.1}{$i = j$}
A router's path-weight to itself is always the trivial path-weight by Lemma~\ref{lem:constant_diagonal} and so if $i = j$ then $\sigX_{ij} = \trivial = \sigY_{ij}$ which contradicts the assumption of Case 2.
\end{case}

\vspace{1em}

\begin{case}{2.2}{$i \neq j$}
Consequently by Equation~\ref{eq:sigma_expanded} in Section~\ref{sec:synchronous-model}:
\begin{align*}
\sigX_{ij} = \bigoplus_k \A^{ep}_{ik}(\X_{kj})
\end{align*}
and so as $\oplus$ is selective there must exist a router $k$ such that:
\begin{equation} \label{eq:extension}
\sigX_{ij} = \A^{ep}_{ik}(\X_{kj})
\end{equation}
If $\X_{kj} = \invalid$ then $\sigX_{ij} = \invalid$ which contradicts (\ref{eq:sXij<sYij}) and therefore:
\begin{equation} \label{eq:Xkj!=0}
\X_{kj} \neq \invalid
\end{equation}
It cannot be the case that $\X_{kj} = \Y_{kj}$ as otherwise it would be possible to prove the following:
\begin{align*}
\sigX_{ij}
	& = \A^{ep}_{ik}(\X_{kj}) 
	& \text{(by (\ref{eq:extension}))} \\
	& = \A^{ep}_{ik}(\Y_{kj})  
	& \text{(by $\X_{kj}=\Y_{kj}$)} \\
	& \geq_{\oplus} \bigoplus_k \A^{ep}_{ik}(\Y_{kj})
	& \text{(by def. of $\leq_\oplus$)} \\
	& = \sigY_{ij} 
	& \text{(by Case 2.2 \& Equation~\ref{eq:sigma_expanded})} 
\end{align*}
which again contradicts (\ref{eq:sXij<sYij}) and so:
\begin{equation} \label{eq:Xkj!=Ykj}
\X_{kj} \neq \Y_{kj}
\end{equation}
The required inequality can now be proved as follows:
\begin{align*}
& \rmetric{i}(\sigX_{ij},\sigY_{ij})) \\
	& \quad = 1 + \max(\rheight(i , \sigX_{ij}), \rheight(i, \sigY_{ij}))  \\
& \quad\qquad\qquad\qquad\qquad\qquad\qquad\quad\: \text{(by def. of $\rmetric{i}$ \& Case 2)} \\
	& \quad = 1 + \rheight(i, \sigX_{ij}) 
\quad\qquad\qquad\qquad\quad \text{(by H2 \& (\ref{eq:sXij<sYij}))} \\
	& \quad < 1 + \rheight(k, \X_{kj})
\quad\qquad\qquad\qquad\qquad\quad \text{(by H1 \& (\ref{eq:extension}))} \\
	& \quad \leq 1 + \max(\rheight(k , \X_{kj}), \rheight(k , \Y_{kj}))
\quad \text{(by def. of $\max$)} \\
	& \quad = \rmetric{k}(\X_{kj},\Y_{kj}) 
\qquad\qquad\qquad\:\, \text{(by def. of $\rmetric{k}$ and (\ref{eq:Xkj!=Ykj}))} \\
	& \quad \leq v
\qquad\qquad\qquad\qquad\qquad\qquad\qquad\qquad\qquad\quad \text{(by (\ref{eq:rXYk<=v}))}
\end{align*}
\end{case}
\end{case}
\end{proof} 

\noindent We can now use Lemma~\ref{lem:strictly-contracting-sub} to prove that $\aFun^{ep}$ is strictly contracting.

\begin{lemma}
\label{lem:strictly_contracting}
The operator $\aFun^{ep}$ is strictly contracting w.r.t. $\smetric$ over $\accSet$.
\begin{align*}
&\forall \X,\Y \in \accSet: \\
&\qquad \X \neq_p \Y \Rightarrow \smetric(\aFun^{ep}(\X),\aFun^{ep}(\Y)) < \smetric(\X,\Y)
\end{align*}
\end{lemma}

\begin{proof}
Consider states $\X, \Y \in \accSet$ such that ${\X \neq \Y}$. As: 
\begin{equation*}
\smetric(\aFun^{ep}(\X),\aFun^{ep}(\Y)) = \max_{i \in p, j \in \nodes} \: \rmetric{i}(\aFun^{ep}(\X)_{ij},\aFun^{ep}(\Y)_{ij})
\end{equation*}
it suffices to show for all routers $i \in p$ and $j \in \nodes$ that
\begin{equation*} \label{eq:d<dXijYij}
\rmetric{i}(\aFun^{ep}(\X)_{ij},\aFun^{ep}(\Y)_{ij}) < \smetric(\X,\Y)
\end{equation*}
This can be achieved by applying Lemma~\ref{lem:strictly-contracting-sub} where ${v = \smetric(\X,\Y)}$. However in order to apply the lemma we must first check that ${0 < \smetric(\X,\Y)}$ and that for all routers $k$ then $\rmetric{k}(\X_{kj},\Y_{kj}) \leq \smetric(\X,\Y)$.

The former holds as $\X \neq \Y$ and $\X, \Y \in \accSet$ so there must exist entries in participating routers in $\X$ and $\Y$ whose dissimilarity is non-zero. To show that the latter holds consider whether or not $k$ is participating. If $k \in p$ then the required inequality holds simply from the definition of~$\smetric$. Otherwise if $k \notin p$ then as $\X$ and $\Y$ are accordant with $p$ then $\X_{kj} = \Y_{kj} = \invalid$ and the inequality holds trivially as $\rmetric{k}(\X_{kj},\Y_{kj}) = 0$.
\end{proof}

\noindent Finally all the pieces to show that the set of functions $\aFun^{ep}$ are an AMCO over $\freeConfigs$ have now been assembled. In particular:
\begin{enumerate}[label=\textbf{D\arabic*}:]
\item $\forall \xt, \yt : \tmetric{i}(\xt,\yt) = 0 \Leftrightarrow \xt = \yt$

Immediate from the definition of $\tmetric{i}$ and $\rmetric{i}$.

\item $\exists n : \forall \xt, \yt: \tmetric{i}(\xt,\yt) \leq n$

The height of assignments is bounded above by $\rheightMax$ and hence every dissimilarity is bounded above $1 + \rheightMax$.

\item $F^{ep}$ is strictly contracting on orbits over $\smetric$

Immediate by applying Lemma~\ref{lem:strictly_contracting} to $\X$ and $\sigX$.

\item $F^{ep}$ is strictly contracting on fixed points over $\smetric$

Immediate by applying Lemma~\ref{lem:strictly_contracting} to $\X$ and $\X^*$.

\item $\forall \X : F^{ep}(\X) \in \accSet$.

Immediate by the definition of $F^{ep}$ as $\A^{ep}$ discards the state of all non-participating routers.
\end{enumerate}
As $\aFun^{ep}$ is an AMCO over $\freeConfigs$ then Theorem~\ref{thm:amco-convergent} we have that $\delta$ is convergent over $\genericRoutingAlgebra$ whenever the network topology is free.

\clearpage
\section{Proof of Theorem~\ref{thm:pv-free} - Path-vector protocols}
\label{app:path-vector-proof}

As discussed in Section~\ref{sec:convergence-results}, the requirement that $\carrier$, the set of path-weights, is finite proves very restrictive in practice. One principled approach to avoiding this requirement is that of path-vector protocols which track and remove path-weights generated along looping paths, as encoded by the path algebra axioms in Definition~\ref{def:path-algebra}. This turns out to be sufficient to guarantee that eventually the protocol will always reach a finite subset of \emph{consistent} path-weights from which it will then converge. As with distance-vector protocols, we will use this to construct a suitable set of dissimilarity functions $\rmetric{i}$ that fulfil the AMCO conditions.

\subsection{Consistent path-weights}
\label{sec:consistent-path-weights}

We now formally define the notion of a path-weight being \emph{consistent} with the current network topology. Consider an arbitrary path algebra~$\genericPathAlgebra$.

\begin{definition}[Weight of a path]
The ${\weightf{\A} : \paths \rightarrow \carrier}$, which calculates the weight of a path~$p$ with respect to the current network topology $\A$, is defined as:
\begin{equation*}
\weightf{\A}(p) \triangleq \begin{cases}
\invalid 
	& \text{if $p = \invalidpath$} \\
\trivial 
	& \text{if $p = \trivialpath$} \\
\A_{ij}(\weightf{\A}(q)) 
	& \text{if $p = (i,j) :: q$}
\end{cases}
\end{equation*}
\end{definition}

\begin{definition} 
\label{def:consistent-path-weight}
A path-weight $x$ is \emph{consistent} with respect to the topology $\A$ iff it is equal to the current weight of the path along which it was generated. The set of consistent routes and the set of consistent routing states are therefore defined to be:
\begin{align*}
\consistent{\A} & \triangleq \{ x \in \carrier \mid \weightf{\A}(\pathf(x)) = x\} \\
\mconsistent{\A} & \triangleq \{\X \in \mcarrier \mid \forall i,j : \X_{ij} \in \consistent{\A}\}
\end{align*}
\end{definition}
\noindent The following lemmas show that consistency is preserved by the routing operations.
\begin{lemma}
\label{lem:choice-consistent-closure}
Choice preserves consistency:
\begin{equation*}
\forall x, y \in \consistent{\A} : x \oplus y \in \consistent{\A}
\end{equation*}
\end{lemma}
\begin{proof}
Immediate from axiom~\ref{ass:sel} that $\oplus$ is selective.
\end{proof}

\begin{lemma}
\label{lem:extension-consistent-closure}
Extension preserves consistency:
\begin{equation*}
\forall i,j \in \nodes, x \in \consistent{\A}: \A_{ij}(x) \in \consistent{\A}
\end{equation*} 
\end{lemma}
\begin{proof}
Consider a path-weight $x \in \consistent{\A}$ being extended along an arbitrary edge $(i , j)$. By axiom~P3 in the definition of a path algebra there are two cases.

\begin{case}{1}{$\pathf(\A_{ij}(x)) = \invalidpath$}
Then the required result follows:
\begin{align*}
& \weightf{\A}(\pathf(\A_{ij}(x))) \\
& \quad = \weightf{\A}(\invalidpath) 
& (\text{by Case 1}) \\
& \quad = \invalid
& (\text{by def. of $\weightf{\A}$}) \\
& \quad = \A_{ij}(x)  
& (\text{by axiom P1 and Case 1})
\end{align*}
\end{case}

\begin{case}{2}{$\pathf(\A_{ij}(x)) = (i , j) \cons \pathf(x)$}
Then we have the required result as follows:
\begin{align*}
& \weightf{\A}(\pathf(\A_{ij}(x))) \\
& \quad = \weightf{\A}((i,j) \cons \pathf(x)) 
& (\text{by Case 2}) \\
& \quad = \A_{ij}(\weightf{\A}(\pathf(x)) 
& (\text{by def. of $\weightf{\A}$}) \\
& \quad = \A_{ij}(x)  
& (\text{by $x \in \consistent{\A}$})
\end{align*}
\end{case}
\end{proof}

\begin{lemma}
\label{lem:F-consistent-closure}
A single synchronous iteration preserves consistency:
\begin{equation*}
\forall \X \in \mconsistent{\A} : \aFun_{\A}(\X) \in \mconsistent{\A}
\end{equation*}
\end{lemma}
\begin{proof}
This is an immediate from Lemmas~\ref{lem:choice-consistent-closure}~\&~\ref{lem:extension-consistent-closure} as every operation in $\aFun_{\A}$ preserves consistency. 
\end{proof}
\noindent One consequence of Lemma~\ref{lem:F-consistent-closure} is that applying an update never spontaneously introduces new inconsistent path-weights into the routing state. Therefore the only way inconsistent path-weights can be introduced into the routing state is by the network topology changing (consequently beginning a new epoch).

\begin{lemma}
$\consistent{\A}$ is finite.
\end{lemma}
\begin{proof}
By Definition~\ref{def:consistent-path-weight} every consistent path-weight is associated with at least one simple path in the current network topology. The set of consistent path-weights, $\consistent{\A}$, is therefore equal to:
\begin{equation*}
\consistent{\A} = \{ \weightf{\A}(p) \mid p \in P \}
\end{equation*}
where $P$ is the set of simple paths in topology $\A$. As $P$ is finite, so too is $\consistent{\A}$.
\end{proof}

As in the previous section, in order to prove that $\aStateFun$ is convergent in epochs where then network topology is free. it is necessary to find a quantity that always decreases when applying $\aFun_{\A}$. The key insight is that because the set of consistent path-weights is finite and consistency is preserved by iteration, it is therefore possible to reuse the dissimilarity functions defined in Section~\ref{app:distance_vector_proof} when calculating the dissimilarity between two consistent states. In particular, $\aFun_{\A}$ remains strictly contracting over this dissimilarity function with respect to the set of consistent states. The remaining problem is to find a quantity that decreases when applying $\aFun_{\A}$ to a routing state that contains inconsistent path-weights. 

 The solution is to consider the lengths of the inconsistent paths in the state. In particular the length of the shortest inconsistent path, $s(\X)$ in the state $\X$ must strictly increase after a single iteration. This is because the contraposition of Lemma~\ref{lem:F-consistent-closure} implies that the shortest inconsistent path-weight in $\aFun_{\A}(\X)$ must be an extension of some inconsistent path-weight in~$\X$. By \ref{ass:path-extension} the paths of the two path-weights must be of the form $(i , j) \cons p$ and $p$ respectively, and by definition of $s(\X)$ we have that $s(\X) \leq |p|$. Hence $s(\X) \leq |p| < |(i,j) \cons p| = s(\aFun_{\A}(\X))$.

One simple consequence of  is that, in the absence of further topology changes, eventually all inconsistent path-weights must be flushed from the state. This follows as if the length of the path of the shortest inconsistent path-weight continues to increase then eventually all inconsistent path-weights must be of at least length $n$. Hence their path will necessarily contain a loop and therefore they will be flushed from the routing state as a consequence of assumption~\ref{ass:path-extension}. Therefore after $n$ applications of $\aFun_\A$ the state must be consistent.

Consequently ``$n - s(\X)$'' provides the required strictly decreasing quantity for inconsistent states. The two separate strictly decreasing quantities, the length of the shortest inconsistent path and the height of the consistent routes, can then be combined to form a unified strictly decreasing dissimilarity function over both inconsistent and consistent path-weights. Figure~\ref{fig:metric_structure} shows how all the dissimilarity functions fit together.

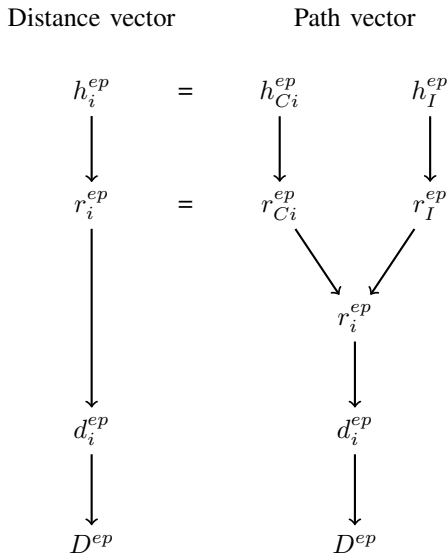
\begin{figure}[tbp]
\begin{center}
\begin{tikzpicture}
\def\typex{-2}
\def\distx{-1}
\def\pathxl{1.5}
\def\pathxc{2.5}
\def\pathxr{3.5}
\def\hy{6}
\def\dyu{4.5}
\def\dyl{3}
\def\ty{1.5}
\def\Dy{0}

% Draw dividing bar
\node [] (top) at (0,1) {};
\node [] (bottom) at (0,7) {};
%\path [draw,thick,dashed] (top) -- (bottom);

%\node [left] at (\typex, \hy)  {Height function};
%\node [left] at (\typex, \dyu) {Path-weights $\:$};
%\node [left] at (\typex, \ty) {Routing tables};
%\node [left] at (\typex, \Dy)  {Global state $\:\:$};

% Distance vector
\node [] at (\distx,7) {Distance vector};
\node [] (hd) at (\distx,\hy) {$\rheight_{i}$};
\node [] (dd) at (\distx,\dyu) {$\rmetric{i}$};
\node [] (td) at (\distx,\ty) {$\tmetric{i}$};
\node [] (Dd) at (\distx,\Dy) {$\smetric$};

% Path vector
\node [] at (\pathxc,7) {Path vector};
\node [] (hcp) at (\pathxl,\hy) {$\rheightC{i}$};
\node [] (hip) at (\pathxr,\hy) {$\rheightI$};
\node [] (dcp) at (\pathxl,\dyu) {$\rmetricC{i}$};
\node [] (dip) at (\pathxr,\dyu) {$\rmetricI$};
\node [] (dp) at (\pathxc,\dyl) {$\rmetric{i}$};
\node [] (tp) at (\pathxc,\ty) {$\tmetric{i}$};
\node [] (Dp) at (\pathxc,\Dy) {$\smetric$};

\path [->,draw,thick] (hd) -- (dd);
\path [->,draw,thick] (dd) -- (td);
\path [->,draw,thick] (td) -- (Dd);

\path [->,draw,thick] (hcp) -- (dcp);
\path [->,draw,thick] (hip) -- (dip);
\path [->,draw,thick] (dcp) -- (dp);
\path [->,draw,thick] (dip) -- (dp);
\path [->,draw,thick] (dp) -- (tp);
\path [->,draw,thick] (tp) -- (Dp);

\node [] at (0.25, \hy) {=};
\node [] at (0.25, \dyu) {=};
\end{tikzpicture}
\caption{The construction of the dissimilarity functions in the convergence proofs.}
\label{fig:metric_structure}
\end{center}
\vspace{-1em}
\end{figure}

\subsection{Constructing the dissimilarity functions}

Consider an arbitrary configuration consisting of an epoch~$e$ and set of participants~$p$ such that the participation topology~$\A^{ep}$ is free.  Note that once again for notational convenience, whenever we refer to a concept that depends an adjacency matrix we will use $ep$ as a shorthand for $\A^{ep}$. For example we will write $\aFun^{ep}$ instead of $\aFun_{\A^{ep}}$ and $\lconsistent$ instead of $\consistent{\A^{ep}}$. 

\vspace{1em}

\subsubsection{Inconsistent height of path-weights} With the above discussion in mind, the \emph{inconsistent height} of a path-weight $\rheightI : \carrier \rightarrow \mathbb{N}$ is defined as:
\begin{equation*}
\rheightI(x) \triangleq \begin{cases}
	0 							& \text{if $x \in \lconsistent$} \\
	1 + (n - \size{\pathf(x)}) 	& \text{otherwise} 
\end{cases}
\end{equation*}
where $n$ is the number of routers in the network. All consistent path-weights have the minimum height 0, and the maximum height is $n+1$ as the $\pathf$ function only returns simple paths. To explicitly highlight the parallels with the previous section, this maximum height will be called $\rheightIMax$. Therefore for all path-weights $x$ then the following relationships hold:
\begin{equation*}
0 \leq \rheightI(x) \leq n+1 = \rheightIMax 
\end{equation*}

\subsubsection{Inconsistent dissimilarity function} The inconsistent height can be used to define the inconsistent dissimilarity function ${\rmetricI: \carrier \times \carrier \rightarrow \mathbb{N}}$ as follows:
\begin{equation*}
\rmetricI(x,y) \triangleq \max(\rheightI(x),\rheightI(y))
\end{equation*}
This will always be used to compare pairs of path-weights in which at least one is inconsistent. Therefore the inconsistent dissimilarity between two path-weights is proportional to the length of the shortest inconsistent path. Note this definition does not contain a check for equality as in practice it will never be used to compare two equal path-weights.

\vspace{1em}

\subsubsection{Consistent dissimilarity function}

As discussed, the set of path-weights consistent with the current network topology is finite and consistency is closed over the routing algebra operations. Therefore $(\lconsistent, \oplus, \extension, \trivial, \invalid)$ forms a finite routing algebra and consequently the previous metric $\rmetric{i}$ from Section~\ref{app:distance_vector_proof} may be used to compare consistent path-weights. This will be renamed to $\rmetricC{i}: \consistent{\A^{ep}} \times \consistent{\A^{ep}} \rightarrow \mathbb{N}$. The maximum such consistent dissimilarity $\rheightMax$ will also be renamed to $\rheightCMax$.

\vspace{1em}

\subsubsection{Dissimilarity between path-weights}

Given both the consistent and inconsistent dissimilarity functions, then the overall dissimilarity function over path-weights $\rmetric{i} : \carrier \times \carrier \rightarrow \mathbb{N}$ is defined as: 
\begin{equation*}
\rmetric{i}(x,y) = \begin{cases}
	0           					& \text{if $x = y$} \\
	\rmetricC{i}(x,y)				& \text{if $x \neq y \wedge \{x,y\} \subseteq \lconsistent$} \\
	1 + \rheightCMax + \rmetricI(x,y)	& \text{if $x \neq y \wedge \{x,y\} \nsubseteq \lconsistent$}
    \end{cases} \\
\end{equation*}
If the two path-weights are equal then the dissimilarity between them is zero. If the two path-weights are not equal but both consistent then $\rmetricC{i}$ is used to compute the dissimilarity between them. Otherwise at least one is inconsistent and $\rmetricI$ is used instead. The quantity $1 + \rheightCMax$ is added to $\rmetricI(x,y)$ to later ensure that the required strictly contracting properties hold. In particular the fact that the dissimilarity between inconsistent path-weights is always greater than the dissimilarity between consistent path-weights guarantees that the dissimilarity is still strictly decreasing when the last inconsistent path-weight is flushed from the routing state.

\vspace{1em}

\subsubsection{Dissimilarity function over routing tables} 

As before, the dissimilarity function over routing tables can now be defined as:
\begin{equation*}
\tmetric{i}(\xt,\yt) = \max_{j \in \nodes} \: \rmetric{i}(\xt_j,\yt_j)
\end{equation*}
The dissimilarity between two routing tables is therefore proportional to the length of the shortest inconsistent path-weight, or, if no such path-weight exists, the desirability of the best consistent path-weight. This is the dissimilarity function $\tmetric{i}$ required by the AMCO conditions in Definition~\ref{def:amco}.

\vspace{1em}

\subsubsection{Dissimilarity function over states} Finally, as required by the AMCO conditions, the dissimilarity function over routing states, $\smetric$, is defined as follows:
\begin{equation*}
\smetric(\X,\Y) = \max_{i \in p} \: \tmetric{i}(\X_{i}, \Y_{i})
\end{equation*}

\subsection{Properties of the dissimilarity functions}

As in the proof for distance-vector protocols, we will now show that $\aFun^{ep}$ is an AMCO. It is obvious that $\tmetric{i}$ is bounded above $1 + \rheightCMax + \rheightIMax$ and that $\tmetric{i}(\xt,\yt) = 0$ iff $\xt = \yt$. As in Appendix~\ref{app:distance_vector_proof}, it is unnecessary but possible to show that $\rmetric{i}$/$\tmetric{i}$/$\smetric$ form an ultrametric space over $\carrier$/$\tcarrier$/$\mcarrier$.
 
 However, unlike Appendix~\ref{app:distance_vector_proof}, the iteration $\aFun^{ep}$ is not strictly contracting with respect to $\smetric$. Therefore it is is necessary to prove that $\aFun^{ep}$ is both strictly contracting on orbits and strictly contracting on fixed points separately. As before, we first prove a couple of useful auxiliary lemmas.

\begin{lemma} \label{lem:incon_force_!=}
If any routing table entry is inconsistent, then there must have an inconsistent routing table entry in the previous state with the same destination which has changed since the last state:
\begin{align*}
& \forall \X, i, j : \aFun^{ep}(\X)_{ij} \notin \consistent{\A} \\
& \qquad \Rightarrow \exists k : \X_{kj} \notin \consistent{\A} \wedge \X_{kj} \neq \aFun^{ep}(\X)_{kj}
\end{align*}
\end{lemma}
\begin{proof}
As the path-weight $\sigX_{ij}$ is inconsistent, by Lemma~\ref{lem:F-consistent-closure} it must be an extension of some inconsistent path-weight in $\X$. Therefore there exists a router $l$ such that $\sigX_{ij} = \A^{ep}_{il}(\X_{lj})$ and $\X_{lj}$ is inconsistent.

If $\X_{lj} \neq \sigX_{lj}$ then $l$ is the required router. Otherwise if $\X_{lj} = \sigX_{lj}$ then $\sigX_{lj}$ is inconsistent. Therefore the entire argument can be repeated with $\sigX_{lj}$. However as $\sigX_{ij} = \A^{ep}_{ij}(\X_{lj}) = \A^{ep}_{ij}(\sigX_{lj})$ the path of $\sigX_{lj}$ must be strictly shorter than the path of $\sigX_{ij}$. The length of the path cannot decrease indefinitely and therefore this argument must eventually terminate.  
\end{proof}

The next lemma corresponds to Lemma~\ref{lem:strictly-contracting-sub} in the proof of convergence for distance-vector protocols.

\begin{lemma}
\label{lem:strictly-contracting-on-orbits-sub}
For each destination, the maximum dissimilarity between all entries of three consecutive routing states must strictly decrease:
\begin{align*}
& \forall \X, i, j, v > 0: 
(\forall k : \rmetric{k}(\X_{kj},\aFun^{ep}(\X)_{kj}) \leq v) \\
& \qquad \Rightarrow \rmetric{i}(\aFun^{ep}(\X)_{ij},(\aFun^{ep})^2(\X)_{ij}) < v
\end{align*}
\end{lemma}
\begin{proof}

Assume
\begin{equation}
\label{def:r<=v}
\forall k : \rmetric{k}(\X_{kj},\aFun^{ep}(\X)_{kj}) \leq v
\end{equation}

\begin{case}{1}{$\sigX_{ij} = \sigsqX_{ij}$}
Then the inequality is immediate as:
\begin{align*}
\rmetric{i}(\sigX_{ij},\sigsqX_{ij}) 
&= 0 & \text {(by Case 1 \& def $\rmetric{i}$)}  \\
&< v & \text{(by lemma ass.)}
\end{align*}
\end{case}

\begin{case}{2}{$\sigX_{ij} \neq \sigsqX_{ij}$ and $\sigX_{ij}$ and $\sigsqX_{ij}$ are both consistent.}
\begin{case}{2.1}{$\X$ is consistent}
If $\X$ is consistent then $\aFun^{ep}(\X)$ must also be consistent and so all path-weights involved are consistent. As the topology is cycle-free then the required inequality is therefore immediate from Lemma~\ref{lem:strictly-contracting-sub} in Appendix~\ref{app:distance_vector_proof}.
\end{case}

\vspace{1em}

\begin{case}{2.2}{$\X$ is inconsistent}
If $\X$ is inconsistent then there must exist routers $k$ and $l$ such that 
\begin{align}
\label{eq:Xkl-ineq}
\X_{kl} \notin \lconsistent \\
\label{eq:Xkl-inconsistent}
\X_{kl} \neq \sigX_{kl}
\end{align}
To see why, consider whether $\sigX$ is consistent. If $\sigX$ is consistent then any inconsistent entry in $\X$ will suffice. If $\sigX$ is inconsistent, then there exists an inconsistent entry $\sigX_{ml}$ in which case the required $k$ may be obtained from Lemma~\ref{lem:incon_force_!=}. 

The required inequality then follows as:
\begin{align*}
& \rmetric{i}(\sigX_{ij},\sigsqX_{ij}) \\
&= \rmetricC{i}(\sigX_{ij},\sigsqX_{ij}) & \text {(by Case 2 \& def $\rmetric{i}$)}  \\
&< \rheightCMax + \rmetricI{k}(\X_{kl},\sigX_{kl}) & \text{(by $\rmetricC{k} \leq \rheightCMax$)} \\
&= \rmetric{k}(\X_{kl},\sigX_{kl}) & \text{(by def $\rmetric{i}$, (\ref{eq:Xkl-ineq}) \& (\ref{eq:Xkl-inconsistent}))} \\
&\leq v & \text{(by (\ref{def:r<=v}))}
\end{align*}
\end{case}

\vspace{1em}
\end{case}

\begin{case}{3}{$\sigX_{ij} \neq \sigsqX_{ij}$ and $\sigX_{ij}$ or $\sigsqX_{ij}$ is inconsistent.}
As either $\sigX$ or $\sigsqX$ is inconsistent then by Lemma~\ref{lem:extension-consistent-closure} we have that $\X$ must be inconsistent as well. Let $\X_{kl}$ be the inconsistent path-weight with the shortest path in $\X$. As all inconsistent path-weights in $\sigX$ and $\sigsqX$ are an extension of inconsistent path-weights in $\X$, the path of $\X_{kl}$ must be shorter than the paths of all inconsistent path-weights in $\sigX$ and $\sigsqX$ and so
\begin{align}
\label{eq:sXssX<X}
\max(\rheightI(\sigX_{ij}), \rheightI(\sigsqX_{ij})) < \rheightI(\X_{kl}) \\
\label{eq:32neq}
\X_{kl} \neq \sigX_{kl} 
\end{align}
The inequality then follows as:
\begin{align*}
& \rmetric{i}(\sigX_{ij},\sigsqX_{ij}) \\
&= \rheightCMax + \rmetricI(\sigX_{ij},\sigsqX_{ij}) \:\: \text {(by Case 3 \& def $\rmetric{i}$)} \\
&< \rheightCMax + 1 + \rheightI(\X_{kl}) \qquad\qquad\qquad\qquad\qquad\qquad\: \text {(by (\ref{eq:sXssX<X}))} \\
&\leq \rheightCMax + 1 + \max(\rheightI(\X_{kl}),\rheightI(\sigX_{kl})) \quad \text {(by def $\max$)} \\
&= \rmetric{k}(\X_{kl},\sigX_{kl}) \qquad\qquad\quad\:\: \text {(by def $\rmetric{k}$, (\ref{eq:Xkl-ineq}) \& (\ref{eq:32neq}))} \\
&\leq v \qquad\qquad\qquad\qquad\qquad\qquad\qquad\qquad  \qquad\quad\:\:\: \text {(by (\ref{def:r<=v}))}
\end{align*}
\end{case}
\noindent Hence the required inequality holds in all cases.
\end{proof}

\noindent The required strictly contracting on orbits property can now be proved using the above lemma.

\begin{lemma}
\label{lem:paths_str_contr_on_orbits}
$\aFun^{ep}$ is strictly contracting on orbits over $\smetric$.
\end{lemma}
\begin{proof}
Consider an arbitrary state $\X \in \accSet$. Then it must be shown that if $\X \neq \aFun^{ep}(\X)$ then:
\begin{equation*}
D^p(\aFun^{ep}(\X),(\aFun^{ep})^2(\X)) < D^p(\X, \aFun^{ep}(\X))
\end{equation*}
As
\begin{align*}
\smetric(\aFun^{ep}(\X),(\aFun^{ep})^2(\X)) = \qquad \qquad\qquad \qquad\\  \max_{i \in p, j \in \nodes} \: \rmetric{i}(\aFun^{ep}(\X)_{ij},(\aFun^{ep})^2(\X)_{ij})
\end{align*}
it suffices to show for all routers $i \in p$ and $j \in \nodes$ that
\begin{equation*}
\rmetric{i}(\aFun^{ep}(\X)_{ij},(\aFun^{ep})^2(\X)_{ij}) < \smetric(\X,\aFun^{ep}(\X))
\end{equation*}
This can be proved by applying Lemma~\ref{lem:strictly-contracting-on-orbits-sub} where $v = \smetric(\X,\aFun^{ep}(\X))$. However to apply the lemma it must be verified that $0 < D^p(\X,\aFun^{ep}(\X))$ and that for all routers $k$ then $\rmetric{k}(\X_{kj},\aFun^{ep}(\X)_{kj}) \leq \smetric(\X,\aFun^{ep}(\X))$.

The former holds as $\X$ is accordant and $\X \neq \aFun^{ep}(\X)$ and so there must exist entries for participating routers in $\X$ and $\aFun^{ep}(\X)$ whose dissimilarity is non-zero. To show that the latter holds consider whether or not router~$k$ is participating. If $k \in p$ then the required inequality holds simply from the definition of $\smetric$. If $k \notin p$ then, as $\X$ and $\aFun^{ep}(\X)$ are accordant, $\X_{kj} = \invalid = \aFun^{ep}(\X)_{kj}$ and therefore the inequality holds trivially as $\rmetric{k}(\X_{kj},\aFun^{ep}(\X)_{kj}) = 0$.
\end{proof}

\begin{lemma} \label{lem:paths_str_contr_on_fixed_point}
$\aFun^{ep}$ is strictly contracting on fixed points over~$\smetric$.
\end{lemma}
\begin{proof}
The proof has the same structure as Lemmas~\ref{lem:strictly-contracting-on-orbits-sub}~\&~\ref{lem:paths_str_contr_on_orbits}. The only major difference is that some of the cases in the corresponding version of Lemma~\ref{lem:strictly-contracting-on-orbits-sub} are redundant as the fixed point $\X^*$ is guaranteed to be consistent. This must be the case, as if $\X^*$ was inconsistent then applying $\aFun^{ep}$ would increase the length of the shortest inconsistent path. Interested readers can find the remaining details in the Agda formalisation~\cite{agda-routing}.
\end{proof}

\noindent Finally all the pieces to show that the set of functions $\aFun^{ep}$ are an AMCO over $\freeConfigs$ have now been assembled. In particular:
\begin{enumerate}[label=(DU\arabic*),leftmargin=1.5cm]

\item $\forall \xt, \yt : \tmetric{i}(\xt,\yt) = 0 \Leftrightarrow \xt = \yt$

Immediate from the definition of $\tmetric{i}$ and $\rmetric{i}$.

\item $\tmetric{i}$ is bounded 

The maximum dissimilarity $\rmetric{i}$, and therefore $\tmetric{i}$, can return is $1 + \rheightCMax + \rheightIMax$.

\item $\aFun^{ep}$ is strictly contracting on orbits over $\smetric$

Proved in Lemma~\ref{lem:paths_str_contr_on_orbits}.

\item $F^{ep}$ is strictly contracting on fixed points over $\smetric$

Proved in Lemma~\ref{lem:paths_str_contr_on_fixed_point}.

\item $\forall \X : \aFun^{ep}(\X) \in \accSet$.

Immediate by the definition of $\aFun^{ep}$.
\end{enumerate}
As $\aFun^{ep}$ is an AMCO over $\freeConfigs$ then by Theorem~\ref{thm:amco-convergent} we have that $\delta$ is convergent over $\genericPathAlgebra$ whenever the network topology is free.
\end{document}